\documentclass[reprint,amsfonts, amssymb, amsmath,  showkeys,pra, superscriptaddress, twocolumn,longbibliography,nofootinbib]{revtex4-2}

\usepackage{xcolor}

\usepackage{float}
\makeatletter
\let\newfloat\newfloat@ltx
\makeatother
\usepackage[english]{babel}
\usepackage[utf8]{inputenc}
\usepackage{graphics}
\usepackage{selinput}
\usepackage[normalem]{ulem}
\usepackage[shortlabels]{enumitem}

\usepackage{braket}
\usepackage{amsthm}
\usepackage{mathtools}
\usepackage{physics}
\usepackage{graphicx}
\usepackage[left=16mm,right=16mm,top=35mm,columnsep=15pt]{geometry} 
\usepackage{adjustbox}
\usepackage{placeins}
\usepackage[T1]{fontenc}
\usepackage{lipsum}
\usepackage{csquotes}
\usepackage{bm}

\usepackage[linesnumbered,ruled,vlined]{algorithm2e}
\SetKwInput{kwInit}{Init}

\def\HC{\mathcal{H}}

\def\LC{\mathcal{L}}

\def\ad{^{\dagger}}

\def\a{\alpha}

\newcommand{\fsnull}[1]{}
\newcommand{\old}[1]{}

\usepackage[makeroom]{cancel}
\usepackage[toc,page]{appendix}
\usepackage[colorlinks=true,citecolor=blue,linkcolor=magenta]{hyperref}

\usepackage{tikz}
\tikzset{every picture/.style=remember picture}

\usepackage[utf8]{inputenc}
\usepackage{graphicx}
\usepackage{xcolor}
\usepackage{amsmath}
\usepackage{amsthm}
\usepackage{bm}
\usepackage{bbm}
\usepackage{comment}
\usepackage{appendix}
\usepackage{mathdots}
\usepackage{lipsum}
\usepackage{verbatim}
\usepackage{enumerate}
\usepackage{nccmath}
\usepackage{amsfonts} 

\newcommand{\dya}[1]{\ket{#1}\!\bra{#1}}

\newcommand{\BC}{\mathcal{B}}
\newcommand{\CC}{\mathcal{C}}

\newcommand{\NC}{\mathcal{N}}

\newcommand{\PC}{\mathcal{P}}

\newcommand{\SC}{\mathcal{S}}

\renewcommand{\geq}{\geqslant}
\renewcommand{\leq}{\leqslant}

\newcommand{\spn}{{\rm span}}

\renewcommand{\vec}[1]{\boldsymbol{#1}}  

\newcommand*{\id}{\openone}

\newcommand{\bs}{\textsf{BS}}

\def\be{\begin{equation}}
\def\ee{\end{equation}}
\def\bs{\begin{split}}
\def\e{\end{split}}
\def\ba{\begin{eqnarray}}
\def\bea{\begin{eqnarray}}

\def\tea{\end{eqnarray}}
\def\ea{\end{eqnarray}}
\def\eea{\end{eqnarray}}

\def\a{\alpha}

\def\a{\alpha}

\def\g{\mathfrak{g}}

\def\a{\alpha}

\newcommand\mbb[1]{\mathbb{#1}}

\newcommand\Z{\text{Z}}

\newtheorem{theorem}{Theorem}

\newtheorem{definition}{Definition}

\usepackage{amssymb}
\usepackage{dsfont}

\def\be{\begin{equation}}
\def\te{\end{equation}}
\def\ee{\end{equation}}
\def\ba{\begin{eqnarray}}
\def\bea{\begin{eqnarray}}

\def\tea{\end{eqnarray}}
\def\ea{\end{eqnarray}}
\def\eea{\end{eqnarray}}

\begin{document}

\title{Characterizing quantum resourcefulness via group-Fourier decompositions}

\author{Pablo Bermejo}
\thanks{The first two authors contributed equally to this work.}
\affiliation{Information Sciences, Los Alamos National Laboratory, Los Alamos, NM 87545, USA}
\affiliation{Donostia International Physics Center, Paseo Manuel de Lardizabal 4, E-20018 San Sebasti\'an, Spain}
\affiliation{Department of Applied Physics, University of the Basque
Country (UPV/EHU), 20018 San Sebastián, Spain}

\author{Paolo Braccia}
\thanks{The first two authors contributed equally to this work.}
\affiliation{Theoretical Division, Los Alamos National Laboratory, Los Alamos, NM 87545, USA}

\author{Antonio Anna Mele}
\affiliation{Dahlem Center for Complex Quantum Systems, Freie Universität Berlin, 14195 Berlin, Germany}
\affiliation{Theoretical Division, Los Alamos National Laboratory, Los Alamos, NM 87545, USA}

\author{N.~L.~Diaz}
\affiliation{Information Sciences, Los Alamos National Laboratory, Los Alamos, New Mexico 87545, USA}
\affiliation{Center for Non-Linear Studies, Los Alamos National Laboratory, 87545 NM, USA}
\affiliation{Departamento de F\'isica-IFLP/CONICET, Universidad Nacional de La Plata, C.C. 67, La Plata 1900, Argentina}

\author{Andrew E. Deneris}
\affiliation{Information Sciences, Los Alamos National Laboratory, Los Alamos, NM 87545, USA}

\author{Mart\'{i}n Larocca}
\affiliation{Theoretical Division, Los Alamos National Laboratory, Los Alamos, NM 87545, USA}

\author{M. Cerezo}
\thanks{cerezo@lanl.gov}
\affiliation{Information Sciences, Los Alamos National Laboratory, Los Alamos, NM 87545, USA}
\affiliation{Quantum Science Center, Oak Ridge, TN 37931, USA}

\begin{abstract}
In this work we present a general framework for studying the resourcefulness in pure states for quantum resource theories (QRTs) whose free operations arise from the unitary representation of a group. We argue that the group Fourier decompositions (GFDs) of a state, i.e., its projection onto the irreducible representations (irreps) of the Hilbert space, operator space, and tensor products thereof, constitute fingerprints of resourcefulness and complexity. By focusing on the norm of the irrep projections, dubbed GFD purities, we find that low-resource states live in the small dimensional irreps of operator space, whereas high-resource states have support in more, and higher dimensional ones. Such behavior not only resembles that appearing in classical harmonic analysis, but is also universal across the QRTs of entanglement,  fermionic Gaussianity, spin coherence, and Clifford stabilizerness. To finish, we show that GFD purities carry operational meaning as they lead to resourcefulness witnesses as well as to notions of state compressibility.
\end{abstract}

\maketitle

\section{Introduction}

Harmonic analysis refers to the framework for studying objects on vector spaces admitting a group action by decomposing them into components indexed by the irreducible representations of such group—a process known as group Fourier decomposition (GFD).~\cite{katznelson2004introduction,mackey1980harmonic}. As part of signal processing, GFD enables tasks such as frequency filtering, noise reduction,  solving differential equations, and also plays a key role in formulating spatial rigid motions for robotics and motion planning~\cite{chirikjian2000engineering}. More recently, GFD has seen widespread use in the context of geometrical classical machine learning where one characterizes the data in terms of its inherent symmetries, thus improving the learning model's efficiency and generalization~\cite{bronstein2021geometric,cohen2018spherical}. 

 As a unifying framework, GFD allows the extrapolation of tools and lessons from one theory onto another, showcasing that certain behaviors are universal across different applications. For instance, when decomposing a real-valued function via a discrete Fourier transform into the one-dimensional irreps of $\mathbb{Z}_n$, one finds that more ``complex functions'' have support into more, and higher-momentum, irreps. A similar behavior appears when studying spherically symmetric images through their decomposition into spherical Harmonics, the irreps of $\mathbb{SO}(3)$. For instance, when taking a three-dimensional image of the sun and its magnetic field~\cite{altschuler1975tabulation,whaler1981spherical,knaack2005spherical}, one finds that the component onto smaller dimensional irreps captures simple coarse global features of the data (i.e., the smaller irrep only describes a solid sphere), whereas finer more localized details of complex phenomenon like sun spots necessitates higher‑dimensional irreps to adequately describe the data (see Fig.~\ref{fig:fig-1}). Crucially, such cross-cutting behavior connecting complexity and irrep dimension is also of practical interest as the truncation of GFD higher-dimensional or higher-moment irreps components is central for data compression techniques~\cite{marks2009handbook}.   

Given the high versatility of GFDs, its basic tools have been used within the context of quantum resource theories (QRTs)~\cite{chitambar2019quantum}. For instance, the computation of projections of a state onto the elements of a Lie group (i.e., characteristic functions)~\cite{gu1985group,korbicz2006group,korbicz2008entanglement,marvian2013theory} or those of a Lie algebra or preferred subset of observables~\cite{barnum2003generalizations,barnum2004subsystem,klyachko2002coherent,delbourgo1977maximum, meyer2002global,brennen2003observable,beckey2021computable,schatzki2022hierarchy,balachandran2013entanglement,harshman2011observables,zanardi2001virtual,zanardi2004quantum,nha2006entanglement,alicki2009quantum,viola2010entanglement,derkacz2011entanglement,gigena2015entanglement,benatti2016entanglement,regula2017convex,sindici2018simple,gigena2021many,guaita2021generalization,ahmad2022quantum} to compute group-invariant quantities~\cite{grassl1998computing,barnum2001monotones,miyake2003classification,leifer2004measuring,mandilara2006quantum,klyachko2007dynamical,oszmaniec2013universal,bravyi2019simulation,larocca2022group,meyer2023exploiting,nguyen2022atheory,skolik2022equivariant}, provides valuable insights regarding quantum phase transitions~\cite{werlang2011spotlighting,viscondi2009generalized}, quantum chaos~\cite{weinstein2006generalized}, equivalence classes between states~\cite{nielsen1999conditions,barnum2001monotones,kus2001geometry,verstraete2003normal,chitambar2019quantum,miyake2003classification,garibaldi2022generic,park2024universal}, as well as the usefulness of quantum states for quantum computation~\cite{gigena2020one,datta2005entanglement,gross2009most,diaz2023showcasing}. Despite these important advancements there is a dearth of results which take a holistic view of QRTs to understand whether the universal behaviors arising in classical GFD analysis still persist in the quantum realm.

In this work we present a general framework for GFDs within QRTs where the set of resource-free operations is described by the unitary representation of a group (Lie and discrete). In particular, we argue that the  projection of a pure quantum state onto the irreps appearing in the group-Fourier decomposition of the Hilbert space, the operator space, or in their tensor products, enable new dimensions to characterize the state's resourcefulness. Central to our framework is the concept of irrep purities, which are defined as the norms of the irrep projections. Indeed, by studying the irrep purities across the QRTs of entanglement (bipartite and multipartite), fermionic Gaussianity, spin coherence, and Clifford stabilizerness we show that resource-free states live in less, but most importantly, lower dimensional irreps. On the other hand, the more resourceful a state is, the more it will have support in more, and higher dimensional, irreps (see Fig.~\ref{fig:fig-1}), therefore reflecting a behavior similar to that observed in classical harmonic analysis.  In addition, we also show that certain purities on the smaller (non-trivial) irreps are resourcefulness witnesses, being maximized only for free states. Then, by borrowing inspiration from classical harmonic analysis we study whether the truncation of a state's component in the higher dimensional irreps, or alternatively, its projection onto the smaller dimensional ones, is meaningful. Here, we show that resource-free states on Lie group-based QRTs can be compressed to, and thus fully described by, the information encoded into the smaller dimensional irreps that constitute a resourcefulness witness, therefore being the ``simpler'' of all states. To finish, we discuss how the GFD purities  can capture features inequivalent to those described by the state's extent --its minimal decomposition into free states.

\begin{figure}[t]
    \centering
\includegraphics[width=.9\linewidth]{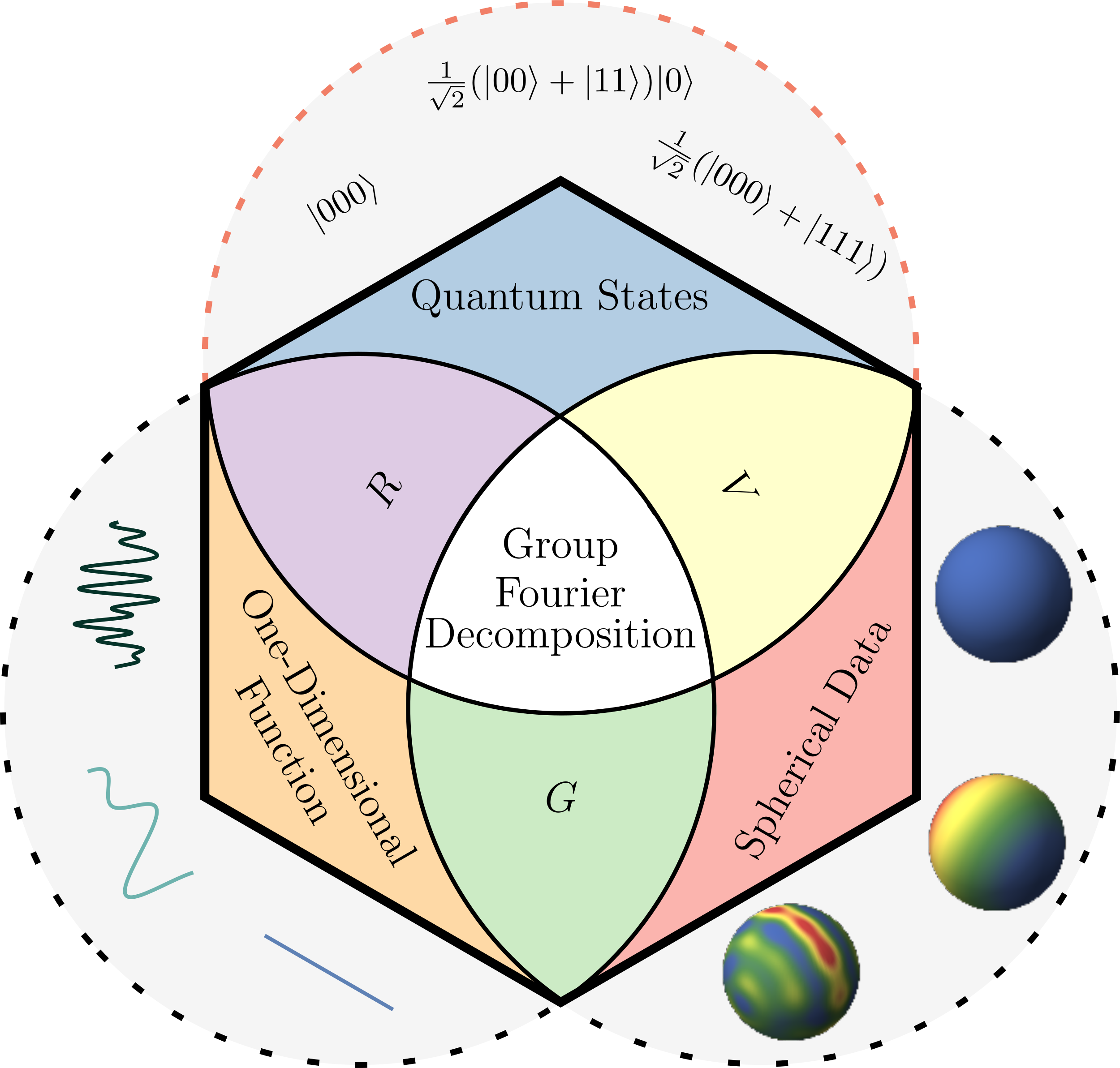}
    \caption{\textbf{GFD framework.} The GFD is based on the irrep structure of a vector space $V$ induced by the action of a representation $R$ of a group $G$. Then, given some mathematical object in the vector space (one dimensional function, spherical images, quantum states), one decomposes them into the group's irreps. We uncover the existence of universal behaviors  such as ``less complex'' vectors living in --and compressible to-- the  lower dimensional irreps, whereas more complex vectors live in more and higher dimensional ones.  }
    \label{fig:fig-1}
\end{figure}

\section{Results}

\subsection{Framework}

Let $\HC\cong\mbb{C}^d$ be a finite dimension Hilbert space and denote $\LC(\HC)\coloneqq \HC\otimes \HC^*$ the associated space of linear operators, and $\mathbb{U}(\HC)$ the group of unitary operators on $\HC$. While \textit{a priori} there are no states in $\HC$, nor unitaries in $\mathbb{U}(\HC)$, that should play a privileged role, QRTs~\cite{chitambar2019quantum} provide a framework to study the fact that under some set of rules, certain states and unitaries are more valuable than others (see also our companion work where we show that QRTs can be defined in terms of preserved algebraic structures~\cite{nahuel2025quantum}). As such, in a QRT one defines a set of \textit{free} operations and a set of \textit{free} states which are assumed to not generate, or contain, any resource. Formally, these ideas can be captured in the following definition
\begin{definition}[Free operations and states]\label{def:free} 
We define the set of free operations as a unitary representation $R$ of a compact group $G$, i.e., $R:G\rightarrow \mathbb{GL}(\HC)$, with  $\mathbb{GL}(\HC)$ the general linear group over $\HC$. Then, the set of pure free states is obtained as the orbit of some reference state $\ket{\psi_{\rm ref }}$ under $R(G)$, i.e., all free states take the form $R(g)\ket{\psi_{\rm ref }}$ for some $g\in G$. 
\end{definition}
Importantly, once we have defined the free operations as a representation of a group, we can use Definition~\ref{def:free} can  to define an action over $t$ copies of $\HC$, and their ensuing operator space, through the $t$-th fold tensor product, and through the associated adjoint action, respectively. 

From here, let us take a small step back and consider a general vector space $V$ (e.g., $V=\HC,\LC(\HC),\HC^{\otimes 2},\LC(\HC^{\otimes 2}),\ldots$) holding a representation $R$ of $G$, the group of free operations.  Maschke's theorem~\cite{serre1977linear,fulton1991representation} states that for any compact $G$, there is a basis in which $R(g\in G)$ and $V$ are maximally block-diagonalized
\begin{equation}
\label{eq:block-diagonalization}
R(g) \simeq \bigoplus_{\alpha} I_{m_\alpha} \otimes r^{\alpha}_G(g)\,, \quad V \simeq \bigoplus_\alpha \mathbb{C}^{m_\alpha} \otimes V^{\alpha}_G\,.
\end{equation}
Above, $r^\alpha_G:G\rightarrow \mathbb{U}(V^\alpha_G)$ is an irreducible representation (an ``irrep'') of $G$ labeled by $\alpha$, and $V^\alpha_G$ the corresponding irreducible $G$-module.  The number of times each irrep $r^\alpha_G$ appears in $R$ (or equivalently $V^\alpha_G$ in $V$) is the multiplicity $m_\alpha$. We will denote the $j$-th copy of the $\alpha$-th irrep as $V_{\alpha,j}\subseteq V$.

The previous allows us to define two important concepts.  
\begin{definition}[Group-Fourier decomposition (GFD)]\label{def:irrep-decomp} 
    Given a vector $v\in V$,  we define its GFD as its projections onto the irreps of Eq.~\eqref{eq:block-diagonalization}. In particular, given an  orthonormal basis $\{w_{\alpha,j}^{(i)}\}_{i=1}^{\dim(V_{\alpha,j})}$ of some irrep, the projection of $v\in V$ onto $V_{\alpha,j}$ is the vector 
\begin{equation}
    v_{j,\alpha}=\sum_{i=1}^{\dim(V_{\alpha,j})}\langle w_{\alpha,j}^{(i)},v\rangle w_{\alpha,j}^{(i)}\subseteq V\,,
\end{equation}
with $\langle\cdot,\cdot\rangle:V\times V\rightarrow R$ the inner product  in $V$.
\end{definition}
\noindent Importantly, the GFD components are orthogonal
\begin{equation}\label{eq:orth}
\langle v_{j,\alpha},v_{j',\alpha'}\rangle=|v_{j,\alpha}|^2\delta_{j,j'}\delta_{\alpha,\alpha'}\,,
\end{equation}
and  they uniquely determine the vector $v$, since 
\begin{equation}\label{eq:comp}
v=  \sum_{\alpha}\sum_{j=1}^{m_\alpha}v_{j,\alpha}\,.
\end{equation}
Next, we define the norm of each GFD component, which quantifies how much support $v$ has within an irrep, as the GFD purity.
\begin{definition}[GFD purity]\label{def:irrep-purity}
   Given a vector $v\in V$, the norm of the vectors $v_{j,\alpha}$ are defined as GFD purities, and are denoted as
\begin{align}\label{eq:irrep-purity}
    \PC_{j,\alpha}(v)=\langle v_{j,\alpha},v_{j,\alpha} \rangle= \sum_{i=1}^{\dim(V_{\alpha,j})}\left|\left\langle w_{\alpha,j}^{(i)},v\right\rangle\right|^2\,.
\end{align}
\end{definition}

As shown below, the GFD purities will play a central role in our framework, as they  constitute a veritable fingerprint for the resourcefulness of a state. Therefore, we find it convenient to first present some important properties satisfied by these quantities.  Firstly, 
combining Eqs.~\eqref{eq:comp} and~\eqref{eq:irrep-purity} leads to
\begin{equation}\label{eq:normalization}
    \sum_{\alpha}\sum_{j=1}^{m_\alpha} \PC_{\alpha,j}(v)=\langle v,v\rangle\,,
\end{equation}
which simply states the fact that since the vector is distributed across the irreps of $V$, then their sums must add-up to the vector's norm. For normalized vectors, i.e., $\langle v,v\rangle=1$, this means that the GFD purities form a probability distribution. In addition, Eq.~\eqref{eq:normalization} indicates that  the more a vector lives in one irrep, the less it can belong to others. Second, one can verify that Eq.~\eqref{eq:block-diagonalization} implies that the GFD purities are group invariants, as for any $g\in G$ 
\begin{equation}\label{eq:invariance-purity}
    \PC_{\alpha,j}(R(g)\cdot v)=\PC_{\alpha,j}(v)\,.
\end{equation}
The previous ensures that the GFD probability distribution is a sane framework to study resourcefulness, as all the vectors in the same orbit of (i.e., connected by) free operations will have the same GFD purity profiles (see Definition~\ref{def:free}). Moreover, as discussed in the Methods (and in our companion manuscript~\cite{mele2025clifford}), the number of different irrep purities can be readily bounded by noting that these quantities arise as projectors of $v^{\otimes 2}$ onto the trivial components of the irrep decomposition of $V^{\otimes 2}$. As such, the more trivial irreps appear in $V^{\otimes 2}$, and the more $v^{\otimes 2}$ has non-trivial and non-equivalent component therein, the richer the ensuing GFD purity structure.

Here we also find it important to note that up to this point we have not made any assumption regarding what $V$ is. As such, since $G$ admits representations in the vector spaces $\HC$, $ \LC(\HC)$, $\HC^{\otimes 2}$, $\LC(\HC^{\otimes 2})$, $\ldots$; we can employ the GFD theory to hierarchically study the purities of $\ket{\psi}$, $\rho=\dya{\psi}$, $\ket{\psi}^{\otimes 2}$, $\rho^{\otimes 2}$, $\ldots$, leading to polynomials of increasing order in the entries of the state and its dual.  

\subsection{Operational meaning of the GFD purities}

Here we discuss the operational meaning of the GFD formalism by connecting it to resourcefulness witnesses, i.e. quantities maximized only for free states, and to notions of state compressibility. 

First, consider the case where the free operations arise from a Lie group $G$, with associated semi-simple Lie algebra $\mathfrak{g}$ (which we assume for simplicity to contain a single non-trivial component in its reductive decomposition). Here, free states correspond to generalized coherent states, i.e., the highest weights of $\mathfrak{g}$~\cite{barnum2003generalizations,barnum2004subsystem,perelomov1977generalized,gilmore1974properties,zhang1990coherent}. By taking  $V=\LC(\HC)$, and noting that a Lie algebra is closed under the adjoint action of its Lie group, then the complexification of the Lie algebra's representation $\mathbb{C}R(\mathfrak{g})$ must be one of the irreps of $V=\LC(\HC)$. As shown in~\cite{barnum2003generalizations,barnum2004subsystem}, the following theorem holds 
\begin{theorem}[Purity in the algebra is a resourcefulness witness, Theorem in~\cite{barnum2004subsystem}]\label{theo:algebra-purity}
    Let $\PC_{\mathbb{C}R(\mathfrak{g})}(\rho)$ be the purity of a state $\rho=\dya{\psi}$ onto the irrep $\mathbb{C}R(\mathfrak{g})$ of $\LC(\HC)$. Then, $\PC_{\mathbb{C}R(\mathfrak{g})}(\rho)$ is maximized iff $\ket{\psi}$ is a free state.
\end{theorem}
Theorem~\ref{theo:algebra-purity} indicates that one of the irreps in operator space\footnote{As shown below, for all examples considered, $\mathbb{C}R(\mathfrak{g})$ (the adjoint representation),  will be one of the smallest irreps in $\LC(\HC)$.} constitutes a resourcefulness witness. Notably, in our companion paper \cite{mele2025clifford} we show that such connection between purities and witnesses can be extended beyond Lie group-based QRTs.  Therein  we analyze the setting of Clifford stabilizerness where the free operations are Clifford unitaries, and thus form a discrete group. In particular, we show that the stabilizer entropy~\cite{leone2022stabilizer}, a well-known measure of non-stabilizerness,  corresponds to a GFD purity in $\LC(\HC^{\otimes 2})$ onto the smallest irrep whose purity is non-trivial.

Then, let us borrow inspiration from classical harmonic analysis, and study the properties of a vector $v$ which has been projected  onto all irreps whose dimension is smaller than, or equal to, a fixed threshold value $\Lambda$. Such a vector is obtained by imposing a cut-off in the summation of~\eqref{eq:comp}
\begin{equation}\label{eq:comp-cutoff}
v_{\Lambda}=  \sum_{\substack{\alpha\\\dim(V_{\alpha,j})\leq \Lambda}}\sum_{j=1}^{m_\alpha}v_{j,\alpha}\,.
\end{equation}
From the previous, we can immediately see that the distance between $v$ and its truncated version $v_{\Lambda}$, as measured through their inner-product, leads to
\begin{equation}
    \PC_\Lambda(v):=\langle v, v_{\Lambda}\rangle=\sum_{\substack{\alpha\\\dim(V_{\alpha,j})\leq \Lambda}}\sum_{j=1}^{m_\alpha}\PC_{j,\alpha}(v)\,,
\end{equation}
thus providing an operational interpretation to the cumulative probabilities arising from the GFD purities. 

Then, it could occur that $v_{\Lambda}$ fails to meet certain desirable properties which $v$ possesses (e.g.,  being normalized). In this case, one could attempt to find a vector $u$ that matches the components of $v$ on all the irreps of dimension smaller than $\Lambda$. This can be achieved by solving the  Maximum-Entropy (MaxEnt) principle, or semidefinite programming, problem
\begin{align}
    \langle w_{\alpha,j}^{(i)},v\rangle=\langle w_{\alpha,j}^{(i)},u&\rangle\,,\,\,  \forall w_{\alpha,j}^{(i)}\in V_{\alpha,j} \text{ s.t. } \dim(V_{\alpha,j})\leq \Lambda\nonumber\\
    &\text{subject to } \{F_i(u)\}_i\,,\label{eq:max-ent}
\end{align}
where $F_i(u)$ are a set of conditions that $u$ must satisfy (e.g.,  being normalized). In particular, given some solution $u$ to Eq.~\eqref{eq:max-ent}, we can readily find via the Cauchy-Schwarz inequality that $\langle v, u\rangle\geq 2 \PC_\Lambda(v)-1$. Hence, the less resourcefulness $v$ possess, i.e., the  more $v$ has support in smaller dimensional irreps, the less quantities one has to match in Eq.~\eqref{eq:max-ent}, to guarantee that $v$ and $u$ are close. Notably, the previous bound is extremely loose for the case of Lie group-based QRTs. Here, if we set a cut-off $\Lambda$ such that $\Lambda=\dim(\mathfrak{g})$, we can  prove the following theorem 
\begin{theorem}[MaxEnt and free states]\label{prop:compress}
Let  $\rho=\dya{\psi}$ be a pure free state. If we find a pure state  $\sigma$ which solves the MaxEnt problem of Eq.~\eqref{eq:max-ent} for $\Lambda=\dim(\mathfrak{g})$, then $\sigma=\rho$.  
\end{theorem}
\noindent    
We refer the reader to the Supplemental Information (SI) for a proof of Theorem~\ref{prop:compress}. This result implies that there exist QRTs for which  solving the MaxEnt problem in the irrep arising from the algebra guarantees that $u=v$. That is, free states are compressible to, and fully determined by, the  quantities in the smallest non-trivial irrep.  Theorem~\ref{prop:compress} also implies that when free states have support in higher-dimensional irreps, which will necessarily happen as it is generally impossible to fit all the purity in the small-dimensional irreps, the information contained therein is redundant.

\subsection{GFD Purities for different QRTs}

Here we showcase the GFD purity formalism for several QRTs based on Lie and discrete groups. In all cases, we refer the reader to the Methods for explicit formulae for the purities, and to the SI for the associated derivations.

\subsubsection{Bipartite entanglement in two-qubit states}
Consider the QRT of standard bipartite entanglement in a two-qubit system. Here, $\HC=(\mathbb{C}^2)^{\otimes 2}$, and  the free operations are given by the standard representation $R$ of $G=\mathbb{SU}(2)\times \mathbb{SU}(2)$. Since $R$ is irreducibly represented in $\HC$, the whole Hilbert space is a single irrep, meaning that the purity of any pair of states is exactly the same ($\PC_{\HC}( \ket{\psi})=\PC_{\HC}( \ket{\phi})=1$ for all $\ket{\psi},\ket{\phi}\in \HC$). Hence, one needs to go up the hierarchy and study the GFD purities of $\rho=\dya{\psi}$ in $V=\LC(\HC)$ under the adjoint action of $R$. The operator space decomposes into four multiplicity-free irreps as
\begin{equation}\label{eq:irreps-L-2q}
    \LC(\HC)\cong\LC_{(0,0)}\oplus \LC_{(0,1)}\oplus \LC_{(1,0)}\oplus \LC_{(1,1)}\,,
\end{equation}
where a basis for $\LC_{\alpha}$, with $\alpha\in\{0,1\}^{\otimes 2}$, is given by all operators of the form $P_1^{\alpha_1}\otimes P_2^{\alpha_2}$. Specifically, $\LC_{(0,0)}={\rm span}_{\mathbb{C}}\{\id_1\otimes \id_2\}$, $ \LC_{(1,0)}={\rm span}_{\mathbb{C}}\{X_1,Y_1,Z_1\}$ (and similarly for $\LC_{(0,1)}$) and $\LC_{(1,1)}={\rm span}_{\mathbb{C}}\{X_1X_2,X_1Y_2,\ldots,Z_1Z_2\}$. Here, $X_\mu$, $Y_\mu$ and $Z_\mu$ denote the Pauli matrices on the $\mu$-th qubit. We also find that $\dim(\LC_{\alpha})=3^{w(\alpha)}$, where $w(\alpha)$ is the Hamming weight of $\alpha$, and therefore $\LC_{(1,1)}$ is the largest irrep.

\begin{figure}[t]
    \centering
\includegraphics[width=.9\linewidth]{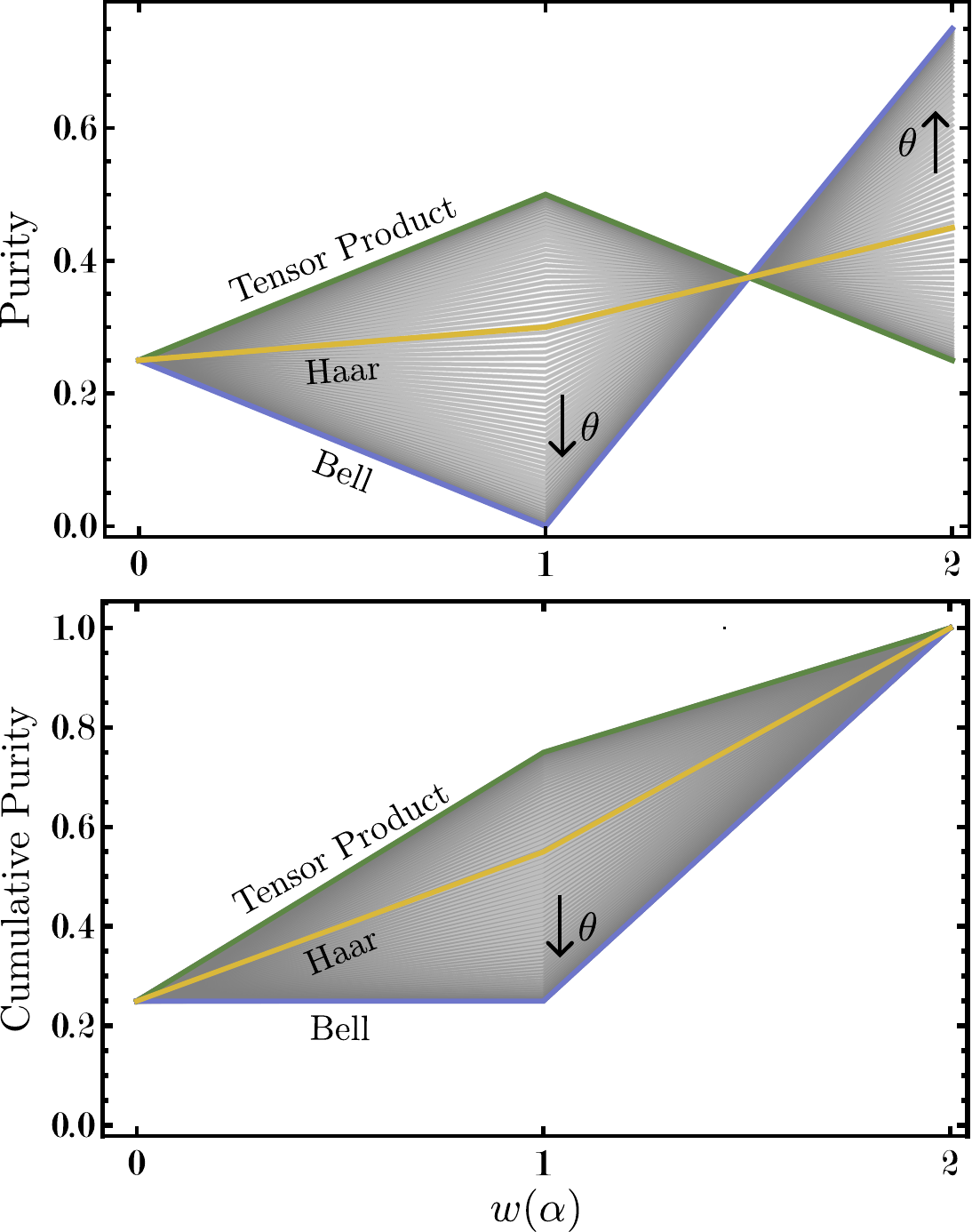}
    \caption{\textbf{GFD purities for two-qubit bipartite entanglement.} (Top) Purity profiles for $\rho_P$, $\rho_B$, $\rho_\theta$ and $\rho_H$ (defined in the main text), aggregated by irrep dimension, i.e., by irreps with the same Hamming weight $w(\alpha)$. As such, the purity shown for irrep with  $w(\alpha)=0$ is $\PC_{(0,0)}(\rho)$, that for  $w(\alpha)=1$ is $\PC_{(0,1)}(\rho)+\PC_{(0,1)}(\rho)$, and that for  $w(\alpha)=2$ is $\PC_{(1,1)}(\rho)$. (Bottom) Cumulative purity for the same states in the top panel.   }
    \label{fig:2q}
\end{figure}

At this point, we find it convenient to note that the GFD purities satisfy the property $\PC_{(0,0)}(\rho)=\frac{1}{4}$, and $\PC_{(1,0)}(\rho)=\PC_{(0,1)}(\rho)$ for any $\rho$. The first equality is trivial as it corresponds to the identity component of the state. Then, one finds  that the purities in $\LC_{(1,0)}$ and $\LC_{(0,1)}$ match due to the fact that the reduced density matrices of a bipartite quantum system (of the same dimension) are equal up to an unimportant local change of basis (i.e., purities are constant under free operations as per Eq.~\eqref{eq:invariance-purity}. Moreover,  Eq.~\eqref{eq:normalization} implies
\begin{equation}\label{eq:relation-purities}
    2\PC_{(1,0)}(\rho)+\PC_{(1,1)}(\rho)=\frac{3}{4}\,,
\end{equation}
clearly indicating that when purity is removed from the  small dimensional irreps, it must go to the larger dimensional one. In particular, Eq.~\eqref{eq:relation-purities} shows there is a single degree of freedom for the purities in $\LC(\HC)$ (as per Theorem~\ref{theo:trivial}~\cite{larocca2022group,mele2023introduction} in the Methods section), which implies that bipartite entanglement has a simple structure over the irreps of $\LC(\HC)$.

Next, consider the following states of interest:  a generic product, non-resourceful,  state $\rho_P=\dya{\psi_1}\otimes \dya{\psi_2}$, the maximally entangled Bell state $\rho_B=\dya{\Phi}$, with $\ket{\Phi}=\frac{1}{\sqrt{2}}(\ket{00}+\ket{11})$, a parametrized state that interpolates between product and maximally-entangled given by $\rho_{\theta}=\dya{\psi(\theta)}$, $\ket{\psi(\theta)}=\cos(\frac{\theta}{2})\ket{00}+\sin(\frac{\theta}{2})\ket{11}$ for $\theta\in [0,\pi/2]$, and a two-qubit Haar random state $\rho_H=\dya{\psi_H}$, with $\ket{\psi_H}$ sampled according to the Haar measure over $\HC$. Their GFD purities\footnote{For the case of Haar random states, we will henceforth report their average, expected purities.}  are shown in Fig.~\ref{fig:2q}, where we can already observe some of the typical behavior that will arise across all QRTs considered in this work. First, we can see from the top panel that as the states become more entangled (e.g., by increasing $\theta$ from $0$ to $\pi/2$ and interpolating between the all zero state and the Bell state), the purity gets shifted from the smaller dimensional irreps $\LC_{(1,0)}$ and $\LC_{(0,1)}$ to the largest irrep $\LC_{(1,1)}$.  
Importantly, we can also see that while a Haar random state has a larger purity in $\LC_{(1,1)}$ than a tensor product state, it does not saturate this value as the Bell state does. In some way, this indicates that for a state to maximize its  purity in the largest irrep, the entanglement in the state needs to be somewhat ``structured''. Moreover, we can observe that the purity in the irreps with $w(\alpha)=1$ is maximized for tensor product state (which we recall, are the free states of the QRT). This is a consequence of  Theorem~\ref{theo:algebra-purity} (generalized to a Lie algebra with two reductive components). In particular, one can verify that for any $\rho$, $\PC_{(1,0)}(\rho)=\frac{1}{2}\left(\Tr[\rho_1^2]-\frac{1}{2}\right)$, where $\rho_1$ denotes the reduced state on the first qubit, meaning that the purity on $\LC_{(1,0)}$, and concomitantly also that in $\LC_{(0,1)}$, are entanglement witnesses as they are proportional to the reduced state purity~\cite{wilde2013quantum}.  Indeed, in the limiting case of a Bell state, the purity in $\LC_{(1,0)}$ and $\LC_{(0,1)}$ is exactly zero, indicating maximal entanglement. Then, we note that in the bottom panel of Fig.~\ref{fig:2q} we also depict the cumulative purity, which also shows that the probability distribution of entangled states have larger mass in the larger irrep.

Then,  we also illustrate the result in Theorem~\ref{prop:compress} by noting that given a pure product state, if we find a state $\sigma$ such that $\Tr[\rho_P B]=\Tr[\sigma B]$ for $B\in\{X_1,Y_1,Z_1,X_2,Y_2,Z_2\}$, then $\sigma =\rho_P$. To see why this is the case we note that since $\Tr[\rho_P B]$ is pure, then $\Tr[\rho_P Z_\mu]^2+\Tr[\rho_P Y_\mu]^2+\Tr[\rho_P Z_\mu]^2=1$ (for $\mu=1,2$), and concomitantly,   $\Tr[\sigma Z_\mu]^2+\Tr[\sigma Y_\mu]^2+\Tr[\sigma Z_\mu]^2=1$. The latter implies that $\sigma$ is pure and a product state. Since product states are determined by their marginals, and the marginals of $\rho_P$ and $\sigma$ match, one obtains $\sigma=\rho_P$ as desired.

\subsubsection{Multipartite entanglement in $n$-qubit states}
Next, we study the QRT of multipartite entanglement in an $n$-qubit system. Now, $\HC=(\mathbb{C}^2)^{\otimes n}$, and the free operations are given by the standard representation $R$ of the group $G=\mathbb{SU}(2)\times \mathbb{SU}(2)\times \cdots\times \mathbb{SU}(2)$. Again, $R$ is irreducibly represented in $\HC$, so we need to study the GFD purities for the irreps in $V=\LC(\HC)$ under the adjoint action of $R$. Operator space decomposes into $2^n$ multiplicity-free irreps as
\begin{equation}\label{eq:irreps-L-nq}
    \LC(\HC)=\bigoplus_{\alpha \in\{0,1\}^{\otimes n}}\LC_{\alpha}\,,
\end{equation}
where the basis of $\LC_{\alpha}$ contains all operators of the form $P_1^{\alpha_1}\otimes P_2^{\alpha_2}\otimes \cdots \otimes P_n^{\alpha_n}$, and thus is of dimension $3^{w(\alpha)}$. For ease of notation, we will henceforth group the purities of all irreducible representations of the same dimension—when $w(\a)=k$, there are $\binom{n}{k}$ such irreps. 

\begin{figure}[t]
    \centering
\includegraphics[width=.9\linewidth]{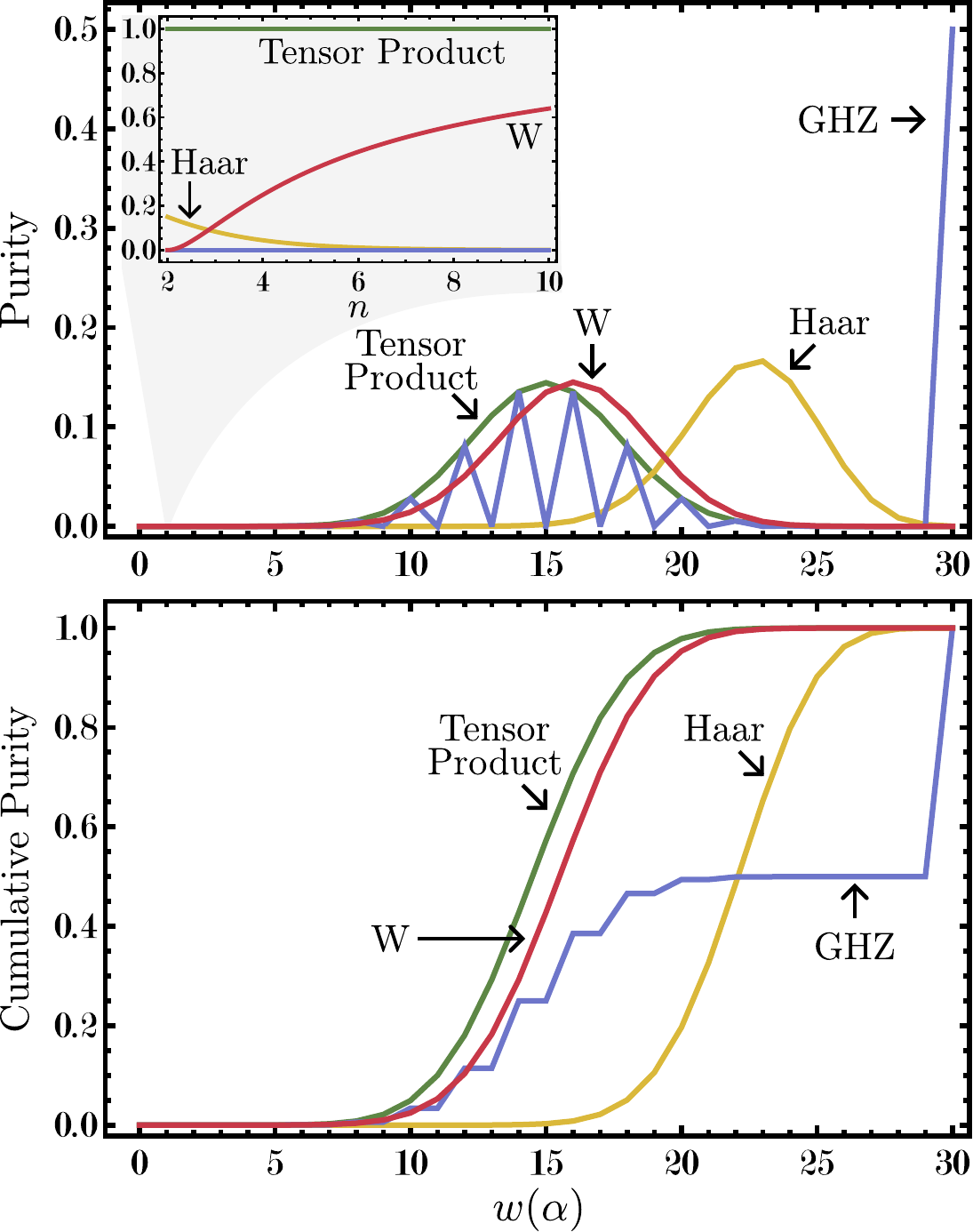}
    \caption{\textbf{GFD purities for $n$-qubit multipartite entanglement.} Purity profiles for $\rho_P$, $\rho_{{\rm GHZ}}$, $\rho_{{\rm W}}$ and $\rho_H$ (defined in the main text), aggregated by irrep dimension, i.e., by irreps with the same Hamming weight $w(\alpha)$. We show representative result $n=30$ qubits. Inset depicts the normalized purity $\LC_{\alpha}$ for an irrep with $w(\alpha)=1$ versus the number of qubits $n$. (Bottom) Cumulative purity for the same data in the top panel.  }
    \label{fig:nq}
\end{figure}

Then, consider the following states:  a generic tensor product state $\rho_P=\bigotimes_{\mu=1}^n \dya{\psi_\mu}$, the GHZ state $\rho_{{\rm GHZ}}=\dya{{\rm GHZ}}$ with $
\ket{{\rm GHZ}}=\frac{1}{\sqrt{2}}(\ket{0}^{\otimes n}+\ket{1}^{\otimes n})$, the $W$ state $\rho_{{\rm W}}=\dya{{\rm W}}$ with $\ket{W}=
\frac{1}{\sqrt{n}}\sum_{\mu=1}^{n}X_i|0 \rangle^{\otimes n}$, and an $n$-qubit Haar random state $\rho_H=\dya{\psi_H}$. Their GFD purities are shown in Fig.~\ref{fig:nq}. From the top panel, we can see that for all considered entangled states, the  purities are shifted towards higher-dimensional irreps (a fact also reflected in the cumulative probabilities of the bottom panel). In particular, while the profiles for the W and the Haar random state are smooth, that of the GHZ state is significantly different. For a GHZ state, the purity exhibits a zig-zag pattern that alternates between being equal to zero or to that of a tensor product state; except for the largest irrep which contains all the remaining purity.  In fact, as we can see from the cumulative purity, almost half of the GHZ's mass is in the largest irrep. As expected, we observe that the GFD purity profiles exhibit a richer structure and a larger variety of shapes than those arising in bipartite entanglement. As per Theorem~\ref{theo:trivial} (see the Methods) this is a consequence of the fact that the number of trivial components in the two-fold tensor product decomposition of $\LC(\HC)$ grows exponentially with $n$~\cite{larocca2022group}, leading to more degrees of freedom. 

Then, the inset of the top panel of Fig.~\ref{fig:nq}  shows  the normalized purity in the irreps $\LC_{\alpha}$ with $w(\alpha)=1$. As expected from Theorem~\ref{theo:algebra-purity}, this quantity is a multipartite entanglement witness and is maximized for the tensor product state, equal to zero for GHZ,  exponentially small for a Haar random state, and converging to one for W states (indicating that W states are ``not too'' entangled). We can understand why these purities are bona fide measures of entanglement from the fact that they are proportional to $\sum_{\mu=1}^n \Tr[\rho_\mu^2]$ (with $\rho_\mu$ the reduced density matrix on the $\mu$-th qubit), a well-known witness for multipartite entanglement~\cite{meyer2002global,brennen2003observable,beckey2021computable,schatzki2022hierarchy}. 

In the bottom panel of Fig.~\ref{fig:nq}, we also show the cumulative probability. Here, the curve for the tensor product state lies at the ``left'' of the plot, indicating that its purity is mostly concentrated in the smaller dimensional irreps. Then, the $W$ state and Haar random states have smooth curves that are shifted to the ``right'', with that of the Haar random state showing that $\rho_H$ is, on average, more distributed among all large dimensional irreps. 

To finish, we note that while the projection onto $\LC_{\alpha}$ with $w(\alpha)=1$ provides an entanglement witness, one can also construct combinations of GFD purities that enable the hierarchical study of different structures of multipartite entanglement, showing that there can exist projections into special subsets of irreps that are also entanglement monotones and which carry additional meaning~\cite{schatzki2022hierarchy}.

As a bonus fact, we note that when they exist, absolutely maximally entangled states $\rho_{\rm AME}$~\cite{helwig2013absolutely,goyeneche2015absolutely,huber2017absolutely,casas2025quantum}, will possess an extreme GFD profile. Indeed, as shown in the Methods,  their purity will be zero on all irreps except the smallest trivial irrep $\LC_{(0,\ldots,0)}$, and the largest irrep $\LC_{(1,\ldots,1)}$. In particular, one finds   $\PC_{(1,\ldots,1)}(\rho_{\rm AME})=\frac{2^n-1}{2^n}$, meaning that absolutely maximally entangled states maximize the support on the largest irrep.

\subsubsection{Fermionic Gaussianity}
Next, we study the QRT of fermion Gaussianity in a system of $n$ spinless fermions.  Here, we can use the isomorphism between the Fock space of $n$ Majorana modes and the
$n$-qubit Hilbert space to set  $\HC=(\mathbb{C}^2)^{\otimes n}$. Then, free operations are obtained from the spinor representation $R$ of  $G=\mathbb{SO}(2n)$~\cite{brauer1935spinors,kokcu2021fixed,kazi2024analyzing}, while free states are Gaussian states~\cite{weedbrook2012gaussian,hebenstreit2019all,diaz2023showcasing,jozsa2008matchgates,mele2024efficient,dias2023classical,goh2023lie,brod2011extending}. Now, $\HC$ decomposes into two irreps as $\HC=\HC_+\oplus \HC_-$, where $\HC_{\pm}$ contains all the eigenstates of the fermionic parity operator $Z^{\otimes n}$ with eigenvalue equal to $\pm 1$. As such, all physical states  within a parity sector (e.g., even) have the same purity regardless of how much non-Gaussianity they possess. Therefore, one needs to study the GFD purities for the irreps in $V=\LC(\HC)$ under the adjoint action of $R$. To understand the decomposition of the operator space we find it convenient to define the (Hermitian) Majorana fermionic operators,
\begin{equation}
\begin{split}
    c_1&=X\id\dots \id,\; c_3= ZX\id\dots \id, \;\dots,\; c_{2n-1}=Z\dots Z X\,, \nonumber\\
        c_2&=Y\id\dots\id,\; c_4= ZY\id\dots \id, \; \dots,\;\; c_{2n}\;\;\;=Z\dots Z Y\,,
\end{split}
\end{equation}
where we also recall that the unitaries in $R(G)$ are obtained as the exponential of products of two discting Majoranas. 
From here, one finds that 
\begin{equation}  \label{eq:irrep-L-ferm}  \LC(\HC)=\bigoplus_{\alpha=0}^{2n}\LC_\alpha\,,
\end{equation}
where $\LC_\alpha$ is spanned by the product of $\alpha$ distinct Majoranas, meaning that $\dim(\LC_\alpha)=\binom{2n}{\alpha}$. 

Then, let $n$ be a multiple of four and let us consider states belonging to the even parity subsector (i.e., their purity in all irreps $\LC_\alpha$ with odd values of $\alpha$ is zero). We will consider here a   generic pure Gaussian state $\rho_G$, and three types of non-Gaussian states. First, the GHZ state (which we note can be expressed as the linear combination of two Gaussian states). Second, the parametrized state $\rho_\gamma=\dya{\chi(\gamma)}$ with 
\begin{equation}
    \ket{\chi(\gamma)}=(\cos(\gamma/4)\ket{0}^{\otimes 4}+\sin(\gamma/4)\ket{1}^{\otimes 4})^{\otimes n/4}\,,
\end{equation}
which interpolates between a Gaussian state at $\gamma=0$ and a highly non-Gaussian state at $\gamma=\pi$. Importantly, for any strictly non-zero value of $\gamma$ the state has exponential  extent~\cite{dias2023classical,cudby2023gaussian}, meaning that its expansion in Gaussian states contains exponentially many terms. As such, we refer to this state as the extent state. And third, a  Haar random state $\rho_H=\dya{\psi_H}$.

\begin{figure}[t]
    \centering
\includegraphics[width=.9\linewidth]{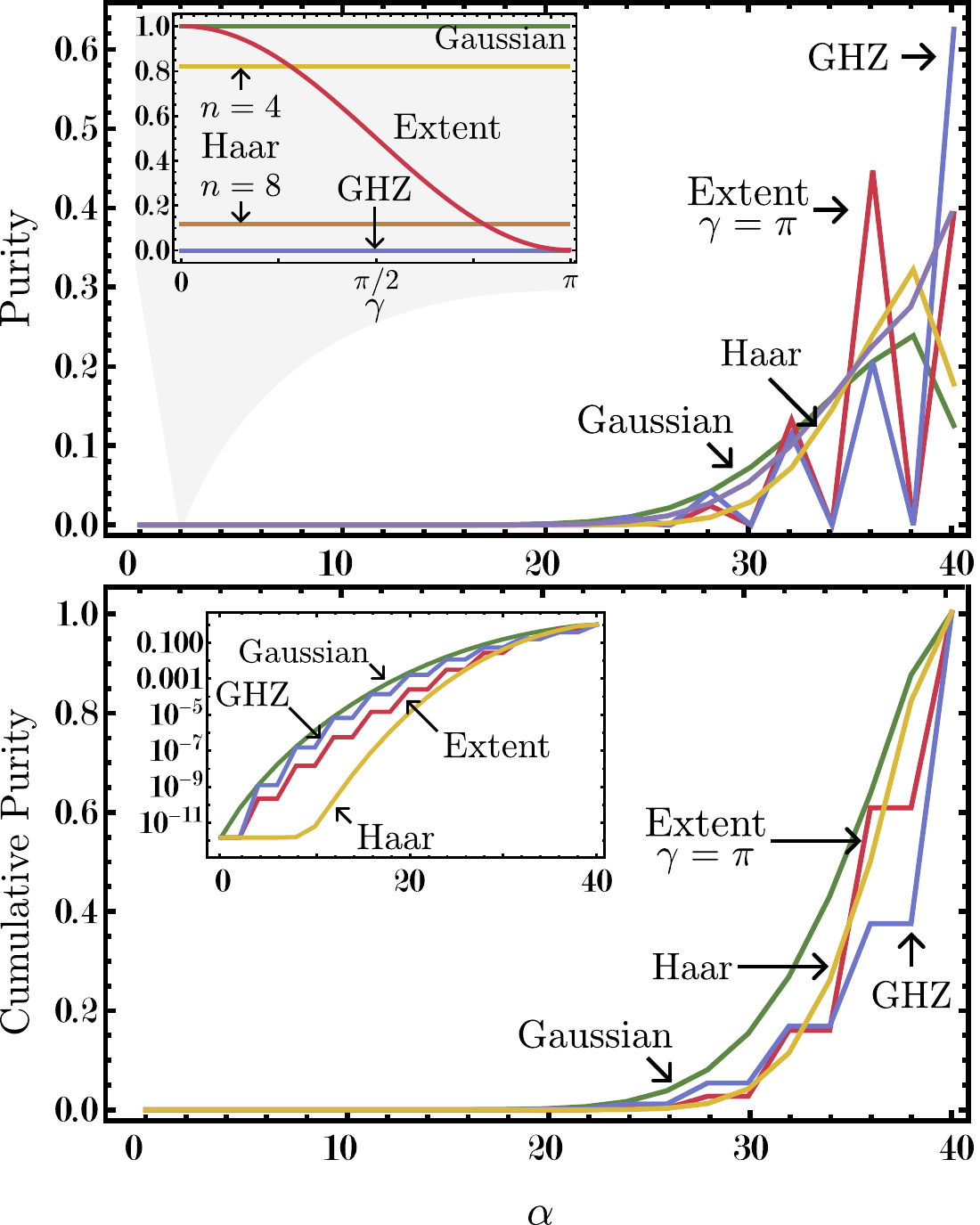}
    \caption{\textbf{GFD purities for fermionic Gaussianity.} Purity profiles for $\rho_G$, $\rho_{{\rm GHZ}}$, $\rho_{\gamma}$ and $\rho_H$ (defined in the main text), aggregated by irrep dimension. That is, for index $\alpha< n$, we add $\LC_{\alpha}+\LC_{2n-\alpha}$. We also only show purity for irreps with $\alpha$ even. We show representative result $n=40$ qubits. Inset depicts the normalized purity $\LC_{2}$ as a function of the extent parameter $\alpha$.  (Bottom) Cumulative purity for the same data in the top panel. Inset shows the same data on a log-linear scale.}
    \label{fig:fer}
\end{figure}

In Fig.~\ref{fig:fer} we show the purity profiles for the aforementioned states. We can again see from the top and bottom panels that as the states become non-Gaussian, the purity gets shifted towards higher-dimensional irreps. Notably, while the profile for a Gaussian and Haar random states are ``smooth'', both GHZ and the extent state (for $\gamma=\pi$) exhibit zig-zag profiles with zero purity in every-other irrep (with the GHZ accumulating purity in the smaller irreps faster than the extent state, as per the bottom panel).  However, a difference between the extent and GHZ state is that the former maximizes its purity in the third-to-last irreps, while the latter has maximal purity in the largest ones. In fact, we see that when the purities of the GHZ state are non-zero, they exactly match those of a Gaussian state (see the Methods), except of course on the largest irreps, where the GHZ accrues all of its mass. 

Then, in the inset of the top panel we show the normalized purity in $\LC_2$ as a function of the extent parameter $\gamma$. Here, we can see that the QRT's free states (i.e. Gaussian states) maximize the value of $\LC_2$. Since this irrep  corresponds to the complexification of the Lie algebra $R(\mathfrak{so}(2n))$ associated to the free operations, the inset showcases an instantiation of Theorem~\ref{theo:algebra-purity}. In fact, it was shown in~\cite{diaz2023showcasing} that such purity is a proper witness which quantifies the amount of  fermionic non-Gaussianity in the state~\cite{gigena2015entanglement, gigena2020one,gigena2021many}. Importantly, here we note that while the values of this purity as a function of $n$ are constant for  $\rho_G$, $\rho_{{\rm GHZ}}$, and $\rho_{\gamma}$ they decay exponentially for $\rho_H$. Moreover,  our results show that there is no monotonic relationship between the amount of extent in the state, and its fermionic non-Gaussianity, as measured through the GFD purity in $\LC_2$. In particular, the GHZ state has an extent of two, but it has zero component on $\LC_2$. On the other hand, the state $\ket{\chi(\gamma)}$ has  exponential extent~\cite{dias2023classical,cudby2023gaussian} as long as $\gamma$ is not strictly zero, but its purity in $\LC_2$ only decays as $\cos\left(\frac{\gamma}{2}\right)$. Hence, exponential-extent states can arise at arbitrarily low levels of non-Gaussianity.

To finish, we note that for Gaussian states, Theorem~\ref{prop:compress} claims that if there exists two states whose expectation values over $\LC_2$ (which corresponds to the complexified algebra) match, then they must be equal. This implies that the two states share the same correlation matrix $C$, which is defined for a state $\rho$ as the real anti-symmetric $2n \times  2n$ matrix with element $[C(\rho)]_{ij}=-i\Tr[\rho c_i c_j]$ for all $1\leq i< j<2n$. Indeed, combining the fact that the Gaussian states maximize the purity in $\LC_2$~\cite{diaz2023showcasing} along with the fact that Gaussian states are fully determined by their correlation matrix~\cite{mele2024efficient}, recovers the statement of Theorem~\ref{prop:compress}.

\subsubsection{Spin coherence}
In the spin coherence QRT the free operations are given by the spin $s=(d-1)/2$ irreducible representation of $\mathbb{SU}(2)$~\cite{barnum2004subsystem,robert2021coherent,perelomov1977generalized,zhang1990coherent}, acting on a $d$-dimensional Hilbert space $\HC=\mathbb{C}^{d}$. Since the group is irreducibly represented, all the states have the same purity on $\HC$, and we need to go up the hierarchy and study the GFD purities in $\LC(\HC)$. Here, one finds  that the operator space decomposes into $2s+1$ multiplicity-free irreps as 
\begin{equation}
\LC(\HC)=\bigoplus_{\alpha=0}^{2s}\LC_\alpha\,,
\end{equation}
where $\LC_s$ is a spin-$\alpha$ irrep of dimension $\dim(\LC_s)=2\alpha +1$. 

Then, we recall that the free states are those in the orbit of $\ket{s,s}$, whereas all the other states $\ket{s,m}$ with $m\neq \pm s$ are resourceful~\cite{barnum2004subsystem,ragone2022representation}. As such, we will will consider the free state $\rho_{s}=\dya{s,s}$, as well as the resourceful states $\rho_{0}=\dya{s,0}$, the GHZ state $\rho_{{\rm GHZ}}=\dya{{\rm GHZ}}$ with $
\ket{{\rm GHZ}}=\frac{1}{\sqrt{2}}(\ket{s,s}+\ket{s,-s}$, and a Haar random state $\rho_H=\dya{\psi_H}$. 
In Fig.~\ref{fig:su2} we show their associated GFD purity profiles. Once again, we can verify that as the state becomes more resourceful, purity shifts towards higher dimensional spin irreps. Notably, the purity of Haar random states increases linearly with $\alpha$ as indicated by the purity (top panel) and its cumulated version (bottom panel). Then, both the $\rho_{0}$ and GHZ states exhibit a zig-zag pattern, with that of $\rho_{0}$ increasing in magnitude as a function of $\alpha$, while that of the GHZ matches the values of the $\rho_{s}$ for $\alpha$ even (see the Methods). The latter explains why the GHZ cumulative purity plateaus as $\alpha$ increases and approaches $2s$  (since the purity in  $\rho_{s}$ is small for such irreps), which makes the jump in the cumulative purity arising from the largest irrep more dramatic. In the inset of the top panel we show the normalized purity in $\LC_1$, which is known to be a witness for the resourcefulness in a pure state as per Theorem~\ref{theo:algebra-purity}~\cite{barnum2004subsystem,ragone2022representation}. Here we can see that such purity is maximized for the free state $\ket{s,s}$, zero for the GHZ state, exponentially small (in $s$) for the Haar states, and decreasing as $\frac{m^2}{s^2}$ for $\rho_{m}$.

\begin{figure}[t]
    \centering
\includegraphics[width=.9\linewidth]{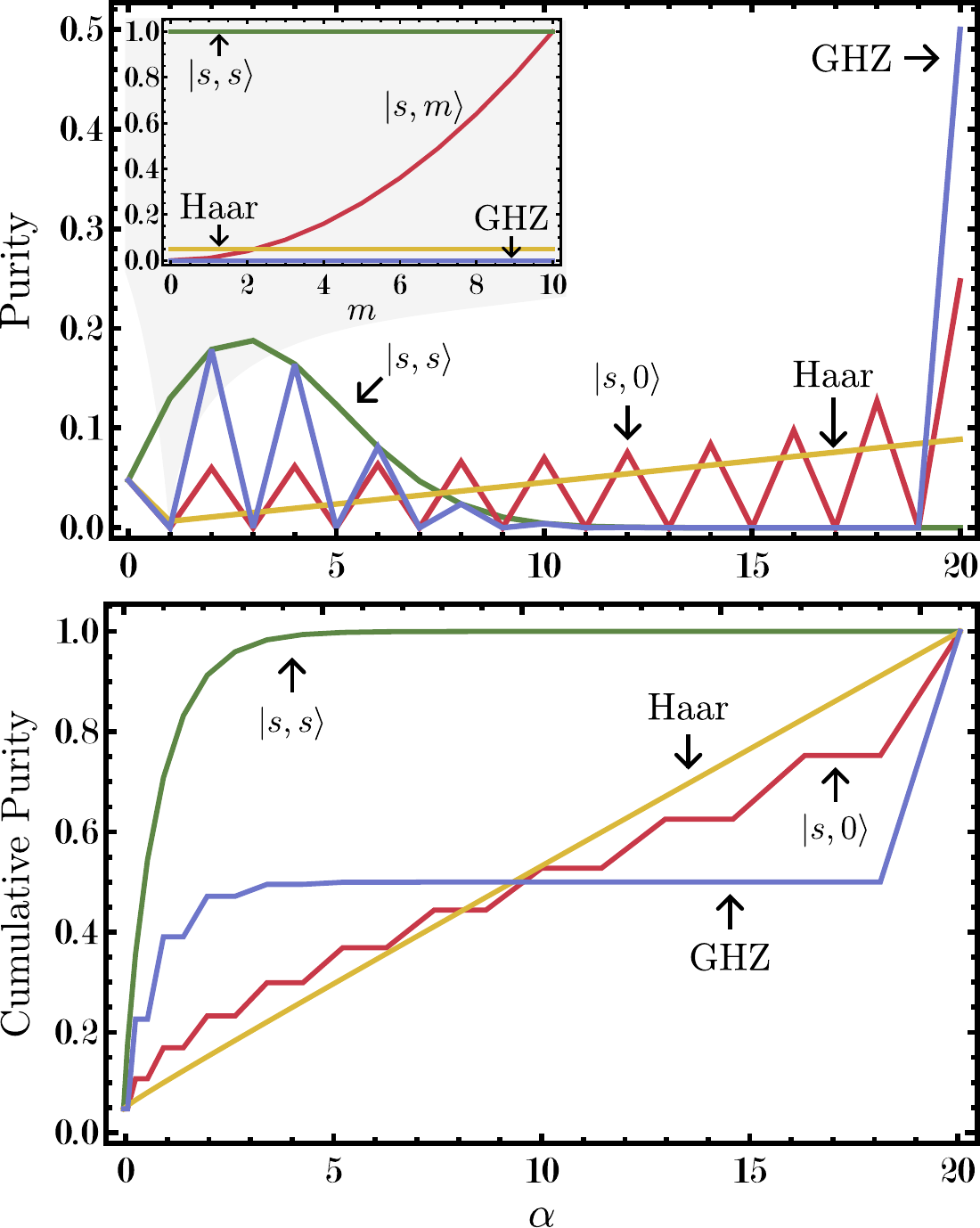}
    \caption{\textbf{GFD purities for spin coherence.} Purity profiles for $\rho_{s}$, $\rho_{0}$, $\rho_{{\rm GHZ}}$, and $\rho_H$. We show representative result for $s=10$. Inset depicts the normalized purity $\LC_{1}$ as a function of the extent index $m$ for $\rho_{m}$.  (Bottom) Cumulative purity for the same data in the top panel.}
    \label{fig:su2}
\end{figure}

\subsubsection{Clifford stabilizerness}

Up to this point we have presented results for QRTs where the free operations are given by the representation of a Lie group. However, we can also apply our GFD theory for the cases where the free operations form a discrete group. In particular, we will study the QRT for Clifford stabilizerness, or magic~\cite{veitch2014resource,howard2017application,leone2024stabilizer}. As such, let $\HC=(\mathbb{C}^2)^{\otimes n}$ be an $n$-qubit Hilbert space of dimension $d=2^n$, and let the free operations be given by the Clifford group $C_l$, i.e., the normalizer of the $n$-qubit Pauli group $P_n$ generated by $\{\pm1 ,\pm i\}\times \{\id,X,Y,Z\}^{\otimes n}$~\cite{mastel2023clifford}. Since Clifford unitaries form a $3$-design~\cite{webb2016clifford}, their irrep decomposition matches that of the full unitary group $\mathbb{U}(\HC)$, not only on $\HC$ itself but also on the spaces $\LC(\HC)$ and $\HC^{\otimes 2}$. Consequently, all pure states will have the same irrep decomposition in all irreps for those spaces, and we need to go up the hierarchy and study the irreps of $\LC(\HC^{\otimes 2})$. The operator space $\LC(\HC^{\otimes 2})$ splits into fourteen irreps as~\cite{helsen2018representations} 
\begin{align}
    \LC(\HC^{\otimes 2})=&\LC_{{\rm id}}\oplus\LC_{r}\oplus \LC_{l}\oplus \LC_{[A]}\oplus\LC_{0}\oplus\LC_{1}\oplus \LC_{2}\oplus \LC_{[{\rm adj}]}\nonumber\\
    &\oplus \LC_{\{{\rm adj}\}}\oplus \LC_{\{{\rm adj}^\perp\}}\oplus \LC_{[1]}\oplus \LC_{[2]}\oplus \LC_{\{1\}}\oplus \LC_{\{2\}}\,,\nonumber
\end{align}
whose description and dimensions are presented in Table~\ref{tab:irreps} of the Methods. 

As shown in our companion manuscript~\cite{mele2025clifford}, given two copies of a pure state  $\rho=\dya{\psi}$, one finds that $\forall \ket{\psi}\in\HC$
\begin{align}
&\PC_{{\rm id}}(\rho^{\otimes 2})=\frac{1}{4^n}\,, \quad \PC_{0}(\rho^{\otimes 2})=\frac{2^n-1}{4^n(4^n-1)}\,, \nonumber\\
&\PC_{ l}(\rho^{\otimes 2})=\PC_{ r}(\rho^{\otimes 2})=\frac{2^n-1}{4^n}\nonumber\\
    &\PC_{[A]}(\rho^{\otimes 2})=\PC_{2}(\rho^{\otimes 2})=\PC_{\{{\rm adj}\}}(\rho^{\otimes 2})=\PC_{\{{\rm adj}^\perp\}}(\rho^{\otimes 2})=0\,.\nonumber
\end{align}
More importantly, one can show that all remaining irrep purities are linearly dependent to that in $\LC_1$. That is, for all remaining irreps $\LC_\alpha$ one finds
\begin{equation}\label{eq:prop-irrep}
    \PC_\alpha(\rho^{\otimes 2})=a_\alpha\PC_1(\rho^{\otimes 2})+b_\alpha\,.
\end{equation}
Notably, we have shown in \cite{mele2025clifford} that the stabilizer entropy~\cite{leone2022stabilizer}, a  well-known measure of non-stabilizerness, can be found to be up-to-a constant, equal to a projection onto $ \LC_1$, implying that this GFD purity is a stabilizerness witness and monotone~\cite{leone2024stabilizer}. Such a notable result indicates that a generalized version of Theorem~\ref{theo:algebra-purity} could potentially be defined for non Lie-groups, whereby there exists a special irrep whose associated GFD purity is a witness. In addition, here we also note that since all GFD purities in $\LC(\HC^{\otimes 2})$ are proportional to that in $\LC_1$ as per Eq.~\eqref{eq:prop-irrep}, there is only a single  degree of freedom available, indicating that the structure of the stabilizer resourcefulness is quite simple. Particularly, when purity leaves $\LC_1$ it deterministically moves to the other irreps.

Then, consider the following states: a generic free stabilizer state $\rho_{\rm stab}=\dya{{\rm stab}}$ (including the GHZ state),  the $n$-qubit magic state $\rho_M=\dya{M}^{\otimes n}$ with  $\ket{M}=\frac{1}{\sqrt{2}}(\ket{0}+e^{i\pi/4}\ket{1})$, and  a Haar random state. As shown in the Methods, their purities in $\LC_1$ obey $\PC_1(\rho_M)<\PC_1(\rho_H)<\PC_1(\rho_{\rm stab})$, with $\PC_1(\rho_{\rm stab})$ being maximal. 

\subsection{GFD purities, extent and classical simulability methods}

In this section we study the connection between GFD purities and extent. For convenience, we recall that, in a nutshell, the extent is used to study the optimal decomposition of a resourceful state into a linear combination of free states~\cite{bravyi2019simulation,heimendahl2021stabilizer,fast2022pashayan,reardon2024fermionic,dias2023classical,reardon2023improved,cudby2023gaussian}, and plays a central role in the QRTs of stabilizerness and Gaussianity. In particular,  the larger the extent, the more free states are needed to decompose a resourceful state. Our work shows that across the considered Lie algebra-based QRTs there seems to be no direct connection between how resourceful a state is through the optics of extent, and through that of the projection into the irreps which form witnesses. For instance, in all cases the GHZ state has zero component on the witness irreps, but it can always be decomposed into a linear combination of only two free states. This is in complete contrast to the QRT of Clifford stabilizerness, where the inverse purity in the witness irrep is a lower bound on the extent~\cite{haug2023stabilizer}. The previous therefore indicates that exponentially small purity implies exponentially large extent. Such difference could arise from the fact that stabilizerness appears to be a somewhat ``simpler'' QRT than others. 

The previous distinction between the QRT of stabilizerness and those based on Lie-algebras could also have deep implications in the study of the classical simulability of quantum systems. This is due to the fact that the amount of extent in a state determines how hard it is to classically simulate it, e.g.,  via the Gottesman–Knill theorem for Clifford evolutions~\cite{gottesman1998theory,
bravyi2019simulation,dias2023classical}, or using Wick-Theorem based techniques~\cite{wick1950evaluation,valiant2001quantum,terhal2002classical} in free-fermionic matchgate circuits. As such, the non-equivalence between extent and purity in fermionic systems could be a strong indicator that there exist different types of fermionic correlations, and concomitantly, that different non-equivalent ways to simulate their free-fermionic evolutions should exist. Indeed, this conjecture is supported by the fact that it was recently shown in~\cite{goh2023lie} that Lie-algebraic simulation techniques (which are actually based on GFDs) can enable the simulation of states with exponentially large extent but large purity in the small irreps, a task which is intractable via standard Wick-based simulations.

\section{Discussions}

In this work we have presented a framework for studying the resourcefulness of quantum states via their decompositions into the irreps induced by the QRT's free operations. Notably, our results reveal that some of the behaviors arising in classical data harmonic analysis still appear when studying quantum states. For instance, low-complexity resource-free states have large support, and can be described by their components, in the smaller dimensional irreps. On the other hand, highly-resourceful and complex states accumulate their probability mass in the larger dimensional irreps. In addition, we also identified universal behaviors across multiple QRTs, such as irrep profiles exhibiting zig-zag patterns for highly structured resourceful states. Indeed, these patterns enable states to have zero purity in certain irreps, and to push this probability mass into the larger irreps.   
Hence, the GFD framework not only provides a toolbox to characterize the resourcefulness fingerprints in states across different QRTs, but given the fact that the GFD purities carry operational meaning, it also enables new dimensions of study for constructing witnesses and state compression techniques. 

Then, it is also worth mentioning that our framework provides tools to further understand the power and limitations of certain quantum computational models. For instance, the GFD purities have been recently shown to appear within the context of quantum machine learning~\cite{biamonte2017quantum,schuld2015introduction,cerezo2020variationalreview,cerezo2022challenges} as determining whether a learning model will be trainable or not~\cite{ragone2023unified,fontana2023theadjoint,larocca2024review}. In fact, it can be shown that many algorithms are only provably trainable for states and observables that live within the smaller dimensional irreps ~\cite{bermejo2024quantum, braccia2024computing}. From the optics of GFD, this establishes a clear connection between trainability and resourcefulness, so that the more resourceful a state is, the harder it will be to train a learning model on it.

Going beyond our results, our work also provides blueprints for several future research directions. These include the study of GFD purities in other QRTs (we point the reader to our companion paper \cite{deneris2025analyzing}, where GFD purities are employed to study several families of well-known states through the lens of different QRTs) and in higher number of copies (i.e., going up the hierarchy), the  exploration of GFDs for mixed states, the study of other properties of the irrep projections (beyond their norm), whether linear combinations of GFD purities can provide different types of information~\cite{schatzki2022hierarchy},  how purities change under non-unitary resource-non increasing operations (we refer the reader to our companion manuscript~\cite{nahuel2025quantum} where we discuss how such operations could be obtained in general QRTs), the connection between purities to resourcefulness distillation and conversion rates, and whether comparisons between two GFD probability distribution encode information useful to compare states. In addition, one can also explore  GFDs  in regular representations or those arising from lifting the states to the group, i.e., re-expressing them in terms of group elements. Finally, and as evidenced from Theorem~\ref{theo:algebra-purity}, it is clear that the GFD purities can contain large quantities of redundant information (e.g., free states are fully described by and compressible to one irrep), indicating that there likely exists other types of analysis which reveal whether the information on larger irreps is redundant or not.  As evidenced from this non-exhaustive list of future research directions, we are confident that our proposed formalism can lead to useful, and new insights regarding QRTs, and more generally, quantum computation and information sciences.

\section{Methods}

\subsection{GFD purities and trivial irreps of $V^{\otimes 2}$}

As mentioned in the main text, one can understand the GFD purities of Definition~\ref{def:irrep-purity} as arising from projectors into the trivial components of the irrep decomposition of the two-fold tensor product  of $V$.  
To see why this is the case, let us consider  $V^{\otimes 2}$, and note that the representation $R$ induces the tensor product representation $R':G\rightarrow \mathbb{GL}(V^{\otimes 2})$ whose action is given by $R'(g)=R(g)\otimes R(g)$. In particular, we recall that one can use Eq.~\eqref{eq:block-diagonalization} as a starting point to build the irrep decomposition of $V^{\otimes 2}$ under the action of $R'$ since
\begin{equation}
    V^{\otimes 2} = \bigoplus_{\alpha,\alpha'} \mathbb{C}^{m_\alpha} \otimes \mathbb{C}^{m_{\alpha'}}\otimes V^{\alpha}_G\otimes V^{\alpha'}_G\simeq \bigoplus_{\widetilde{\alpha}} \mathbb{C}^{m_{\widetilde{\alpha}}} \otimes V^{\widetilde{\alpha}}_G\,,
\end{equation}
where $\widetilde{\alpha}$ are the irreps that appear within the decomposition of $V^{\otimes 2}$ and which arise from further decomposing each tensor product $V^{\alpha}_G\otimes V^{\alpha'}_G$\footnote{We recall that one of the basic primitives of representation theory is the fact that one can obtain new irreps from the tensor product of two irreps}. In particular, from Shur's Lemma~\cite{hall2013lie}, we know that if $V^{\alpha}_G$ is isomorphic to $V^{\alpha'}_G$ (which in particular occurs if $\alpha=\alpha'$), then the decomposition of $V^{\alpha}_G\otimes V^{\alpha'}_G$ will contain (at least one) trivial irrep. Next, since the inner product $\langle\cdot,\cdot\rangle$ in $V$ induces an inner product in $V^{\otimes 2}$ as $\langle v_1\otimes v_2,w_1\otimes w_2\rangle=\langle v_1,w_1\rangle\langle v_2,w_2\rangle$, we find that the GFD purities can be expressed as 
\begin{align}
     \PC_{j,\alpha}(v)&=\sum_{i=1}^{\dim(V_{\alpha,j})}\left\langle \left(w_{\alpha,j}^{(i)}\right)^{\otimes 2},v^{\otimes 2}\right\rangle\nonumber\\
     &=\left\langle \Pi_{j,\alpha} ,v^{\otimes 2}\right\rangle\,.
\end{align}
Above we have defined $\Pi_{j,\alpha}=\sum_{i=1}^{\dim(V_{\alpha,j})} (w_{\alpha,j}^{(i)})^{\otimes 2}$. Importantly, the invariance of the purities in~\eqref{eq:invariance-purity} implies that $\Pi_{j,\alpha}$ must belong to the trivial irreps --also known as the group's representation commutant-- appearing in $V^{\alpha}_G\otimes V^{\alpha}_G$. Indeed, this implies the following  theorem
\begin{theorem}[ GFD purities and trivial irreps]\label{theo:trivial}
    The GFD purities can be expressed as a summation of the projections of $v^{\otimes 2}$ onto trivial irreps of $V^{\otimes 2}$ under the two-fold tensor product representation of $R$.
\end{theorem}

\subsection{Bipartite entanglement in two-qubit states}

Consider  any product, non-resourceful,  state $\rho_P=\dya{\psi_1}\otimes \dya{\psi_2}$. We can readily find that its purity in the irreps of Eq.~\eqref{eq:irreps-L-std-ent} are given by
\small
\begin{equation}\label{eq:irreps-L-std-ent}
    \PC_{(0,0)}(\rho_P)=\frac{1}{4}\,, \PC_{(1,0)}(\rho_P)=\PC_{(0,1)}(\rho_P)=\frac{1}{4}\,, \PC_{(1,1)}(\rho_P)=\frac{1}{4}\,.\nonumber
\end{equation}
\normalsize
Then, the maximally entangled Bell state $\rho_B=\dya{\Phi}$, with $\ket{\Phi}=\frac{1}{\sqrt{2}}(\ket{00}+\ket{11})$ satisfies
\small
\begin{equation}
    \PC_{(0,0)}(\rho_B)=\frac{1}{4}\,, \PC_{(1,0)}(\rho_B)=\PC_{(0,1)}(\rho_B)=0\,, \PC_{(1,1)}(\rho_B)=\frac{3}{4}\,.\nonumber
\end{equation}
\normalsize
Then, consider a parametrized state which interpolates between product and maximally-entangled given by $\rho_{\theta}=\dya{\psi(\theta)}$, $\ket{\psi(\theta)}=\cos(\frac{\theta}{2})\ket{00}+\sin(\frac{\theta}{2})\ket{11}$ for $\theta\in [0,\pi/2]$. Now, one obtains
 \small
\begin{align}
    &\PC_{(0,0)}(\rho_\theta)=\frac{1}{4}\,, \PC_{(1,0)}(\rho_\theta)=\PC_{(0,1)}(\rho_\theta)=\frac{\cos(\theta)^2}{4}\,,\nonumber\\ &\PC_{(1,1)}(\rho_\theta)=\frac{2-\cos(2\theta)}{4}\,.\nonumber
\end{align}
\normalsize
Finally, for a two-qubit Haar random state $\rho_H=\dya{\psi_H}$, with $\ket{\psi_H}$ sampled according to the Haar measure over $\HC$, we obtain that, in average,
\small
\begin{equation}
    \PC_{(0,0)}(\rho_H)\!=\frac{1}{4}\,, \PC_{(1,0)}(\rho_H)\!=\PC_{(0,1)}(\rho_H)\!=\!\frac{3}{20}\,, \PC_{(1,1)}(\rho_H)\!=\!\frac{9}{20}.\nonumber
\end{equation}
\normalsize

\subsection{Multipartite entanglement in $n$-qubit states}

Consider the irrep decomposition of Eq.~\eqref{eq:irreps-L-nq}, and let us define $k$ the Hamming weight of the irreps $\LC_\alpha$, i.e. $k=w(\a)$. Moreover, let us define
\begin{equation}
   \PC_{k}(\rho)= \sum_{\a:\,w(\alpha)=k} \PC_{\alpha}(\rho)\,,
\end{equation}
the aggregated purities of a state $\rho$ over all the irreps $\LC_\a$ with fixed Hamming weight $k$.
For any pure tensor product state $\rho_P=\bigotimes_{\mu=1}^n \dya{\psi_\mu}$ we find
\begin{equation}\label{eq:pur-TP}
   \PC_{k}(\rho_P)=\frac{\binom{n}{k}}{2^n}\,.
\end{equation}
Then, for the GHZ state $\rho_{{\rm GHZ}}=\dya{{\rm GHZ}}$ with $
    \ket{{\rm GHZ}}=\frac{1}{\sqrt{2}}(\ket{0}^{\otimes n}+\ket{1}^{\otimes n})$, we obtain that, for $k<n-1$  
\begin{equation}\label{eq:pur-GHZ-ent}
    \PC_{k}(\rho_{\rm GHZ})=\begin{cases}
        \frac{\binom{n}{k}}{2^n}\, & \text{if $k$ is even}\\
        0\, &\text{if $k$ is odd}
    \end{cases}\,.
\end{equation}
The cases $k=n-1,n$ depend on the parity of $n$. For $n$ even and  $k=n-1$, one finds $\PC_{\alpha}(\rho_{\rm GHZ})=0$, whereas it turns out that $\PC_{\alpha}(\rho_{\rm GHZ})=\frac{1}{2}+\frac{1}{2^n}$ if $k=n$. For odd $n$ and $k=n-1$, then $
    \PC_{\alpha}(\rho_{\rm GHZ})=\frac{1}{2^n}$, while $
    \PC_{\alpha}(\rho_{\rm GHZ})=\frac{1}{2}$ if $k=n$.
Hence 
\begin{align}\label{eq:pur-GHZ-ent-continued}
    \PC_{n-1}(\rho_{\rm GHZ})&=\begin{cases}
        0\, & \text{if $n$ is even}\\
        \frac{1}{2^n}\, &\text{if $n$ is odd}
    \end{cases}\,,\\
    \PC_{n}(\rho_{\rm GHZ})&=\begin{cases}
        \frac{1}{2}+\frac{1}{2^n}\,, & \text{if $n$ is even}\\
        \frac{1}{2}\,, &\text{if $n$ is odd}
    \end{cases}\,.
\end{align}

Next, consider the $W$ state $\rho_{{\rm W}}=\dya{{W}}$ with $\ket{W}=\frac{1}{\sqrt{n}}
    \sum_{\mu=1}^{n}X_i|0 \rangle^{\otimes n}$. We now find
\begin{equation}
\PC_{k}(\rho_{{\rm W}}) = \frac{(n-2k)^2 + 8\binom{k}{k-2}}{n^2 2^n}\binom{n}{k} \,.
\end{equation}
In addition, let $\rho_H=\dya{\psi_H}$ be an $n$-qubit Haar random state, where $\ket{\psi_H}$ was sampled according to the Haar measure over $\HC$. Then, we obtain that, in average
\begin{equation}
    \mathbb{E}_{\ket{\psi_H}}[\PC_{k}(\rho_H)]=\begin{cases}
        \frac{1}{2^n} &\text{if $k=0$}\\
        \frac{3^{k}\binom{n}{k}}{2^n(2^{n}+1)} &\text{if $k>0$}
    \end{cases}\,.
\end{equation}

Finally, we also note that when they exist, absolutely maximally entangled states~\cite{helwig2013absolutely,goyeneche2015absolutely,huber2017absolutely}, denoted as $\rho_{{\rm AME}}$ constitute the limiting case of GFD purities as they satisfy
\begin{equation}
    \PC_{k}(\rho_{{\rm AME}})=\begin{cases}
        \frac{1}{2^n} & \text{if $k=0$}\\
        0 & \text{if $0<k<n$}\\
        \frac{2^n-1}{2^n} & \text{if $k=n$}
    \end{cases}\,.
\end{equation}
That is, they maximize the  purity is concentrated in the largest irrep.  

\subsection{Fermionic Gaussianity}

Consider the irrep decomposition of Eq.~\eqref{eq:irrep-L-ferm}. We will consider only (spinless) fermionic states, i.e. states with well defined parity, corresponding to eigenstates of the parity operator $Z^{\otimes n}$, with $2n$ being the number of Majorana modes, and $n$ being the number of qubits after a Jordan-Wigner transformation. Furthermore, let us pick $n$ as a multiple of four, and restrict our analysis to the even parity sector (eigenvectors of $Z^{\otimes n}$ with eigenvalue $+1$). As a consequence, the GFD purities in the irreps $\LC_\a$ with $\a$ odd are automatically zero. The odd parity case can be recovered analogously.

Let us start by studying this QRT's free states. For a  pure Gaussian state $\rho_G$ we find
\begin{equation}\label{eq:pur-gauss}
\PC_\alpha(\rho_G)=\frac{\binom{n}{\alpha/2}}{2^n}\,.
\end{equation}

Next, we consider three types of non-Gaussian state. We start with the GHZ state that we already studied from the point of view of the entanglement QRT. Notice that $\ket{{\rm GHZ}}$ can be expressed as the linear combination of two Gaussian states, $\ket{0}^{\otimes n}$ and $\ket{1}^{\otimes n}$. Here we find 

\smaller
\begin{equation}\label{eq:ghz_fent}
    \PC_\alpha(\rho_{{\rm GHZ}})=\begin{cases}
        0 & \text{if $\a=2(2j-1)\neq n$ for 
        $j=1,\dots,\lfloor\frac{n}{2}\rfloor$}\\
        \frac{\binom{n}{\alpha/2}}{2^n} &\text{if $\alpha=4j \neq n$ for $j=0,\dots,\lfloor\frac{n}{2}\rfloor$}\\
        \frac{2^{n-1}+\binom{n}{n/2}}{2^n} & \text{if $\alpha=n$}
    \end{cases}\hspace{-0.25cm}\,.
\end{equation}
\normalsize

That is, for the GHZ state, the GFD purities exhibit a zig-zag pattern: They vanish whenever
$\alpha = 2(2j-1)$ (even but not a multiple of 4), take nonzero values when $\alpha$ is a multiple of four, and are mostly concentrated in the largest irrep $\LC_n$.

Next, we consider the parametrized state $\rho_\gamma=\dya{\chi(\gamma)}$ with 
\begin{equation}
    \ket{\chi(\gamma)}=(\cos(\gamma/4)\ket{0}^{\otimes 4}+\sin(\gamma/4)\ket{1}^{\otimes 4})^{\otimes n/4}\,,
\end{equation}
which interpolates between a Gaussian state at $\gamma=0$ and a highly non-Gaussian state at $\gamma=\pi$. We here find
\small
\begin{widetext}
\begin{align}
   \PC_\alpha(\rho_\gamma)=\frac{1}{2^n}\sum_{i_8=0}^{\lfloor \frac{\alpha}{8} \rfloor} \sum_{i_6=0}^{\lfloor \frac{\alpha}{6} \rfloor} \sum_{i_4=0}^{\lfloor \frac{\alpha}{4} \rfloor} \sum_{i_2=0}^{\lfloor \frac{\alpha}{2} \rfloor}
    \binom{\frac{n}{4}}{i_8} \binom{\frac{n}{4}-i_8}{i_6}  \binom{\frac{n}{4}-i_8-i_6}{i_4} \binom{\frac{n}{4}-i_8-i_6-i_4}{i_2}\,  N_2(\gamma)^{i_2}N_4(\gamma)^{i_4}N_6(\gamma)^{i_6}N_8^{i_8}\,\delta_{\a,8i_8+6i_6+4i_4+2i_2}\,.
\end{align}
\end{widetext}
\normalsize
where we defined 
\begin{align}
    N_2(\gamma)&=\binom{4}{1}\cos^2(\gamma/2)\,,\\
    N_4(\gamma)&=\left(\binom{4}{2}+\left(\binom{4}{2}+2\right)\sin^2(\gamma/2)\right)\,,\nonumber\\
    N_6(\gamma)&=\binom{4}{1}\cos^2(\gamma/2)\,, \quad N_8 =1\,.\nonumber
\end{align}
Finally, consider a Haar random state $\rho_H=\dya{\psi_H}$ with even fermionic parity. That is, $\ket{\psi_H}$ is sampled according to the Haar measure over $\HC_+$. We obtain, on average
\begin{equation}
    \mathbb{E}_{\ket{\phi}\sim\HC_+}[ \PC_\a(\dya{\phi})]\approx\begin{cases}
        0 & \text{if $\a$ is odd}\\
        \frac{1}{2^n} & \text{if $\a=0$}\\
        \frac{2\binom{2n}{\a}}{2^n(2^n+1)}  & \text{if $0<\a<2n$ is even} \\
        \frac{3}{2^n(2^n+1)}  & \text{if $\a=2n$}\
    \end{cases}
\end{equation}
Note that this result is not exact, as we have neglected corrections that are exponentially small in $n$ relative to the values shown.

\subsection{Spin coherence}

Here we begin by recalling that in this QRT, free states are those in the orbit of $\ket{s,s}$, while  all the other basis states $\ket{s,m}$ with $m\neq \pm s$ are resourceful. Defining $\rho_{m}=\dya{s,m}$,  we find 
\begin{align}\label{eq:pur-sm}
    \PC_{\alpha}(\rho_{m})=\left( c_{m,-m,0}^{s,s,\alpha}\right)^2\,,
\end{align}
where $c_{m_1,m_2,m}^{s_1,s_2,s}$ is a Clebsch-Gordan coefficient. For the free states one then explicitly finds
\begin{align}\label{eq:pur-ss}
    \PC_{\alpha}(\rho_{s})=\frac{(2 \a+1) ((2 s)!)^2}{(2 s-\a)! \Gamma (2 s+\a+2)}\,.
\end{align}

Next, we study again the GHZ state. We note that we can re-write it in terms of the basis states $\ket{s,m}$ as $\rho_{{\rm GHZ}}=\dya{{\rm GHZ}}$ with $
\ket{{\rm GHZ}}=\frac{1}{\sqrt{2}}(\ket{s,s}+\ket{s,-s})$, so that if $\alpha\neq 2s$  we get
\begin{align}\label{eq:pur-ghz-su2}
    \PC_{\alpha}(\rho_{{\rm GHZ}})&=\frac{1}{4}\left( c_{m,-m,0}^{s,s,\alpha}+(-1)^{2 s}  c_{m,-m,0}^{s,s,\alpha}\right)^2\nonumber\\
    &=\begin{cases}\left( c_{m,-m,0}^{s,s,\alpha}\right)^2\,  &   \text{ if $\alpha$ is even}\\
    0\, & \text{if $\alpha$ is odd}
    \end{cases}\,,
\end{align}
while for $\alpha=2s$ 
\begin{align}
    \PC_{2s}(\rho_{{\rm GHZ}})=&\frac{1}{4}\left( c_{s,-s,0}^{s,s,2s}+(-1)^{2 s}  c_{-s,s,0}^{s,s,2s}\right)^2\nonumber\\
    &+\frac{1}{4} (c_{-s,-s,-2s}^{s,s,2s})^2 + (c_{s,s,2s}^{s,s,2s})^2\,.
\end{align}
Finally, for a Haar random state $\rho_H=\dya{\psi_H}$,  where $\ket{\psi_H}$ is sampled according to the Haar measure over $\HC$, we obtain
\small
\begin{align}
    \mathcal{P}_{\alpha}(\rho_H) = 
    \begin{cases}
        \dfrac{1}{(2s + 1)(2s + 2)} + \dfrac{1}{2s + 2}\, & \text{if } \alpha = 0 \\
        \dfrac{2\alpha + 1}{(2s + 1)(2s + 2)}\, & \text{if } \alpha > 0
    \end{cases}
\end{align}

\normalsize

\subsection{Clifford stabilizerness}

As discussed in the main text, the GFD purities arising from the irreps of $\LC(\HC^{\otimes2})$ under the action of the Clifford group, whose definition is presented in Table~\ref{tab:irreps}, are all linearly dependent on the purity in the smallest non-trivial irrep $\LC_1$. As such, we only need to compute the purity therein. 
Let $\HC=(\mathbb{C}^2)^{\otimes n}$ be the Hilbert space of $n$ qubits, and let $\rho_{\rm stab}=\dya{{\rm stab}}$ be a generic free stabilizer state, which includes the GHZ state. For these states we find
\begin{equation}
    \LC_1(\rho_{\rm stab}^{\otimes 2})=\frac{1}{2^n}\,,
\end{equation}
Then, for a Haar random state $\rho_H=\dya{\psi_H}$,  where $\ket{\psi_H}$ is sampled according to the Haar measure over $\HC$ one finds~\cite{leone2022stabilizer} 
\begin{equation}
    \LC_1(\rho_H^{\otimes 2})=\frac{4}{2^n(2^n+3)}\,.
\end{equation}
Lastly, defining the single qubit magic state $\ket{M}=\frac{1}{\sqrt{2}}(\ket{0}+e^{i\pi/4}\ket{1})$, and the $n$-qubit state $\rho_M=\dya{M}^{\otimes n}$, one obtains 
\begin{equation}
    \LC_1(\rho_M^{\otimes 2})=\left(\frac{3}{16}\right)^n\,.
\end{equation}

\begin{table*}[t]
    \centering
\begin{tabular}{|c|c|c|}\hline
Irrep  & Operators in the basis of the irrep &Dimension \\\hline\hline
$\LC_{{\rm id}}$   & $\left\{\id\otimes \id\right\}$ &  $1$  \\\hline
$\LC_{ r}$   & $\left\{\sum_{\sigma \in \SC }\id\otimes \sigma\right\}$ &  $d^2-1$  \\\hline
$\LC_{ l}$   & $\left\{\sum_{\sigma \in \SC }\sigma\otimes \id\right\}$ &  $d^2-1$  \\\hline
$\LC_{[A]}$   & $\left\{\sigma\otimes \tau-\tau\otimes \sigma\,,\,\,\sigma\in\CC_\tau\,,\,\,\tau\in\SC\right\}$ &  $\frac{d^2-1}{2}\left(\frac{d^2}{2}-2\right)$  \\\hline
$\LC_{0}$   & $\left\{\sum_{\sigma \in \SC} \sigma\otimes\sigma\right\}$ &  $1$  \\\hline
$\LC_{1}$   & $\left\{\sum_{\sigma \in \SC }\lambda_\sigma\sigma\otimes\sigma\,|\,\sum_{\sigma \in \SC} \lambda_\sigma=0\,, \,\, \sum_{\sigma \in \NC_\tau}\lambda_\sigma=-\frac{d}{2}\lambda_\tau\,,\,\, \tau\in \SC\right\}$ &  $\frac{d(d+1)}{2}-1$  \\\hline
$\LC_{2}$   & $\left\{\sum_{\sigma \in \SC }\lambda_\sigma\sigma\otimes\sigma\,|\,\sum_{\sigma \in \SC} \lambda_\sigma=0\,, \,\, \sum_{\sigma \in \NC_\tau}\lambda_\sigma=\frac{d}{2}\lambda_\tau\,,\,\, \tau\in \SC\right\}$ &  $\frac{d(d-1)}{2}-1$  \\\hline
$\LC_{[{\rm adj}]}$   & ${\rm span}_{\mathbb{C}}\left\{\sum_{\sigma \in \CC_\tau }\sigma\otimes \sigma\cdot\tau+\sigma\cdot\tau\otimes \sigma\right\}$ &  $d^2-1$  \\\hline
$\LC_{\{{\rm adj}\}}$   & $\left\{\sum_{\sigma \in \CC_\tau }\sigma\otimes \sigma\cdot\tau-\sigma\cdot\tau\otimes \sigma\right\}$ &  $d^2-1$  \\\hline
$\LC_{\{{\rm adj}^\perp\}}$   & $\left\{A=\sigma\otimes \tau-\tau\otimes \sigma\,,\,\,\sigma\in\NC_\tau\,,\,\,\tau\in\SC\,|\,\Tr[B\ad A]=0\,,\forall B\in\LC_{\{{\rm adj}\}} \right\}$ &  $(d^2-1)\left(\frac{d^2}{2}-2\right)$  \\\hline
$\LC_{[1]}$   & $\left\{\sum_{\sigma \in \NC_\tau }\lambda_\sigma\sigma\otimes\sigma\,|\,\sum_{\sigma \in \CC_\tau\cap\NC_\pi} \lambda_\sigma=-d\lambda_\pi\,, \,\, \pi\in\CC_\tau\,,\,\, \tau\in \SC\right\}$ &  $(d^2-1)\left(\frac{d(d+2)}{8}-1\right)$  \\\hline
$\LC_{[2]}$   & $\left\{\sum_{\sigma \in \NC_\tau }\lambda_\sigma\sigma\otimes\sigma\,|\,\sum_{\sigma \in \NC_\tau\cap\NC_\pi} \lambda_\sigma=d\lambda_\pi\,, \,\, \pi\in\CC_\tau\,,\,\, \tau\in \SC\right\}$ &  $(d^2-1)\left(\frac{d(d-2)}{8}-1\right)$  \\\hline
$\LC_{\{1\}}$   & $\left\{\sum_{\sigma \in \NC_\tau }\lambda_\sigma(\sigma\otimes \sigma\cdot\tau+\sigma\cdot\tau\otimes \sigma)\,|\,\sum_{\sigma \in \NC_\tau\cap\CC_\pi} \lambda_\sigma-\sum_{\sigma \in \NC_\tau\cap\NC_\pi} \lambda_\sigma=\frac{d}{2}\lambda_\pi\,, \,\, \pi\in\NC_\tau\,,\,\, \tau\in \SC\right\}$ &  $(d^2-1)\left(\frac{d(d+2)}{8}\right)$  \\\hline
$\LC_{\{2\}}$   & $\left\{\sum_{\sigma \in \NC_\tau }\lambda_\sigma(\sigma\otimes \sigma\cdot\tau+\sigma\cdot\tau\otimes \sigma)\,|\,\sum_{\sigma \in \NC_\tau\cap\CC_\pi} \lambda_\sigma-\sum_{\sigma \in \NC_\tau\cap\NC_\pi} \lambda_\sigma=-\frac{d}{2}\lambda_\pi\,, \,\, \pi\in\NC_\tau\,,\,\, \tau\in \SC\right\}$ &  $(d^2-1)\left(\frac{d(d-2)}{8}\right)$  \\\hline
\end{tabular}
    \caption{Irreps for $\LC(\HC^{\otimes 2})$ under the action of the $n$-qubit Clifford group~\cite{helsen2018representations}. 
    Here we have defined $\SC=\frac{1}{\sqrt{d}}\{\id_2,X,Y,Z\}^{\otimes n}\backslash\{\id^{\otimes n }\}$, $\NC_\tau=\{\sigma \in \SC\,|\,\{\sigma,\tau=0\}\}$, $\CC_\tau=\{\sigma \in \SC\,|\,[\sigma,\tau]=0\}$ }
    \label{tab:irreps}
\end{table*}

\section{Acknowledgments}

The authors thank Maria Schuld, David Wierichs, Nathan Killoran, Diego Garc\'ia-Martin, and Luke Coffman for useful and extremely insightful discussions. In particular, MC would like to thank Maria Schuld for her invaluable feedback on the earlier versions of the manuscript. P. Bermejo and P. Braccia and NLD were supported by the Laboratory Directed Research and Development (LDRD) program of Los Alamos National Laboratory (LANL) under project numbers 20230527ECR and 20230049DR. P. Bermejo also acknowledges constant support from DIPC.  AAM acknowledges support by the U.S. DOE through a quantum computing program sponsored by the LANL Information Science \& Technology Institute. AAM also acknowledges support by the German Federal Ministry for Education and Research (BMBF) under the project FermiQP. NLD also acknowledges support by  the Center for Nonlinear Studies at LANL. NLD was also initially supported by UNLP and CONICET  Argentina. AED, ML and MC were supported by LANL ASC Beyond Moore’s Law project.  This material is based upon work supported by the U.S. Department of Energy, Office of Science, National Quantum Information Science Research Centers, Quantum Science Center (LC). Research presented in this article was supported by the National Security Education Center (NSEC) Informational Science and Technology Institute (ISTI) using the LDRD program of LANL project number 20240479CR-IST.

\bibliography{quantum}

\begin{thebibliography}{116}%
\makeatletter
\providecommand \@ifxundefined [1]{%
 \@ifx{#1\undefined}
}%
\providecommand \@ifnum [1]{%
 \ifnum #1\expandafter \@firstoftwo
 \else \expandafter \@secondoftwo
 \fi
}%
\providecommand \@ifx [1]{%
 \ifx #1\expandafter \@firstoftwo
 \else \expandafter \@secondoftwo
 \fi
}%
\providecommand \natexlab [1]{#1}%
\providecommand \enquote  [1]{``#1''}%
\providecommand \bibnamefont  [1]{#1}%
\providecommand \bibfnamefont [1]{#1}%
\providecommand \citenamefont [1]{#1}%
\providecommand \href@noop [0]{\@secondoftwo}%
\providecommand \href [0]{\begingroup \@sanitize@url \@href}%
\providecommand \@href[1]{\@@startlink{#1}\@@href}%
\providecommand \@@href[1]{\endgroup#1\@@endlink}%
\providecommand \@sanitize@url [0]{\catcode `\\12\catcode `\$12\catcode `\&12\catcode `\#12\catcode `\^12\catcode `\_12\catcode `\%12\relax}%
\providecommand \@@startlink[1]{}%
\providecommand \@@endlink[0]{}%
\providecommand \url  [0]{\begingroup\@sanitize@url \@url }%
\providecommand \@url [1]{\endgroup\@href {#1}{\urlprefix }}%
\providecommand \urlprefix  [0]{URL }%
\providecommand \Eprint [0]{\href }%
\providecommand \doibase [0]{https://doi.org/}%
\providecommand \selectlanguage [0]{\@gobble}%
\providecommand \bibinfo  [0]{\@secondoftwo}%
\providecommand \bibfield  [0]{\@secondoftwo}%
\providecommand \translation [1]{[#1]}%
\providecommand \BibitemOpen [0]{}%
\providecommand \bibitemStop [0]{}%
\providecommand \bibitemNoStop [0]{.\EOS\space}%
\providecommand \EOS [0]{\spacefactor3000\relax}%
\providecommand \BibitemShut  [1]{\csname bibitem#1\endcsname}%
\let\auto@bib@innerbib\@empty
\bibitem [{\citenamefont {Katznelson}(2004)}]{katznelson2004introduction}%
  \BibitemOpen
  \bibfield  {author} {\bibinfo {author} {\bibfnamefont {Y.}~\bibnamefont {Katznelson}},\ }\href {https://www.cambridge.org/core/books/an-introduction-to-harmonic-analysis/67C4CE356E7420BA17F3F1337291EF82} {\emph {\bibinfo {title} {An introduction to harmonic analysis}}}\ (\bibinfo  {publisher} {Cambridge University Press},\ \bibinfo {year} {2004})\BibitemShut {NoStop}%
\bibitem [{\citenamefont {Mackey}(1980)}]{mackey1980harmonic}%
  \BibitemOpen
  \bibfield  {author} {\bibinfo {author} {\bibfnamefont {G.~W.}\ \bibnamefont {Mackey}},\ }\bibfield  {title} {\bibinfo {title} {Harmonic analysis as the exploitation of symmetry--a historical survey},\ }\href {https://doi.org/10.1090/S0273-0979-1980-14783-7} {\bibfield  {journal} {\bibinfo  {journal} {Bulletin (New Series) of the American Mathematical Society}\ }\textbf {\bibinfo {volume} {3}},\ \bibinfo {pages} {543 } (\bibinfo {year} {1980})}\BibitemShut {NoStop}%
\bibitem [{\citenamefont {Chirikjian}\ and\ \citenamefont {Kyatkin}(2000)}]{chirikjian2000engineering}%
  \BibitemOpen
  \bibfield  {author} {\bibinfo {author} {\bibfnamefont {G.~S.}\ \bibnamefont {Chirikjian}}\ and\ \bibinfo {author} {\bibfnamefont {A.~B.}\ \bibnamefont {Kyatkin}},\ }\href@noop {} {\emph {\bibinfo {title} {Engineering applications of noncommutative harmonic analysis: with emphasis on rotation and motion groups}}}\ (\bibinfo  {publisher} {CRC press},\ \bibinfo {year} {2000})\BibitemShut {NoStop}%
\bibitem [{\citenamefont {Bronstein}\ \emph {et~al.}(2021)\citenamefont {Bronstein}, \citenamefont {Bruna}, \citenamefont {Cohen},\ and\ \citenamefont {Veli{\v{c}}kovi{\'c}}}]{bronstein2021geometric}%
  \BibitemOpen
  \bibfield  {author} {\bibinfo {author} {\bibfnamefont {M.~M.}\ \bibnamefont {Bronstein}}, \bibinfo {author} {\bibfnamefont {J.}~\bibnamefont {Bruna}}, \bibinfo {author} {\bibfnamefont {T.}~\bibnamefont {Cohen}},\ and\ \bibinfo {author} {\bibfnamefont {P.}~\bibnamefont {Veli{\v{c}}kovi{\'c}}},\ }\bibfield  {title} {\bibinfo {title} {Geometric deep learning: Grids, groups, graphs, geodesics, and gauges},\ }\href {https://arxiv.org/abs/2104.13478} {\bibfield  {journal} {\bibinfo  {journal} {arXiv preprint arXiv:2104.13478}\ } (\bibinfo {year} {2021})}\BibitemShut {NoStop}%
\bibitem [{\citenamefont {Cohen}\ \emph {et~al.}(2018)\citenamefont {Cohen}, \citenamefont {Geiger}, \citenamefont {K{\"o}hler},\ and\ \citenamefont {Welling}}]{cohen2018spherical}%
  \BibitemOpen
  \bibfield  {author} {\bibinfo {author} {\bibfnamefont {T.~S.}\ \bibnamefont {Cohen}}, \bibinfo {author} {\bibfnamefont {M.}~\bibnamefont {Geiger}}, \bibinfo {author} {\bibfnamefont {J.}~\bibnamefont {K{\"o}hler}},\ and\ \bibinfo {author} {\bibfnamefont {M.}~\bibnamefont {Welling}},\ }\bibfield  {title} {\bibinfo {title} {Spherical cnns},\ }\href {https://arxiv.org/abs/1801.10130} {\bibfield  {journal} {\bibinfo  {journal} {arXiv preprint arXiv:1801.10130}\ } (\bibinfo {year} {2018})}\BibitemShut {NoStop}%
\bibitem [{\citenamefont {Altschuler}\ \emph {et~al.}(1975)\citenamefont {Altschuler}, \citenamefont {Trotter}, \citenamefont {Newkirk},\ and\ \citenamefont {Howard}}]{altschuler1975tabulation}%
  \BibitemOpen
  \bibfield  {author} {\bibinfo {author} {\bibfnamefont {M.~D.}\ \bibnamefont {Altschuler}}, \bibinfo {author} {\bibfnamefont {D.~E.}\ \bibnamefont {Trotter}}, \bibinfo {author} {\bibfnamefont {G.}~\bibnamefont {Newkirk}},\ and\ \bibinfo {author} {\bibfnamefont {R.}~\bibnamefont {Howard}},\ }\bibfield  {title} {\bibinfo {title} {Tabulation of the harmonic coefficients of the solar magnetic fields},\ }\href {https://doi.org/10.1007/BF00152968} {\bibfield  {journal} {\bibinfo  {journal} {Solar Physics}\ }\textbf {\bibinfo {volume} {41}},\ \bibinfo {pages} {225} (\bibinfo {year} {1975})}\BibitemShut {NoStop}%
\bibitem [{\citenamefont {Whaler}\ and\ \citenamefont {Gubbins}(1981)}]{whaler1981spherical}%
  \BibitemOpen
  \bibfield  {author} {\bibinfo {author} {\bibfnamefont {K.}~\bibnamefont {Whaler}}\ and\ \bibinfo {author} {\bibfnamefont {D.}~\bibnamefont {Gubbins}},\ }\bibfield  {title} {\bibinfo {title} {Spherical harmonic analysis of the geomagnetic field: an example of a linear inverse problem},\ }\href {https://doi.org/10.1111/j.1365-246X.1981.tb04877.x} {\bibfield  {journal} {\bibinfo  {journal} {Geophysical Journal International}\ }\textbf {\bibinfo {volume} {65}},\ \bibinfo {pages} {645} (\bibinfo {year} {1981})}\BibitemShut {NoStop}%
\bibitem [{\citenamefont {Knaack}\ and\ \citenamefont {Stenflo}(2005)}]{knaack2005spherical}%
  \BibitemOpen
  \bibfield  {author} {\bibinfo {author} {\bibfnamefont {R.}~\bibnamefont {Knaack}}\ and\ \bibinfo {author} {\bibfnamefont {J.~O.}\ \bibnamefont {Stenflo}},\ }\bibfield  {title} {\bibinfo {title} {Spherical harmonic decomposition of solar magnetic fields},\ }\href {https://doi.org/10.1051/0004-6361:20052765} {\bibfield  {journal} {\bibinfo  {journal} {Astronomy \& Astrophysics}\ }\textbf {\bibinfo {volume} {438}},\ \bibinfo {pages} {349} (\bibinfo {year} {2005})}\BibitemShut {NoStop}%
\bibitem [{\citenamefont {Marks}(2009)}]{marks2009handbook}%
  \BibitemOpen
  \bibfield  {author} {\bibinfo {author} {\bibfnamefont {R.~J.}\ \bibnamefont {Marks}},\ }\href@noop {} {\emph {\bibinfo {title} {Handbook of Fourier analysis \& its applications}}}\ (\bibinfo  {publisher} {Oxford University Press},\ \bibinfo {year} {2009})\BibitemShut {NoStop}%
\bibitem [{\citenamefont {Chitambar}\ and\ \citenamefont {Gour}(2019)}]{chitambar2019quantum}%
  \BibitemOpen
  \bibfield  {author} {\bibinfo {author} {\bibfnamefont {E.}~\bibnamefont {Chitambar}}\ and\ \bibinfo {author} {\bibfnamefont {G.}~\bibnamefont {Gour}},\ }\bibfield  {title} {\bibinfo {title} {Quantum resource theories},\ }\href {https://doi.org/10.1103/RevModPhys.91.025001} {\bibfield  {journal} {\bibinfo  {journal} {Reviews of modern physics}\ }\textbf {\bibinfo {volume} {91}},\ \bibinfo {pages} {025001} (\bibinfo {year} {2019})}\BibitemShut {NoStop}%
\bibitem [{\citenamefont {Gu}(1985)}]{gu1985group}%
  \BibitemOpen
  \bibfield  {author} {\bibinfo {author} {\bibfnamefont {Y.}~\bibnamefont {Gu}},\ }\bibfield  {title} {\bibinfo {title} {Group-theoretical formalism of quantum mechanics based on quantum generalization of characteristic functions},\ }\href {https://doi.org/10.1103/PhysRevA.32.1310} {\bibfield  {journal} {\bibinfo  {journal} {Physical Review A}\ }\textbf {\bibinfo {volume} {32}},\ \bibinfo {pages} {1310} (\bibinfo {year} {1985})}\BibitemShut {NoStop}%
\bibitem [{\citenamefont {Korbicz}\ and\ \citenamefont {Lewenstein}(2006)}]{korbicz2006group}%
  \BibitemOpen
  \bibfield  {author} {\bibinfo {author} {\bibfnamefont {J.}~\bibnamefont {Korbicz}}\ and\ \bibinfo {author} {\bibfnamefont {M.}~\bibnamefont {Lewenstein}},\ }\bibfield  {title} {\bibinfo {title} {Group-theoretical approach to entanglement},\ }\href {https://doi.org/10.1103/PhysRevA.74.022318} {\bibfield  {journal} {\bibinfo  {journal} {Physical Review A—Atomic, Molecular, and Optical Physics}\ }\textbf {\bibinfo {volume} {74}},\ \bibinfo {pages} {022318} (\bibinfo {year} {2006})}\BibitemShut {NoStop}%
\bibitem [{\citenamefont {Korbicz}\ \emph {et~al.}(2008)\citenamefont {Korbicz}, \citenamefont {Wehr},\ and\ \citenamefont {Lewenstein}}]{korbicz2008entanglement}%
  \BibitemOpen
  \bibfield  {author} {\bibinfo {author} {\bibfnamefont {J.}~\bibnamefont {Korbicz}}, \bibinfo {author} {\bibfnamefont {J.}~\bibnamefont {Wehr}},\ and\ \bibinfo {author} {\bibfnamefont {M.}~\bibnamefont {Lewenstein}},\ }\bibfield  {title} {\bibinfo {title} {Entanglement of positive definite functions on compact groups},\ }\href {https://doi.org/10.1007/s00220-008-0493-6} {\bibfield  {journal} {\bibinfo  {journal} {Communications in mathematical physics}\ }\textbf {\bibinfo {volume} {281}},\ \bibinfo {pages} {753} (\bibinfo {year} {2008})}\BibitemShut {NoStop}%
\bibitem [{\citenamefont {Marvian}\ and\ \citenamefont {Spekkens}(2013)}]{marvian2013theory}%
  \BibitemOpen
  \bibfield  {author} {\bibinfo {author} {\bibfnamefont {I.}~\bibnamefont {Marvian}}\ and\ \bibinfo {author} {\bibfnamefont {R.~W.}\ \bibnamefont {Spekkens}},\ }\bibfield  {title} {\bibinfo {title} {The theory of manipulations of pure state asymmetry: I. basic tools, equivalence classes and single copy transformations},\ }\href {https://doi.org/10.1088/1367-2630/15/3/033001} {\bibfield  {journal} {\bibinfo  {journal} {New Journal of Physics}\ }\textbf {\bibinfo {volume} {15}},\ \bibinfo {pages} {033001} (\bibinfo {year} {2013})}\BibitemShut {NoStop}%
\bibitem [{\citenamefont {Barnum}\ \emph {et~al.}(2003)\citenamefont {Barnum}, \citenamefont {Knill}, \citenamefont {Ortiz},\ and\ \citenamefont {Viola}}]{barnum2003generalizations}%
  \BibitemOpen
  \bibfield  {author} {\bibinfo {author} {\bibfnamefont {H.}~\bibnamefont {Barnum}}, \bibinfo {author} {\bibfnamefont {E.}~\bibnamefont {Knill}}, \bibinfo {author} {\bibfnamefont {G.}~\bibnamefont {Ortiz}},\ and\ \bibinfo {author} {\bibfnamefont {L.}~\bibnamefont {Viola}},\ }\bibfield  {title} {\bibinfo {title} {Generalizations of entanglement based on coherent states and convex sets},\ }\href {https://doi.org/10.1103/PhysRevA.68.032308} {\bibfield  {journal} {\bibinfo  {journal} {Physical Review A}\ }\textbf {\bibinfo {volume} {68}},\ \bibinfo {pages} {032308} (\bibinfo {year} {2003})}\BibitemShut {NoStop}%
\bibitem [{\citenamefont {Barnum}\ \emph {et~al.}(2004)\citenamefont {Barnum}, \citenamefont {Knill}, \citenamefont {Ortiz}, \citenamefont {Somma},\ and\ \citenamefont {Viola}}]{barnum2004subsystem}%
  \BibitemOpen
  \bibfield  {author} {\bibinfo {author} {\bibfnamefont {H.}~\bibnamefont {Barnum}}, \bibinfo {author} {\bibfnamefont {E.}~\bibnamefont {Knill}}, \bibinfo {author} {\bibfnamefont {G.}~\bibnamefont {Ortiz}}, \bibinfo {author} {\bibfnamefont {R.}~\bibnamefont {Somma}},\ and\ \bibinfo {author} {\bibfnamefont {L.}~\bibnamefont {Viola}},\ }\bibfield  {title} {\bibinfo {title} {A subsystem-independent generalization of entanglement},\ }\href {https://doi.org/10.1103/PhysRevLett.92.107902} {\bibfield  {journal} {\bibinfo  {journal} {Physical Review Letters}\ }\textbf {\bibinfo {volume} {92}},\ \bibinfo {pages} {107902} (\bibinfo {year} {2004})}\BibitemShut {NoStop}%
\bibitem [{\citenamefont {Klyachko}(2002)}]{klyachko2002coherent}%
  \BibitemOpen
  \bibfield  {author} {\bibinfo {author} {\bibfnamefont {A.}~\bibnamefont {Klyachko}},\ }\bibfield  {title} {\bibinfo {title} {Coherent states, entanglement, and geometric invariant theory},\ }\bibfield  {journal} {\bibinfo  {journal} {arXiv preprint quant-ph/0206012}\ }\href {https://doi.org/10.48550/arXiv.quant-ph/0206012} {10.48550/arXiv.quant-ph/0206012} (\bibinfo {year} {2002})\BibitemShut {NoStop}%
\bibitem [{\citenamefont {Delbourgo}\ and\ \citenamefont {Fox}(1977)}]{delbourgo1977maximum}%
  \BibitemOpen
  \bibfield  {author} {\bibinfo {author} {\bibfnamefont {R.}~\bibnamefont {Delbourgo}}\ and\ \bibinfo {author} {\bibfnamefont {J.}~\bibnamefont {Fox}},\ }\bibfield  {title} {\bibinfo {title} {Maximum weight vectors possess minimal uncertainty},\ }\href {https://doi.org/10.1088/0305-4470/10/12/004} {\bibfield  {journal} {\bibinfo  {journal} {Journal of Physics A: Mathematical and General}\ }\textbf {\bibinfo {volume} {10}},\ \bibinfo {pages} {L233} (\bibinfo {year} {1977})}\BibitemShut {NoStop}%
\bibitem [{\citenamefont {Meyer}\ and\ \citenamefont {Wallach}(2002)}]{meyer2002global}%
  \BibitemOpen
  \bibfield  {author} {\bibinfo {author} {\bibfnamefont {D.~A.}\ \bibnamefont {Meyer}}\ and\ \bibinfo {author} {\bibfnamefont {N.~R.}\ \bibnamefont {Wallach}},\ }\bibfield  {title} {\bibinfo {title} {Global entanglement in multiparticle systems},\ }\href {https://doi.org/10.1063/1.1497700} {\bibfield  {journal} {\bibinfo  {journal} {Journal of Mathematical Physics}\ }\textbf {\bibinfo {volume} {43}},\ \bibinfo {pages} {4273} (\bibinfo {year} {2002})}\BibitemShut {NoStop}%
\bibitem [{\citenamefont {Brennen}(2003)}]{brennen2003observable}%
  \BibitemOpen
  \bibfield  {author} {\bibinfo {author} {\bibfnamefont {G.~K.}\ \bibnamefont {Brennen}},\ }\bibfield  {title} {\bibinfo {title} {An observable measure of entanglement for pure states of multi-qubit systems},\ }\href {https://arxiv.org/abs/quant-ph/0305094} {\bibfield  {journal} {\bibinfo  {journal} {arXiv preprint quant-ph/0305094}\ } (\bibinfo {year} {2003})}\BibitemShut {NoStop}%
\bibitem [{\citenamefont {Beckey}\ \emph {et~al.}(2021)\citenamefont {Beckey}, \citenamefont {Gigena}, \citenamefont {Coles},\ and\ \citenamefont {Cerezo}}]{beckey2021computable}%
  \BibitemOpen
  \bibfield  {author} {\bibinfo {author} {\bibfnamefont {J.~L.}\ \bibnamefont {Beckey}}, \bibinfo {author} {\bibfnamefont {N.}~\bibnamefont {Gigena}}, \bibinfo {author} {\bibfnamefont {P.~J.}\ \bibnamefont {Coles}},\ and\ \bibinfo {author} {\bibfnamefont {M.}~\bibnamefont {Cerezo}},\ }\bibfield  {title} {\bibinfo {title} {Computable and operationally meaningful multipartite entanglement measures},\ }\href {https://doi.org/10.1103/PhysRevLett.127.140501} {\bibfield  {journal} {\bibinfo  {journal} {Phys. Rev. Lett.}\ }\textbf {\bibinfo {volume} {127}},\ \bibinfo {pages} {140501} (\bibinfo {year} {2021})}\BibitemShut {NoStop}%
\bibitem [{\citenamefont {Schatzki}\ \emph {et~al.}(2024)\citenamefont {Schatzki}, \citenamefont {Liu}, \citenamefont {Cerezo},\ and\ \citenamefont {Chitambar}}]{schatzki2022hierarchy}%
  \BibitemOpen
  \bibfield  {author} {\bibinfo {author} {\bibfnamefont {L.}~\bibnamefont {Schatzki}}, \bibinfo {author} {\bibfnamefont {G.}~\bibnamefont {Liu}}, \bibinfo {author} {\bibfnamefont {M.}~\bibnamefont {Cerezo}},\ and\ \bibinfo {author} {\bibfnamefont {E.}~\bibnamefont {Chitambar}},\ }\bibfield  {title} {\bibinfo {title} {A hierarchy of multipartite correlations based on concentratable entanglement},\ }\href {https://doi.org/10.1103/PhysRevResearch.6.023019} {\bibfield  {journal} {\bibinfo  {journal} {Physical Review Research}\ }\textbf {\bibinfo {volume} {6}},\ \bibinfo {pages} {023019} (\bibinfo {year} {2024})}\BibitemShut {NoStop}%
\bibitem [{\citenamefont {Balachandran}\ \emph {et~al.}(2013)\citenamefont {Balachandran}, \citenamefont {Govindarajan}, \citenamefont {de~Queiroz},\ and\ \citenamefont {Reyes-Lega}}]{balachandran2013entanglement}%
  \BibitemOpen
  \bibfield  {author} {\bibinfo {author} {\bibfnamefont {A.~P.}\ \bibnamefont {Balachandran}}, \bibinfo {author} {\bibfnamefont {T.~R.}\ \bibnamefont {Govindarajan}}, \bibinfo {author} {\bibfnamefont {A.~R.}\ \bibnamefont {de~Queiroz}},\ and\ \bibinfo {author} {\bibfnamefont {A.~F.}\ \bibnamefont {Reyes-Lega}},\ }\bibfield  {title} {\bibinfo {title} {Entanglement and particle identity: a unifying approach},\ }\href {https://doi.org/10.1103/PhysRevLett.110.080503} {\bibfield  {journal} {\bibinfo  {journal} {Physical review letters}\ }\textbf {\bibinfo {volume} {110}},\ \bibinfo {pages} {080503} (\bibinfo {year} {2013})}\BibitemShut {NoStop}%
\bibitem [{\citenamefont {Harshman}\ and\ \citenamefont {Ranade}(2011)}]{harshman2011observables}%
  \BibitemOpen
  \bibfield  {author} {\bibinfo {author} {\bibfnamefont {N.~L.}\ \bibnamefont {Harshman}}\ and\ \bibinfo {author} {\bibfnamefont {K.~S.}\ \bibnamefont {Ranade}},\ }\bibfield  {title} {\bibinfo {title} {Observables can be tailored to change the entanglement of any pure state},\ }\href {https://doi.org/10.1103/PhysRevA.84.012303} {\bibfield  {journal} {\bibinfo  {journal} {Physical Review A—Atomic, Molecular, and Optical Physics}\ }\textbf {\bibinfo {volume} {84}},\ \bibinfo {pages} {012303} (\bibinfo {year} {2011})}\BibitemShut {NoStop}%
\bibitem [{\citenamefont {Zanardi}(2001)}]{zanardi2001virtual}%
  \BibitemOpen
  \bibfield  {author} {\bibinfo {author} {\bibfnamefont {P.}~\bibnamefont {Zanardi}},\ }\bibfield  {title} {\bibinfo {title} {Virtual quantum subsystems},\ }\href {https://doi.org/10.1103/PhysRevLett.87.077901} {\bibfield  {journal} {\bibinfo  {journal} {Physical Review Letters}\ }\textbf {\bibinfo {volume} {87}},\ \bibinfo {pages} {077901} (\bibinfo {year} {2001})}\BibitemShut {NoStop}%
\bibitem [{\citenamefont {Zanardi}\ \emph {et~al.}(2004)\citenamefont {Zanardi}, \citenamefont {Lidar},\ and\ \citenamefont {Lloyd}}]{zanardi2004quantum}%
  \BibitemOpen
  \bibfield  {author} {\bibinfo {author} {\bibfnamefont {P.}~\bibnamefont {Zanardi}}, \bibinfo {author} {\bibfnamefont {D.~A.}\ \bibnamefont {Lidar}},\ and\ \bibinfo {author} {\bibfnamefont {S.}~\bibnamefont {Lloyd}},\ }\bibfield  {title} {\bibinfo {title} {Quantum tensor product structures are observable induced},\ }\href {https://doi.org/10.1103/PhysRevLett.92.060402} {\bibfield  {journal} {\bibinfo  {journal} {Physical review letters}\ }\textbf {\bibinfo {volume} {92}},\ \bibinfo {pages} {060402} (\bibinfo {year} {2004})}\BibitemShut {NoStop}%
\bibitem [{\citenamefont {Nha}\ and\ \citenamefont {Kim}(2006)}]{nha2006entanglement}%
  \BibitemOpen
  \bibfield  {author} {\bibinfo {author} {\bibfnamefont {H.}~\bibnamefont {Nha}}\ and\ \bibinfo {author} {\bibfnamefont {J.}~\bibnamefont {Kim}},\ }\bibfield  {title} {\bibinfo {title} {Entanglement criteria via the uncertainty relations in su (2) and su (1, 1) algebras: Detection of non-gaussian entangled states},\ }\href {https://doi.org/10.1103/PhysRevA.74.012317} {\bibfield  {journal} {\bibinfo  {journal} {Physical Review A—Atomic, Molecular, and Optical Physics}\ }\textbf {\bibinfo {volume} {74}},\ \bibinfo {pages} {012317} (\bibinfo {year} {2006})}\BibitemShut {NoStop}%
\bibitem [{\citenamefont {Alicki}\ \emph {et~al.}(2009)\citenamefont {Alicki}, \citenamefont {Fannes},\ and\ \citenamefont {Pogorzelska}}]{alicki2009quantum}%
  \BibitemOpen
  \bibfield  {author} {\bibinfo {author} {\bibfnamefont {R.}~\bibnamefont {Alicki}}, \bibinfo {author} {\bibfnamefont {M.}~\bibnamefont {Fannes}},\ and\ \bibinfo {author} {\bibfnamefont {M.}~\bibnamefont {Pogorzelska}},\ }\bibfield  {title} {\bibinfo {title} {Quantum generalized subsystems},\ }\href {https://doi.org/10.1103/PhysRevA.79.052111} {\bibfield  {journal} {\bibinfo  {journal} {Physical Review A—Atomic, Molecular, and Optical Physics}\ }\textbf {\bibinfo {volume} {79}},\ \bibinfo {pages} {052111} (\bibinfo {year} {2009})}\BibitemShut {NoStop}%
\bibitem [{\citenamefont {Viola}\ and\ \citenamefont {Barnum}(2010)}]{viola2010entanglement}%
  \BibitemOpen
  \bibfield  {author} {\bibinfo {author} {\bibfnamefont {L.}~\bibnamefont {Viola}}\ and\ \bibinfo {author} {\bibfnamefont {H.}~\bibnamefont {Barnum}},\ }\bibfield  {title} {\bibinfo {title} {Entanglement and subsystems, entanglement beyond subsystems, and all that},\ }in\ \href {https://www.cambridge.org/core/books/abs/philosophy-of-quantum-information-and-entanglement/entanglement-and-subsystems-entanglement-beyond-subsystems-and-all-that/F25C06513B32E0A676E2B4528461D08B} {\emph {\bibinfo {booktitle} {Philosophy of quantum information and entanglement}}}\ (\bibinfo  {publisher} {Cambridge University Press Cambridge, UK},\ \bibinfo {year} {2010})\ pp.\ \bibinfo {pages} {16--43}\BibitemShut {NoStop}%
\bibitem [{\citenamefont {Derkacz}\ \emph {et~al.}(2011)\citenamefont {Derkacz}, \citenamefont {Gw{\'o}{\'z}d{\'z}},\ and\ \citenamefont {Jak{\'o}bczyk}}]{derkacz2011entanglement}%
  \BibitemOpen
  \bibfield  {author} {\bibinfo {author} {\bibfnamefont {{\L}.}~\bibnamefont {Derkacz}}, \bibinfo {author} {\bibfnamefont {M.}~\bibnamefont {Gw{\'o}{\'z}d{\'z}}},\ and\ \bibinfo {author} {\bibfnamefont {L.}~\bibnamefont {Jak{\'o}bczyk}},\ }\bibfield  {title} {\bibinfo {title} {Entanglement beyond tensor product structure: algebraic aspects of quantum non-separability},\ }\href {https://doi.org/10.1103/PhysRevA.70.042311} {\bibfield  {journal} {\bibinfo  {journal} {Journal of Physics A: Mathematical and Theoretical}\ }\textbf {\bibinfo {volume} {45}},\ \bibinfo {pages} {025302} (\bibinfo {year} {2011})}\BibitemShut {NoStop}%
\bibitem [{\citenamefont {Gigena}\ and\ \citenamefont {Rossignoli}(2015)}]{gigena2015entanglement}%
  \BibitemOpen
  \bibfield  {author} {\bibinfo {author} {\bibfnamefont {N.}~\bibnamefont {Gigena}}\ and\ \bibinfo {author} {\bibfnamefont {R.}~\bibnamefont {Rossignoli}},\ }\bibfield  {title} {\bibinfo {title} {Entanglement in fermion systems},\ }\href {https://doi.org/10.1103/PhysRevA.92.042326} {\bibfield  {journal} {\bibinfo  {journal} {Physical Review A}\ }\textbf {\bibinfo {volume} {92}},\ \bibinfo {pages} {042326} (\bibinfo {year} {2015})}\BibitemShut {NoStop}%
\bibitem [{\citenamefont {Benatti}\ and\ \citenamefont {Floreanini}(2016)}]{benatti2016entanglement}%
  \BibitemOpen
  \bibfield  {author} {\bibinfo {author} {\bibfnamefont {F.}~\bibnamefont {Benatti}}\ and\ \bibinfo {author} {\bibfnamefont {R.}~\bibnamefont {Floreanini}},\ }\bibfield  {title} {\bibinfo {title} {Entanglement in algebraic quantum mechanics: Majorana fermion systems},\ }\href {https://doi.org/10.1088/1751-8113/49/30/305303} {\bibfield  {journal} {\bibinfo  {journal} {Journal of Physics A: Mathematical and Theoretical}\ }\textbf {\bibinfo {volume} {49}},\ \bibinfo {pages} {305303} (\bibinfo {year} {2016})}\BibitemShut {NoStop}%
\bibitem [{\citenamefont {Regula}(2017)}]{regula2017convex}%
  \BibitemOpen
  \bibfield  {author} {\bibinfo {author} {\bibfnamefont {B.}~\bibnamefont {Regula}},\ }\bibfield  {title} {\bibinfo {title} {Convex geometry of quantum resource quantification},\ }\href {https://doi.org/10.1088/1751-8121/aa9100} {\bibfield  {journal} {\bibinfo  {journal} {Journal of Physics A: Mathematical and Theoretical}\ }\textbf {\bibinfo {volume} {51}},\ \bibinfo {pages} {045303} (\bibinfo {year} {2017})}\BibitemShut {NoStop}%
\bibitem [{\citenamefont {Sindici}\ and\ \citenamefont {Piani}(2018)}]{sindici2018simple}%
  \BibitemOpen
  \bibfield  {author} {\bibinfo {author} {\bibfnamefont {E.}~\bibnamefont {Sindici}}\ and\ \bibinfo {author} {\bibfnamefont {M.}~\bibnamefont {Piani}},\ }\bibfield  {title} {\bibinfo {title} {Simple class of bound entangled states based on the properties of the antisymmetric subspace},\ }\href {https://doi.org/10.1103/PhysRevA.97.032319} {\bibfield  {journal} {\bibinfo  {journal} {Physical Review A}\ }\textbf {\bibinfo {volume} {97}},\ \bibinfo {pages} {032319} (\bibinfo {year} {2018})}\BibitemShut {NoStop}%
\bibitem [{\citenamefont {Gigena}\ \emph {et~al.}(2021)\citenamefont {Gigena}, \citenamefont {Di~Tullio},\ and\ \citenamefont {Rossignoli}}]{gigena2021many}%
  \BibitemOpen
  \bibfield  {author} {\bibinfo {author} {\bibfnamefont {N.}~\bibnamefont {Gigena}}, \bibinfo {author} {\bibfnamefont {M.}~\bibnamefont {Di~Tullio}},\ and\ \bibinfo {author} {\bibfnamefont {R.}~\bibnamefont {Rossignoli}},\ }\bibfield  {title} {\bibinfo {title} {Many-body entanglement in fermion systems},\ }\href {https://doi.org/10.1103/PhysRevA.103.052424} {\bibfield  {journal} {\bibinfo  {journal} {Physical Review A}\ }\textbf {\bibinfo {volume} {103}},\ \bibinfo {pages} {052424} (\bibinfo {year} {2021})}\BibitemShut {NoStop}%
\bibitem [{\citenamefont {Guaita}\ \emph {et~al.}(2021)\citenamefont {Guaita}, \citenamefont {Hackl}, \citenamefont {Shi}, \citenamefont {Demler},\ and\ \citenamefont {Cirac}}]{guaita2021generalization}%
  \BibitemOpen
  \bibfield  {author} {\bibinfo {author} {\bibfnamefont {T.}~\bibnamefont {Guaita}}, \bibinfo {author} {\bibfnamefont {L.}~\bibnamefont {Hackl}}, \bibinfo {author} {\bibfnamefont {T.}~\bibnamefont {Shi}}, \bibinfo {author} {\bibfnamefont {E.}~\bibnamefont {Demler}},\ and\ \bibinfo {author} {\bibfnamefont {J.~I.}\ \bibnamefont {Cirac}},\ }\bibfield  {title} {\bibinfo {title} {Generalization of group-theoretic coherent states for variational calculations},\ }\href {https://doi.org/10.1103/PhysRevResearch.3.023090} {\bibfield  {journal} {\bibinfo  {journal} {Physical Review Research}\ }\textbf {\bibinfo {volume} {3}},\ \bibinfo {pages} {023090} (\bibinfo {year} {2021})}\BibitemShut {NoStop}%
\bibitem [{\citenamefont {Ali~Ahmad}\ \emph {et~al.}(2022)\citenamefont {Ali~Ahmad}, \citenamefont {Galley}, \citenamefont {H\"ohn}, \citenamefont {Lock},\ and\ \citenamefont {Smith}}]{ahmad2022quantum}%
  \BibitemOpen
  \bibfield  {author} {\bibinfo {author} {\bibfnamefont {S.}~\bibnamefont {Ali~Ahmad}}, \bibinfo {author} {\bibfnamefont {T.~D.}\ \bibnamefont {Galley}}, \bibinfo {author} {\bibfnamefont {P.~A.}\ \bibnamefont {H\"ohn}}, \bibinfo {author} {\bibfnamefont {M.~P.~E.}\ \bibnamefont {Lock}},\ and\ \bibinfo {author} {\bibfnamefont {A.~R.~H.}\ \bibnamefont {Smith}},\ }\bibfield  {title} {\bibinfo {title} {Quantum relativity of subsystems},\ }\href {https://doi.org/10.1103/PhysRevLett.128.170401} {\bibfield  {journal} {\bibinfo  {journal} {Phys. Rev. Lett.}\ }\textbf {\bibinfo {volume} {128}},\ \bibinfo {pages} {170401} (\bibinfo {year} {2022})}\BibitemShut {NoStop}%
\bibitem [{\citenamefont {Grassl}\ \emph {et~al.}(1998)\citenamefont {Grassl}, \citenamefont {R{\"o}tteler},\ and\ \citenamefont {Beth}}]{grassl1998computing}%
  \BibitemOpen
  \bibfield  {author} {\bibinfo {author} {\bibfnamefont {M.}~\bibnamefont {Grassl}}, \bibinfo {author} {\bibfnamefont {M.}~\bibnamefont {R{\"o}tteler}},\ and\ \bibinfo {author} {\bibfnamefont {T.}~\bibnamefont {Beth}},\ }\bibfield  {title} {\bibinfo {title} {Computing local invariants of quantum-bit systems},\ }\href {https://doi.org/10.1103/PhysRevA.58.1833} {\bibfield  {journal} {\bibinfo  {journal} {Physical Review A}\ }\textbf {\bibinfo {volume} {58}},\ \bibinfo {pages} {1833} (\bibinfo {year} {1998})}\BibitemShut {NoStop}%
\bibitem [{\citenamefont {Barnum}\ and\ \citenamefont {Linden}(2001)}]{barnum2001monotones}%
  \BibitemOpen
  \bibfield  {author} {\bibinfo {author} {\bibfnamefont {H.}~\bibnamefont {Barnum}}\ and\ \bibinfo {author} {\bibfnamefont {N.}~\bibnamefont {Linden}},\ }\bibfield  {title} {\bibinfo {title} {Monotones and invariants for multi-particle quantum states},\ }\href {https://doi.org/10.1088/0305-4470/34/35/305} {\bibfield  {journal} {\bibinfo  {journal} {Journal of Physics A: Mathematical and General}\ }\textbf {\bibinfo {volume} {34}},\ \bibinfo {pages} {6787} (\bibinfo {year} {2001})}\BibitemShut {NoStop}%
\bibitem [{\citenamefont {Miyake}(2003)}]{miyake2003classification}%
  \BibitemOpen
  \bibfield  {author} {\bibinfo {author} {\bibfnamefont {A.}~\bibnamefont {Miyake}},\ }\bibfield  {title} {\bibinfo {title} {Classification of multipartite entangled states by multidimensional determinants},\ }\href {https://doi.org/10.1103/PhysRevA.67.012108} {\bibfield  {journal} {\bibinfo  {journal} {Physical Review A}\ }\textbf {\bibinfo {volume} {67}},\ \bibinfo {pages} {012108} (\bibinfo {year} {2003})}\BibitemShut {NoStop}%
\bibitem [{\citenamefont {Leifer}\ \emph {et~al.}(2004)\citenamefont {Leifer}, \citenamefont {Linden},\ and\ \citenamefont {Winter}}]{leifer2004measuring}%
  \BibitemOpen
  \bibfield  {author} {\bibinfo {author} {\bibfnamefont {M.~S.}\ \bibnamefont {Leifer}}, \bibinfo {author} {\bibfnamefont {N.}~\bibnamefont {Linden}},\ and\ \bibinfo {author} {\bibfnamefont {A.}~\bibnamefont {Winter}},\ }\bibfield  {title} {\bibinfo {title} {Measuring polynomial invariants of multiparty quantum states},\ }\href {https://doi.org/https://doi.org/10.1103/PhysRevA.69.052304} {\bibfield  {journal} {\bibinfo  {journal} {Physical Review A}\ }\textbf {\bibinfo {volume} {69}},\ \bibinfo {pages} {052304} (\bibinfo {year} {2004})}\BibitemShut {NoStop}%
\bibitem [{\citenamefont {Mandilara}\ \emph {et~al.}(2006)\citenamefont {Mandilara}, \citenamefont {Akulin}, \citenamefont {Smilga},\ and\ \citenamefont {Viola}}]{mandilara2006quantum}%
  \BibitemOpen
  \bibfield  {author} {\bibinfo {author} {\bibfnamefont {A.}~\bibnamefont {Mandilara}}, \bibinfo {author} {\bibfnamefont {V.~M.}\ \bibnamefont {Akulin}}, \bibinfo {author} {\bibfnamefont {A.~V.}\ \bibnamefont {Smilga}},\ and\ \bibinfo {author} {\bibfnamefont {L.}~\bibnamefont {Viola}},\ }\bibfield  {title} {\bibinfo {title} {Quantum entanglement via nilpotent polynomials},\ }\href {https://doi.org/10.1103/PhysRevA.74.022331} {\bibfield  {journal} {\bibinfo  {journal} {Physical Review A—Atomic, Molecular, and Optical Physics}\ }\textbf {\bibinfo {volume} {74}},\ \bibinfo {pages} {022331} (\bibinfo {year} {2006})}\BibitemShut {NoStop}%
\bibitem [{\citenamefont {Klyachko}(2007)}]{klyachko2007dynamical}%
  \BibitemOpen
  \bibfield  {author} {\bibinfo {author} {\bibfnamefont {A.}~\bibnamefont {Klyachko}},\ }\bibfield  {title} {\bibinfo {title} {Dynamical symmetry approach to entanglement},\ }\href {https://arxiv.org/abs/0802.4008} {\bibfield  {journal} {\bibinfo  {journal} {NATO SECURITY THROUGH SCIENCE SERIES D-INFORMATION AND COMMUNICATION SECURITY}\ }\textbf {\bibinfo {volume} {7}},\ \bibinfo {pages} {25} (\bibinfo {year} {2007})}\BibitemShut {NoStop}%
\bibitem [{\citenamefont {Oszmaniec}\ and\ \citenamefont {Ku{\'s}}(2013)}]{oszmaniec2013universal}%
  \BibitemOpen
  \bibfield  {author} {\bibinfo {author} {\bibfnamefont {M.}~\bibnamefont {Oszmaniec}}\ and\ \bibinfo {author} {\bibfnamefont {M.}~\bibnamefont {Ku{\'s}}},\ }\bibfield  {title} {\bibinfo {title} {Universal framework for entanglement detection},\ }\href {https://doi.org/10.1103/PhysRevA.88.052328} {\bibfield  {journal} {\bibinfo  {journal} {Physical Review A—Atomic, Molecular, and Optical Physics}\ }\textbf {\bibinfo {volume} {88}},\ \bibinfo {pages} {052328} (\bibinfo {year} {2013})}\BibitemShut {NoStop}%
\bibitem [{\citenamefont {Bravyi}\ \emph {et~al.}(2019)\citenamefont {Bravyi}, \citenamefont {Browne}, \citenamefont {Calpin}, \citenamefont {Campbell}, \citenamefont {Gosset},\ and\ \citenamefont {Howard}}]{bravyi2019simulation}%
  \BibitemOpen
  \bibfield  {author} {\bibinfo {author} {\bibfnamefont {S.}~\bibnamefont {Bravyi}}, \bibinfo {author} {\bibfnamefont {D.}~\bibnamefont {Browne}}, \bibinfo {author} {\bibfnamefont {P.}~\bibnamefont {Calpin}}, \bibinfo {author} {\bibfnamefont {E.}~\bibnamefont {Campbell}}, \bibinfo {author} {\bibfnamefont {D.}~\bibnamefont {Gosset}},\ and\ \bibinfo {author} {\bibfnamefont {M.}~\bibnamefont {Howard}},\ }\bibfield  {title} {\bibinfo {title} {Simulation of quantum circuits by low-rank stabilizer decompositions},\ }\href {https://doi.org/10.22331/q-2019-09-02-181} {\bibfield  {journal} {\bibinfo  {journal} {Quantum}\ }\textbf {\bibinfo {volume} {3}},\ \bibinfo {pages} {181} (\bibinfo {year} {2019})}\BibitemShut {NoStop}%
\bibitem [{\citenamefont {Larocca}\ \emph {et~al.}(2022)\citenamefont {Larocca}, \citenamefont {Sauvage}, \citenamefont {Sbahi}, \citenamefont {Verdon}, \citenamefont {Coles},\ and\ \citenamefont {Cerezo}}]{larocca2022group}%
  \BibitemOpen
  \bibfield  {author} {\bibinfo {author} {\bibfnamefont {M.}~\bibnamefont {Larocca}}, \bibinfo {author} {\bibfnamefont {F.}~\bibnamefont {Sauvage}}, \bibinfo {author} {\bibfnamefont {F.~M.}\ \bibnamefont {Sbahi}}, \bibinfo {author} {\bibfnamefont {G.}~\bibnamefont {Verdon}}, \bibinfo {author} {\bibfnamefont {P.~J.}\ \bibnamefont {Coles}},\ and\ \bibinfo {author} {\bibfnamefont {M.}~\bibnamefont {Cerezo}},\ }\bibfield  {title} {\bibinfo {title} {Group-invariant quantum machine learning},\ }\href {https://doi.org/10.1103/PRXQuantum.3.030341} {\bibfield  {journal} {\bibinfo  {journal} {PRX Quantum}\ }\textbf {\bibinfo {volume} {3}},\ \bibinfo {pages} {030341} (\bibinfo {year} {2022})}\BibitemShut {NoStop}%
\bibitem [{\citenamefont {Meyer}\ \emph {et~al.}(2023)\citenamefont {Meyer}, \citenamefont {Mularski}, \citenamefont {Gil-Fuster}, \citenamefont {Mele}, \citenamefont {Arzani}, \citenamefont {Wilms},\ and\ \citenamefont {Eisert}}]{meyer2023exploiting}%
  \BibitemOpen
  \bibfield  {author} {\bibinfo {author} {\bibfnamefont {J.~J.}\ \bibnamefont {Meyer}}, \bibinfo {author} {\bibfnamefont {M.}~\bibnamefont {Mularski}}, \bibinfo {author} {\bibfnamefont {E.}~\bibnamefont {Gil-Fuster}}, \bibinfo {author} {\bibfnamefont {A.~A.}\ \bibnamefont {Mele}}, \bibinfo {author} {\bibfnamefont {F.}~\bibnamefont {Arzani}}, \bibinfo {author} {\bibfnamefont {A.}~\bibnamefont {Wilms}},\ and\ \bibinfo {author} {\bibfnamefont {J.}~\bibnamefont {Eisert}},\ }\bibfield  {title} {\bibinfo {title} {Exploiting symmetry in variational quantum machine learning},\ }\href {https://doi.org/10.1103/PRXQuantum.4.010328} {\bibfield  {journal} {\bibinfo  {journal} {PRX Quantum}\ }\textbf {\bibinfo {volume} {4}},\ \bibinfo {pages} {010328} (\bibinfo {year} {2023})}\BibitemShut {NoStop}%
\bibitem [{\citenamefont {Nguyen}\ \emph {et~al.}(2024)\citenamefont {Nguyen}, \citenamefont {Schatzki}, \citenamefont {Braccia}, \citenamefont {Ragone}, \citenamefont {Coles}, \citenamefont {Sauvage}, \citenamefont {Larocca},\ and\ \citenamefont {Cerezo}}]{nguyen2022atheory}%
  \BibitemOpen
  \bibfield  {author} {\bibinfo {author} {\bibfnamefont {Q.~T.}\ \bibnamefont {Nguyen}}, \bibinfo {author} {\bibfnamefont {L.}~\bibnamefont {Schatzki}}, \bibinfo {author} {\bibfnamefont {P.}~\bibnamefont {Braccia}}, \bibinfo {author} {\bibfnamefont {M.}~\bibnamefont {Ragone}}, \bibinfo {author} {\bibfnamefont {P.~J.}\ \bibnamefont {Coles}}, \bibinfo {author} {\bibfnamefont {F.}~\bibnamefont {Sauvage}}, \bibinfo {author} {\bibfnamefont {M.}~\bibnamefont {Larocca}},\ and\ \bibinfo {author} {\bibfnamefont {M.}~\bibnamefont {Cerezo}},\ }\bibfield  {title} {\bibinfo {title} {Theory for equivariant quantum neural networks},\ }\href {https://doi.org/10.1103/PRXQuantum.5.020328} {\bibfield  {journal} {\bibinfo  {journal} {PRX Quantum}\ }\textbf {\bibinfo {volume} {5}},\ \bibinfo {pages} {020328} (\bibinfo {year} {2024})}\BibitemShut {NoStop}%
\bibitem [{\citenamefont {Skolik}\ \emph {et~al.}(2023)\citenamefont {Skolik}, \citenamefont {Cattelan}, \citenamefont {Yarkoni}, \citenamefont {B{\"a}ck},\ and\ \citenamefont {Dunjko}}]{skolik2022equivariant}%
  \BibitemOpen
  \bibfield  {author} {\bibinfo {author} {\bibfnamefont {A.}~\bibnamefont {Skolik}}, \bibinfo {author} {\bibfnamefont {M.}~\bibnamefont {Cattelan}}, \bibinfo {author} {\bibfnamefont {S.}~\bibnamefont {Yarkoni}}, \bibinfo {author} {\bibfnamefont {T.}~\bibnamefont {B{\"a}ck}},\ and\ \bibinfo {author} {\bibfnamefont {V.}~\bibnamefont {Dunjko}},\ }\bibfield  {title} {\bibinfo {title} {Equivariant quantum circuits for learning on weighted graphs},\ }\href {https://doi.org/10.1038/s41534-023-00710-y} {\bibfield  {journal} {\bibinfo  {journal} {npj Quantum Information}\ }\textbf {\bibinfo {volume} {9}},\ \bibinfo {pages} {47} (\bibinfo {year} {2023})}\BibitemShut {NoStop}%
\bibitem [{\citenamefont {Werlang}\ \emph {et~al.}(2011)\citenamefont {Werlang}, \citenamefont {Ribeiro},\ and\ \citenamefont {Rigolin}}]{werlang2011spotlighting}%
  \BibitemOpen
  \bibfield  {author} {\bibinfo {author} {\bibfnamefont {T.}~\bibnamefont {Werlang}}, \bibinfo {author} {\bibfnamefont {G.}~\bibnamefont {Ribeiro}},\ and\ \bibinfo {author} {\bibfnamefont {G.}~\bibnamefont {Rigolin}},\ }\bibfield  {title} {\bibinfo {title} {Spotlighting quantum critical points via quantum correlations at finite temperatures},\ }\href {https://doi.org/10.1103/PhysRevA.83.062334} {\bibfield  {journal} {\bibinfo  {journal} {Physical Review A—Atomic, Molecular, and Optical Physics}\ }\textbf {\bibinfo {volume} {83}},\ \bibinfo {pages} {062334} (\bibinfo {year} {2011})}\BibitemShut {NoStop}%
\bibitem [{\citenamefont {Viscondi}\ \emph {et~al.}(2009)\citenamefont {Viscondi}, \citenamefont {Furuya},\ and\ \citenamefont {de~Oliveira}}]{viscondi2009generalized}%
  \BibitemOpen
  \bibfield  {author} {\bibinfo {author} {\bibfnamefont {T.~F.}\ \bibnamefont {Viscondi}}, \bibinfo {author} {\bibfnamefont {K.}~\bibnamefont {Furuya}},\ and\ \bibinfo {author} {\bibfnamefont {M.}~\bibnamefont {de~Oliveira}},\ }\bibfield  {title} {\bibinfo {title} {Generalized purity and quantum phase transition for bose-einstein condensates in a symmetric double well},\ }\href {https://doi.org/10.1103/PhysRevA.80.013610} {\bibfield  {journal} {\bibinfo  {journal} {Physical Review A}\ }\textbf {\bibinfo {volume} {80}},\ \bibinfo {pages} {013610} (\bibinfo {year} {2009})}\BibitemShut {NoStop}%
\bibitem [{\citenamefont {Weinstein}\ and\ \citenamefont {Viola}(2006)}]{weinstein2006generalized}%
  \BibitemOpen
  \bibfield  {author} {\bibinfo {author} {\bibfnamefont {Y.~S.}\ \bibnamefont {Weinstein}}\ and\ \bibinfo {author} {\bibfnamefont {L.}~\bibnamefont {Viola}},\ }\bibfield  {title} {\bibinfo {title} {Generalized entanglement as a natural framework for exploring quantum chaos},\ }\href {https://doi.org/10.1209/epl/i2006-10354-7/meta} {\bibfield  {journal} {\bibinfo  {journal} {Europhysics Letters}\ }\textbf {\bibinfo {volume} {76}},\ \bibinfo {pages} {746} (\bibinfo {year} {2006})}\BibitemShut {NoStop}%
\bibitem [{\citenamefont {Nielsen}(1999)}]{nielsen1999conditions}%
  \BibitemOpen
  \bibfield  {author} {\bibinfo {author} {\bibfnamefont {M.~A.}\ \bibnamefont {Nielsen}},\ }\bibfield  {title} {\bibinfo {title} {Conditions for a class of entanglement transformations},\ }\href {https://journals.aps.org/prl/abstract/10.1103/PhysRevLett.83.436} {\bibfield  {journal} {\bibinfo  {journal} {Physical Review Letters}\ }\textbf {\bibinfo {volume} {83}},\ \bibinfo {pages} {436} (\bibinfo {year} {1999})}\BibitemShut {NoStop}%
\bibitem [{\citenamefont {Ku{\'s}}\ and\ \citenamefont {{\.Z}yczkowski}(2001)}]{kus2001geometry}%
  \BibitemOpen
  \bibfield  {author} {\bibinfo {author} {\bibfnamefont {M.}~\bibnamefont {Ku{\'s}}}\ and\ \bibinfo {author} {\bibfnamefont {K.}~\bibnamefont {{\.Z}yczkowski}},\ }\bibfield  {title} {\bibinfo {title} {Geometry of entangled states},\ }\href {https://doi.org/10.1103/PhysRevA.63.032307} {\bibfield  {journal} {\bibinfo  {journal} {Physical Review A}\ }\textbf {\bibinfo {volume} {63}},\ \bibinfo {pages} {032307} (\bibinfo {year} {2001})}\BibitemShut {NoStop}%
\bibitem [{\citenamefont {Verstraete}\ \emph {et~al.}(2003)\citenamefont {Verstraete}, \citenamefont {Dehaene},\ and\ \citenamefont {De~Moor}}]{verstraete2003normal}%
  \BibitemOpen
  \bibfield  {author} {\bibinfo {author} {\bibfnamefont {F.}~\bibnamefont {Verstraete}}, \bibinfo {author} {\bibfnamefont {J.}~\bibnamefont {Dehaene}},\ and\ \bibinfo {author} {\bibfnamefont {B.}~\bibnamefont {De~Moor}},\ }\bibfield  {title} {\bibinfo {title} {Normal forms and entanglement measures for multipartite quantum states},\ }\href {https://doi.org/https://doi.org/10.1103/PhysRevA.68.012103} {\bibfield  {journal} {\bibinfo  {journal} {Physical Review A}\ }\textbf {\bibinfo {volume} {68}},\ \bibinfo {pages} {012103} (\bibinfo {year} {2003})}\BibitemShut {NoStop}%
\bibitem [{\citenamefont {Garibaldi}\ and\ \citenamefont {Guralnick}(2022)}]{garibaldi2022generic}%
  \BibitemOpen
  \bibfield  {author} {\bibinfo {author} {\bibfnamefont {S.}~\bibnamefont {Garibaldi}}\ and\ \bibinfo {author} {\bibfnamefont {R.~M.}\ \bibnamefont {Guralnick}},\ }\bibfield  {title} {\bibinfo {title} {Generic stabilizers for simple algebraic groups},\ }\href {https://doi.org/10.1307/mmj/20217216} {\bibfield  {journal} {\bibinfo  {journal} {Michigan Mathematical Journal}\ }\textbf {\bibinfo {volume} {72}},\ \bibinfo {pages} {343} (\bibinfo {year} {2022})}\BibitemShut {NoStop}%
\bibitem [{\citenamefont {Park}\ \emph {et~al.}(2024)\citenamefont {Park}, \citenamefont {Jung}, \citenamefont {Park},\ and\ \citenamefont {Youn}}]{park2024universal}%
  \BibitemOpen
  \bibfield  {author} {\bibinfo {author} {\bibfnamefont {S.-J.}\ \bibnamefont {Park}}, \bibinfo {author} {\bibfnamefont {Y.-G.}\ \bibnamefont {Jung}}, \bibinfo {author} {\bibfnamefont {J.}~\bibnamefont {Park}},\ and\ \bibinfo {author} {\bibfnamefont {S.-G.}\ \bibnamefont {Youn}},\ }\bibfield  {title} {\bibinfo {title} {A universal framework for entanglement detection under group symmetry},\ }\href {https://doi.org/10.1103/PhysRevA.88.052328} {\bibfield  {journal} {\bibinfo  {journal} {Journal of Physics A: Mathematical and Theoretical}\ }\textbf {\bibinfo {volume} {57}},\ \bibinfo {pages} {325304} (\bibinfo {year} {2024})}\BibitemShut {NoStop}%
\bibitem [{\citenamefont {Gigena}\ \emph {et~al.}(2020)\citenamefont {Gigena}, \citenamefont {Di~Tullio},\ and\ \citenamefont {Rossignoli}}]{gigena2020one}%
  \BibitemOpen
  \bibfield  {author} {\bibinfo {author} {\bibfnamefont {N.}~\bibnamefont {Gigena}}, \bibinfo {author} {\bibfnamefont {M.}~\bibnamefont {Di~Tullio}},\ and\ \bibinfo {author} {\bibfnamefont {R.}~\bibnamefont {Rossignoli}},\ }\bibfield  {title} {\bibinfo {title} {One-body entanglement as a quantum resource in fermionic systems},\ }\href {https://doi.org/10.1103/PhysRevA.102.042410} {\bibfield  {journal} {\bibinfo  {journal} {Physical Review A}\ }\textbf {\bibinfo {volume} {102}},\ \bibinfo {pages} {042410} (\bibinfo {year} {2020})}\BibitemShut {NoStop}%
\bibitem [{\citenamefont {Datta}\ \emph {et~al.}(2005)\citenamefont {Datta}, \citenamefont {Flammia},\ and\ \citenamefont {Caves}}]{datta2005entanglement}%
  \BibitemOpen
  \bibfield  {author} {\bibinfo {author} {\bibfnamefont {A.}~\bibnamefont {Datta}}, \bibinfo {author} {\bibfnamefont {S.~T.}\ \bibnamefont {Flammia}},\ and\ \bibinfo {author} {\bibfnamefont {C.~M.}\ \bibnamefont {Caves}},\ }\bibfield  {title} {\bibinfo {title} {Entanglement and the power of one qubit},\ }\href {https://doi.org/10.1103/PhysRevA.72.042316} {\bibfield  {journal} {\bibinfo  {journal} {Physical Review A}\ }\textbf {\bibinfo {volume} {72}},\ \bibinfo {pages} {042316} (\bibinfo {year} {2005})}\BibitemShut {NoStop}%
\bibitem [{\citenamefont {Gross}\ \emph {et~al.}(2009)\citenamefont {Gross}, \citenamefont {Flammia},\ and\ \citenamefont {Eisert}}]{gross2009most}%
  \BibitemOpen
  \bibfield  {author} {\bibinfo {author} {\bibfnamefont {D.}~\bibnamefont {Gross}}, \bibinfo {author} {\bibfnamefont {S.~T.}\ \bibnamefont {Flammia}},\ and\ \bibinfo {author} {\bibfnamefont {J.}~\bibnamefont {Eisert}},\ }\bibfield  {title} {\bibinfo {title} {Most quantum states are too entangled to be useful as computational resources},\ }\href {https://doi.org/10.1103/PhysRevLett.102.190501} {\bibfield  {journal} {\bibinfo  {journal} {Physical review letters}\ }\textbf {\bibinfo {volume} {102}},\ \bibinfo {pages} {190501} (\bibinfo {year} {2009})}\BibitemShut {NoStop}%
\bibitem [{\citenamefont {Diaz}\ \emph {et~al.}(2023)\citenamefont {Diaz}, \citenamefont {Garc{\'\i}a-Mart{\'\i}n}, \citenamefont {Kazi}, \citenamefont {Larocca},\ and\ \citenamefont {Cerezo}}]{diaz2023showcasing}%
  \BibitemOpen
  \bibfield  {author} {\bibinfo {author} {\bibfnamefont {N.~L.}\ \bibnamefont {Diaz}}, \bibinfo {author} {\bibfnamefont {D.}~\bibnamefont {Garc{\'\i}a-Mart{\'\i}n}}, \bibinfo {author} {\bibfnamefont {S.}~\bibnamefont {Kazi}}, \bibinfo {author} {\bibfnamefont {M.}~\bibnamefont {Larocca}},\ and\ \bibinfo {author} {\bibfnamefont {M.}~\bibnamefont {Cerezo}},\ }\bibfield  {title} {\bibinfo {title} {Showcasing a barren plateau theory beyond the dynamical lie algebra},\ }\href {https://arxiv.org/abs/2310.11505} {\bibfield  {journal} {\bibinfo  {journal} {arXiv preprint arXiv:2310.11505}\ } (\bibinfo {year} {2023})}\BibitemShut {NoStop}%
\bibitem [{\citenamefont {Diaz}\ \emph {et~al.}(2025)\citenamefont {Diaz}, \citenamefont {Mele}, \citenamefont {Bermejo}, \citenamefont {Braccia}, \citenamefont {Larocca},\ and\ \citenamefont {Cerezo}}]{nahuel2025quantum}%
  \BibitemOpen
  \bibfield  {author} {\bibinfo {author} {\bibfnamefont {N.~L.}\ \bibnamefont {Diaz}}, \bibinfo {author} {\bibfnamefont {A.~A.}\ \bibnamefont {Mele}}, \bibinfo {author} {\bibfnamefont {P.}~\bibnamefont {Bermejo}}, \bibinfo {author} {\bibfnamefont {P.}~\bibnamefont {Braccia}}, \bibinfo {author} {\bibfnamefont {M.}~\bibnamefont {Larocca}},\ and\ \bibinfo {author} {\bibfnamefont {M.}~\bibnamefont {Cerezo}},\ }\bibfield  {title} {\bibinfo {title} {Resource non-increasing operations from a unification of quantum resource theories},\ }\href@noop {} {\bibfield  {journal} {\bibinfo  {journal} {Manuscript in preparation}\ } (\bibinfo {year} {2025})}\BibitemShut {NoStop}%
\bibitem [{\citenamefont {Serre}\ \emph {et~al.}(1977)\citenamefont {Serre} \emph {et~al.}}]{serre1977linear}%
  \BibitemOpen
  \bibfield  {author} {\bibinfo {author} {\bibfnamefont {J.-P.}\ \bibnamefont {Serre}} \emph {et~al.},\ }\href@noop {} {\emph {\bibinfo {title} {Linear representations of finite groups}}},\ Vol.~\bibinfo {volume} {42}\ (\bibinfo  {publisher} {Springer},\ \bibinfo {year} {1977})\BibitemShut {NoStop}%
\bibitem [{\citenamefont {Fulton}\ and\ \citenamefont {Harris}(1991)}]{fulton1991representation}%
  \BibitemOpen
  \bibfield  {author} {\bibinfo {author} {\bibfnamefont {W.}~\bibnamefont {Fulton}}\ and\ \bibinfo {author} {\bibfnamefont {J.}~\bibnamefont {Harris}},\ }\href@noop {} {\emph {\bibinfo {title} {Representation Theory: A First Course}}}\ (\bibinfo  {publisher} {Springer},\ \bibinfo {year} {1991})\BibitemShut {NoStop}%
\bibitem [{\citenamefont {Mele}\ \emph {et~al.}(2025)\citenamefont {Mele}, \citenamefont {Leone}, \citenamefont {Bermejo}, \citenamefont {Braccia}, \citenamefont {Larocca},\ and\ \citenamefont {Cerezo}}]{mele2025clifford}%
  \BibitemOpen
  \bibfield  {author} {\bibinfo {author} {\bibfnamefont {A.~A.}\ \bibnamefont {Mele}}, \bibinfo {author} {\bibfnamefont {L.}~\bibnamefont {Leone}}, \bibinfo {author} {\bibfnamefont {P.}~\bibnamefont {Bermejo}}, \bibinfo {author} {\bibfnamefont {P.}~\bibnamefont {Braccia}}, \bibinfo {author} {\bibfnamefont {M.}~\bibnamefont {Larocca}},\ and\ \bibinfo {author} {\bibfnamefont {M.}~\bibnamefont {Cerezo}},\ }\bibfield  {title} {\bibinfo {title} {Clifford fourier decompositions are magic},\ }\href@noop {} {\bibfield  {journal} {\bibinfo  {journal} {Manuscript in preparation}\ } (\bibinfo {year} {2025})}\BibitemShut {NoStop}%
\bibitem [{\citenamefont {Perelomov}(1977)}]{perelomov1977generalized}%
  \BibitemOpen
  \bibfield  {author} {\bibinfo {author} {\bibfnamefont {A.~M.}\ \bibnamefont {Perelomov}},\ }\bibfield  {title} {\bibinfo {title} {Generalized coherent states and some of their applications},\ }\href {https://doi.org/10.1070/PU1977v020n09ABEH005459} {\bibfield  {journal} {\bibinfo  {journal} {Soviet Physics Uspekhi}\ }\textbf {\bibinfo {volume} {20}},\ \bibinfo {pages} {703} (\bibinfo {year} {1977})}\BibitemShut {NoStop}%
\bibitem [{\citenamefont {Gilmore}(1974)}]{gilmore1974properties}%
  \BibitemOpen
  \bibfield  {author} {\bibinfo {author} {\bibfnamefont {R.}~\bibnamefont {Gilmore}},\ }\bibfield  {title} {\bibinfo {title} {On the properties of coherent states},\ }\href {https://rmf.smf.mx/ojs/index.php/rmf/article/view/1046} {\bibfield  {journal} {\bibinfo  {journal} {Revista Mexicana de F{\'\i}sica}\ }\textbf {\bibinfo {volume} {23}},\ \bibinfo {pages} {143} (\bibinfo {year} {1974})}\BibitemShut {NoStop}%
\bibitem [{\citenamefont {Zhang}\ \emph {et~al.}(1990)\citenamefont {Zhang}, \citenamefont {Gilmore} \emph {et~al.}}]{zhang1990coherent}%
  \BibitemOpen
  \bibfield  {author} {\bibinfo {author} {\bibfnamefont {W.-M.}\ \bibnamefont {Zhang}}, \bibinfo {author} {\bibfnamefont {R.}~\bibnamefont {Gilmore}}, \emph {et~al.},\ }\bibfield  {title} {\bibinfo {title} {Coherent states: Theory and some applications},\ }\href {https://doi.org/10.1103/RevModPhys.62.867} {\bibfield  {journal} {\bibinfo  {journal} {Reviews of Modern Physics}\ }\textbf {\bibinfo {volume} {62}},\ \bibinfo {pages} {867} (\bibinfo {year} {1990})}\BibitemShut {NoStop}%
\bibitem [{\citenamefont {Leone}\ \emph {et~al.}(2022)\citenamefont {Leone}, \citenamefont {Oliviero},\ and\ \citenamefont {Hamma}}]{leone2022stabilizer}%
  \BibitemOpen
  \bibfield  {author} {\bibinfo {author} {\bibfnamefont {L.}~\bibnamefont {Leone}}, \bibinfo {author} {\bibfnamefont {S.~F.}\ \bibnamefont {Oliviero}},\ and\ \bibinfo {author} {\bibfnamefont {A.}~\bibnamefont {Hamma}},\ }\bibfield  {title} {\bibinfo {title} {Stabilizer r{\'e}nyi entropy},\ }\href {https://doi.org/10.1103/PhysRevLett.128.050402} {\bibfield  {journal} {\bibinfo  {journal} {Physical Review Letters}\ }\textbf {\bibinfo {volume} {128}},\ \bibinfo {pages} {050402} (\bibinfo {year} {2022})}\BibitemShut {NoStop}%
\bibitem [{\citenamefont {Mele}(2024)}]{mele2023introduction}%
  \BibitemOpen
  \bibfield  {author} {\bibinfo {author} {\bibfnamefont {A.~A.}\ \bibnamefont {Mele}},\ }\bibfield  {title} {\bibinfo {title} {Introduction to haar measure tools in quantum information: A beginner's tutorial},\ }\href {https://doi.org/10.22331/q-2024-05-08-1340} {\bibfield  {journal} {\bibinfo  {journal} {Quantum}\ }\textbf {\bibinfo {volume} {8}},\ \bibinfo {pages} {1340} (\bibinfo {year} {2024})}\BibitemShut {NoStop}%
\bibitem [{\citenamefont {Wilde}(2013)}]{wilde2013quantum}%
  \BibitemOpen
  \bibfield  {author} {\bibinfo {author} {\bibfnamefont {M.~M.}\ \bibnamefont {Wilde}},\ }\href@noop {} {\emph {\bibinfo {title} {Quantum information theory}}}\ (\bibinfo  {publisher} {Cambridge University Press},\ \bibinfo {year} {2013})\BibitemShut {NoStop}%
\bibitem [{\citenamefont {Helwig}\ and\ \citenamefont {Cui}(2013)}]{helwig2013absolutely}%
  \BibitemOpen
  \bibfield  {author} {\bibinfo {author} {\bibfnamefont {W.}~\bibnamefont {Helwig}}\ and\ \bibinfo {author} {\bibfnamefont {W.}~\bibnamefont {Cui}},\ }\bibfield  {title} {\bibinfo {title} {Absolutely maximally entangled states: existence and applications},\ }\bibfield  {journal} {\bibinfo  {journal} {arXiv preprint arXiv:1306.2536}\ }\href {https://doi.org/10.48550/arXiv.1306.2536} {10.48550/arXiv.1306.2536} (\bibinfo {year} {2013})\BibitemShut {NoStop}%
\bibitem [{\citenamefont {Goyeneche}\ \emph {et~al.}(2015)\citenamefont {Goyeneche}, \citenamefont {Alsina}, \citenamefont {Latorre}, \citenamefont {Riera},\ and\ \citenamefont {{\.Z}yczkowski}}]{goyeneche2015absolutely}%
  \BibitemOpen
  \bibfield  {author} {\bibinfo {author} {\bibfnamefont {D.}~\bibnamefont {Goyeneche}}, \bibinfo {author} {\bibfnamefont {D.}~\bibnamefont {Alsina}}, \bibinfo {author} {\bibfnamefont {J.~I.}\ \bibnamefont {Latorre}}, \bibinfo {author} {\bibfnamefont {A.}~\bibnamefont {Riera}},\ and\ \bibinfo {author} {\bibfnamefont {K.}~\bibnamefont {{\.Z}yczkowski}},\ }\bibfield  {title} {\bibinfo {title} {Absolutely maximally entangled states, combinatorial designs, and multiunitary matrices},\ }\href {https://doi.org/10.1103/PhysRevA.92.032316} {\bibfield  {journal} {\bibinfo  {journal} {Physical Review A}\ }\textbf {\bibinfo {volume} {92}},\ \bibinfo {pages} {032316} (\bibinfo {year} {2015})}\BibitemShut {NoStop}%
\bibitem [{\citenamefont {Huber}\ \emph {et~al.}(2017)\citenamefont {Huber}, \citenamefont {G{\"u}hne},\ and\ \citenamefont {Siewert}}]{huber2017absolutely}%
  \BibitemOpen
  \bibfield  {author} {\bibinfo {author} {\bibfnamefont {F.}~\bibnamefont {Huber}}, \bibinfo {author} {\bibfnamefont {O.}~\bibnamefont {G{\"u}hne}},\ and\ \bibinfo {author} {\bibfnamefont {J.}~\bibnamefont {Siewert}},\ }\bibfield  {title} {\bibinfo {title} {Absolutely maximally entangled states of seven qubits do not exist},\ }\href {https://journals.aps.org/prl/abstract/10.1103/PhysRevLett.118.200502} {\bibfield  {journal} {\bibinfo  {journal} {Physical review letters}\ }\textbf {\bibinfo {volume} {118}},\ \bibinfo {pages} {200502} (\bibinfo {year} {2017})}\BibitemShut {NoStop}%
\bibitem [{\citenamefont {Casas}\ \emph {et~al.}(2025)\citenamefont {Casas}, \citenamefont {Rajchel-Mieldzioc}, \citenamefont {Rather}, \citenamefont {Plodzien}, \citenamefont {Bruzda}, \citenamefont {Cervera-Lierta},\ and\ \citenamefont {Zyczkowski}}]{casas2025quantum}%
  \BibitemOpen
  \bibfield  {author} {\bibinfo {author} {\bibfnamefont {B.}~\bibnamefont {Casas}}, \bibinfo {author} {\bibfnamefont {G.}~\bibnamefont {Rajchel-Mieldzioc}}, \bibinfo {author} {\bibfnamefont {S.~A.}\ \bibnamefont {Rather}}, \bibinfo {author} {\bibfnamefont {M.}~\bibnamefont {Plodzien}}, \bibinfo {author} {\bibfnamefont {W.}~\bibnamefont {Bruzda}}, \bibinfo {author} {\bibfnamefont {A.}~\bibnamefont {Cervera-Lierta}},\ and\ \bibinfo {author} {\bibfnamefont {K.}~\bibnamefont {Zyczkowski}},\ }\bibfield  {title} {\bibinfo {title} {Quantum circuits for high-dimensional absolutely maximally entangled states},\ }\bibfield  {journal} {\bibinfo  {journal} {arXiv preprint arXiv:2504.05394}\ }\href {https://doi.org/10.48550/arXiv.2504.05394} {10.48550/arXiv.2504.05394} (\bibinfo {year} {2025})\BibitemShut {NoStop}%
\bibitem [{\citenamefont {Brauer}\ and\ \citenamefont {Weyl}(1935)}]{brauer1935spinors}%
  \BibitemOpen
  \bibfield  {author} {\bibinfo {author} {\bibfnamefont {R.}~\bibnamefont {Brauer}}\ and\ \bibinfo {author} {\bibfnamefont {H.}~\bibnamefont {Weyl}},\ }\bibfield  {title} {\bibinfo {title} {Spinors in n dimensions},\ }\href {https://doi.org/10.2307/2371218} {\bibfield  {journal} {\bibinfo  {journal} {American Journal of Mathematics}\ }\textbf {\bibinfo {volume} {57}},\ \bibinfo {pages} {425} (\bibinfo {year} {1935})}\BibitemShut {NoStop}%
\bibitem [{\citenamefont {K{\"o}kc{\"u}}\ \emph {et~al.}(2022)\citenamefont {K{\"o}kc{\"u}}, \citenamefont {Steckmann}, \citenamefont {Wang}, \citenamefont {Freericks}, \citenamefont {Dumitrescu},\ and\ \citenamefont {Kemper}}]{kokcu2021fixed}%
  \BibitemOpen
  \bibfield  {author} {\bibinfo {author} {\bibfnamefont {E.}~\bibnamefont {K{\"o}kc{\"u}}}, \bibinfo {author} {\bibfnamefont {T.}~\bibnamefont {Steckmann}}, \bibinfo {author} {\bibfnamefont {Y.}~\bibnamefont {Wang}}, \bibinfo {author} {\bibfnamefont {J.}~\bibnamefont {Freericks}}, \bibinfo {author} {\bibfnamefont {E.~F.}\ \bibnamefont {Dumitrescu}},\ and\ \bibinfo {author} {\bibfnamefont {A.~F.}\ \bibnamefont {Kemper}},\ }\bibfield  {title} {\bibinfo {title} {Fixed depth hamiltonian simulation via cartan decomposition},\ }\href {https://doi.org/10.1103/PhysRevLett.129.070501} {\bibfield  {journal} {\bibinfo  {journal} {Physical Review Letters}\ }\textbf {\bibinfo {volume} {129}},\ \bibinfo {pages} {070501} (\bibinfo {year} {2022})}\BibitemShut {NoStop}%
\bibitem [{\citenamefont {Kazi}\ \emph {et~al.}(2024)\citenamefont {Kazi}, \citenamefont {Larocca}, \citenamefont {Farinati}, \citenamefont {Coles}, \citenamefont {Cerezo},\ and\ \citenamefont {Zeier}}]{kazi2024analyzing}%
  \BibitemOpen
  \bibfield  {author} {\bibinfo {author} {\bibfnamefont {S.}~\bibnamefont {Kazi}}, \bibinfo {author} {\bibfnamefont {M.}~\bibnamefont {Larocca}}, \bibinfo {author} {\bibfnamefont {M.}~\bibnamefont {Farinati}}, \bibinfo {author} {\bibfnamefont {P.~J.}\ \bibnamefont {Coles}}, \bibinfo {author} {\bibfnamefont {M.}~\bibnamefont {Cerezo}},\ and\ \bibinfo {author} {\bibfnamefont {R.}~\bibnamefont {Zeier}},\ }\bibfield  {title} {\bibinfo {title} {Analyzing the quantum approximate optimization algorithm: ans\"atze, symmetries, and lie algebras},\ }\href {https://arxiv.org/abs/2410.05187} {\bibfield  {journal} {\bibinfo  {journal} {arXiv preprint arXiv:2410.05187}\ } (\bibinfo {year} {2024})}\BibitemShut {NoStop}%
\bibitem [{\citenamefont {Weedbrook}\ \emph {et~al.}(2012)\citenamefont {Weedbrook}, \citenamefont {Pirandola}, \citenamefont {Garc{\'\i}a-Patr{\'o}n}, \citenamefont {Cerf}, \citenamefont {Ralph}, \citenamefont {Shapiro},\ and\ \citenamefont {Lloyd}}]{weedbrook2012gaussian}%
  \BibitemOpen
  \bibfield  {author} {\bibinfo {author} {\bibfnamefont {C.}~\bibnamefont {Weedbrook}}, \bibinfo {author} {\bibfnamefont {S.}~\bibnamefont {Pirandola}}, \bibinfo {author} {\bibfnamefont {R.}~\bibnamefont {Garc{\'\i}a-Patr{\'o}n}}, \bibinfo {author} {\bibfnamefont {N.~J.}\ \bibnamefont {Cerf}}, \bibinfo {author} {\bibfnamefont {T.~C.}\ \bibnamefont {Ralph}}, \bibinfo {author} {\bibfnamefont {J.~H.}\ \bibnamefont {Shapiro}},\ and\ \bibinfo {author} {\bibfnamefont {S.}~\bibnamefont {Lloyd}},\ }\bibfield  {title} {\bibinfo {title} {Gaussian quantum information},\ }\href {https://doi.org/10.1103/RevModPhys.84.621} {\bibfield  {journal} {\bibinfo  {journal} {Reviews of Modern Physics}\ }\textbf {\bibinfo {volume} {84}},\ \bibinfo {pages} {621} (\bibinfo {year} {2012})}\BibitemShut {NoStop}%
\bibitem [{\citenamefont {Hebenstreit}\ \emph {et~al.}(2019)\citenamefont {Hebenstreit}, \citenamefont {Jozsa}, \citenamefont {Kraus}, \citenamefont {Strelchuk},\ and\ \citenamefont {Yoganathan}}]{hebenstreit2019all}%
  \BibitemOpen
  \bibfield  {author} {\bibinfo {author} {\bibfnamefont {M.}~\bibnamefont {Hebenstreit}}, \bibinfo {author} {\bibfnamefont {R.}~\bibnamefont {Jozsa}}, \bibinfo {author} {\bibfnamefont {B.}~\bibnamefont {Kraus}}, \bibinfo {author} {\bibfnamefont {S.}~\bibnamefont {Strelchuk}},\ and\ \bibinfo {author} {\bibfnamefont {M.}~\bibnamefont {Yoganathan}},\ }\bibfield  {title} {\bibinfo {title} {All pure fermionic non-gaussian states are magic states for matchgate computations},\ }\href {https://doi.org/10.1103/PhysRevLett.123.080503} {\bibfield  {journal} {\bibinfo  {journal} {Physical review letters}\ }\textbf {\bibinfo {volume} {123}},\ \bibinfo {pages} {080503} (\bibinfo {year} {2019})}\BibitemShut {NoStop}%
\bibitem [{\citenamefont {Jozsa}\ and\ \citenamefont {Miyake}(2008)}]{jozsa2008matchgates}%
  \BibitemOpen
  \bibfield  {author} {\bibinfo {author} {\bibfnamefont {R.}~\bibnamefont {Jozsa}}\ and\ \bibinfo {author} {\bibfnamefont {A.}~\bibnamefont {Miyake}},\ }\bibfield  {title} {\bibinfo {title} {Matchgates and classical simulation of quantum circuits},\ }\href {https://doi.org/10.1098/rspa.2008.0189} {\bibfield  {journal} {\bibinfo  {journal} {Proceedings of the Royal Society A: Mathematical, Physical and Engineering Sciences}\ }\textbf {\bibinfo {volume} {464}},\ \bibinfo {pages} {3089} (\bibinfo {year} {2008})}\BibitemShut {NoStop}%
\bibitem [{\citenamefont {Mele}\ and\ \citenamefont {Herasymenko}(2025)}]{mele2024efficient}%
  \BibitemOpen
  \bibfield  {author} {\bibinfo {author} {\bibfnamefont {A.~A.}\ \bibnamefont {Mele}}\ and\ \bibinfo {author} {\bibfnamefont {Y.}~\bibnamefont {Herasymenko}},\ }\bibfield  {title} {\bibinfo {title} {Efficient learning of quantum states prepared with few fermionic non-gaussian gates},\ }\href {https://doi.org/10.1103/PRXQuantum.6.010319} {\bibfield  {journal} {\bibinfo  {journal} {PRX Quantum}\ }\textbf {\bibinfo {volume} {6}},\ \bibinfo {pages} {010319} (\bibinfo {year} {2025})}\BibitemShut {NoStop}%
\bibitem [{\citenamefont {Dias}\ and\ \citenamefont {Koenig}(2024)}]{dias2023classical}%
  \BibitemOpen
  \bibfield  {author} {\bibinfo {author} {\bibfnamefont {B.}~\bibnamefont {Dias}}\ and\ \bibinfo {author} {\bibfnamefont {R.}~\bibnamefont {Koenig}},\ }\bibfield  {title} {\bibinfo {title} {Classical simulation of non-gaussian fermionic circuits},\ }\href {https://doi.org/10.22331/q-2019-09-02-181} {\bibfield  {journal} {\bibinfo  {journal} {Quantum}\ }\textbf {\bibinfo {volume} {8}},\ \bibinfo {pages} {1350} (\bibinfo {year} {2024})}\BibitemShut {NoStop}%
\bibitem [{\citenamefont {Goh}\ \emph {et~al.}(2023)\citenamefont {Goh}, \citenamefont {Larocca}, \citenamefont {Cincio}, \citenamefont {Cerezo},\ and\ \citenamefont {Sauvage}}]{goh2023lie}%
  \BibitemOpen
  \bibfield  {author} {\bibinfo {author} {\bibfnamefont {M.~L.}\ \bibnamefont {Goh}}, \bibinfo {author} {\bibfnamefont {M.}~\bibnamefont {Larocca}}, \bibinfo {author} {\bibfnamefont {L.}~\bibnamefont {Cincio}}, \bibinfo {author} {\bibfnamefont {M.}~\bibnamefont {Cerezo}},\ and\ \bibinfo {author} {\bibfnamefont {F.}~\bibnamefont {Sauvage}},\ }\bibfield  {title} {\bibinfo {title} {Lie-algebraic classical simulations for quantum computing},\ }\href {https://arxiv.org/abs/2308.01432} {\bibfield  {journal} {\bibinfo  {journal} {arXiv preprint arXiv:2308.01432}\ } (\bibinfo {year} {2023})}\BibitemShut {NoStop}%
\bibitem [{\citenamefont {Brod}\ and\ \citenamefont {Galvao}(2011)}]{brod2011extending}%
  \BibitemOpen
  \bibfield  {author} {\bibinfo {author} {\bibfnamefont {D.~J.}\ \bibnamefont {Brod}}\ and\ \bibinfo {author} {\bibfnamefont {E.~F.}\ \bibnamefont {Galvao}},\ }\bibfield  {title} {\bibinfo {title} {Extending matchgates into universal quantum computation},\ }\href {https://doi.org/10.1103/PhysRevA.84.022310} {\bibfield  {journal} {\bibinfo  {journal} {Physical Review A—Atomic, Molecular, and Optical Physics}\ }\textbf {\bibinfo {volume} {84}},\ \bibinfo {pages} {022310} (\bibinfo {year} {2011})}\BibitemShut {NoStop}%
\bibitem [{\citenamefont {Cudby}\ and\ \citenamefont {Strelchuk}(2023)}]{cudby2023gaussian}%
  \BibitemOpen
  \bibfield  {author} {\bibinfo {author} {\bibfnamefont {J.}~\bibnamefont {Cudby}}\ and\ \bibinfo {author} {\bibfnamefont {S.}~\bibnamefont {Strelchuk}},\ }\bibfield  {title} {\bibinfo {title} {Gaussian decomposition of magic states for matchgate computations},\ }\href {https://arxiv.org/abs/2307.12654} {\bibfield  {journal} {\bibinfo  {journal} {arXiv preprint arXiv:2307.12654}\ } (\bibinfo {year} {2023})}\BibitemShut {NoStop}%
\bibitem [{\citenamefont {Robert}\ and\ \citenamefont {Combescure}(2021)}]{robert2021coherent}%
  \BibitemOpen
  \bibfield  {author} {\bibinfo {author} {\bibfnamefont {D.}~\bibnamefont {Robert}}\ and\ \bibinfo {author} {\bibfnamefont {M.}~\bibnamefont {Combescure}},\ }\href@noop {} {\emph {\bibinfo {title} {Coherent states and applications in mathematical physics}}}\ (\bibinfo  {publisher} {Springer},\ \bibinfo {year} {2021})\BibitemShut {NoStop}%
\bibitem [{\citenamefont {Ragone}\ \emph {et~al.}(2022)\citenamefont {Ragone}, \citenamefont {Nguyen}, \citenamefont {Schatzki}, \citenamefont {Braccia}, \citenamefont {Larocca}, \citenamefont {Sauvage}, \citenamefont {Coles},\ and\ \citenamefont {Cerezo}}]{ragone2022representation}%
  \BibitemOpen
  \bibfield  {author} {\bibinfo {author} {\bibfnamefont {M.}~\bibnamefont {Ragone}}, \bibinfo {author} {\bibfnamefont {Q.~T.}\ \bibnamefont {Nguyen}}, \bibinfo {author} {\bibfnamefont {L.}~\bibnamefont {Schatzki}}, \bibinfo {author} {\bibfnamefont {P.}~\bibnamefont {Braccia}}, \bibinfo {author} {\bibfnamefont {M.}~\bibnamefont {Larocca}}, \bibinfo {author} {\bibfnamefont {F.}~\bibnamefont {Sauvage}}, \bibinfo {author} {\bibfnamefont {P.~J.}\ \bibnamefont {Coles}},\ and\ \bibinfo {author} {\bibfnamefont {M.}~\bibnamefont {Cerezo}},\ }\bibfield  {title} {\bibinfo {title} {Representation theory for geometric quantum machine learning},\ }\href {https://arxiv.org/abs/2210.07980} {\bibfield  {journal} {\bibinfo  {journal} {arXiv preprint arXiv:2210.07980}\ } (\bibinfo {year} {2022})}\BibitemShut {NoStop}%
\bibitem [{\citenamefont {Veitch}\ \emph {et~al.}(2014)\citenamefont {Veitch}, \citenamefont {Mousavian}, \citenamefont {Gottesman},\ and\ \citenamefont {Emerson}}]{veitch2014resource}%
  \BibitemOpen
  \bibfield  {author} {\bibinfo {author} {\bibfnamefont {V.}~\bibnamefont {Veitch}}, \bibinfo {author} {\bibfnamefont {S.~H.}\ \bibnamefont {Mousavian}}, \bibinfo {author} {\bibfnamefont {D.}~\bibnamefont {Gottesman}},\ and\ \bibinfo {author} {\bibfnamefont {J.}~\bibnamefont {Emerson}},\ }\bibfield  {title} {\bibinfo {title} {The resource theory of stabilizer quantum computation},\ }\href {https://doi.org/10.1088/1367-2630/16/1/013009} {\bibfield  {journal} {\bibinfo  {journal} {New Journal of Physics}\ }\textbf {\bibinfo {volume} {16}},\ \bibinfo {pages} {013009} (\bibinfo {year} {2014})}\BibitemShut {NoStop}%
\bibitem [{\citenamefont {Howard}\ and\ \citenamefont {Campbell}(2017)}]{howard2017application}%
  \BibitemOpen
  \bibfield  {author} {\bibinfo {author} {\bibfnamefont {M.}~\bibnamefont {Howard}}\ and\ \bibinfo {author} {\bibfnamefont {E.}~\bibnamefont {Campbell}},\ }\bibfield  {title} {\bibinfo {title} {Application of a resource theory for magic states to fault-tolerant quantum computing},\ }\href {https://doi.org/10.1103/PhysRevLett.118.090501} {\bibfield  {journal} {\bibinfo  {journal} {Physical review letters}\ }\textbf {\bibinfo {volume} {118}},\ \bibinfo {pages} {090501} (\bibinfo {year} {2017})}\BibitemShut {NoStop}%
\bibitem [{\citenamefont {Leone}\ and\ \citenamefont {Bittel}(2024)}]{leone2024stabilizer}%
  \BibitemOpen
  \bibfield  {author} {\bibinfo {author} {\bibfnamefont {L.}~\bibnamefont {Leone}}\ and\ \bibinfo {author} {\bibfnamefont {L.}~\bibnamefont {Bittel}},\ }\bibfield  {title} {\bibinfo {title} {Stabilizer entropies are monotones for magic-state resource theory},\ }\href {https://arxiv.org/abs/2404.11652} {\bibfield  {journal} {\bibinfo  {journal} {arXiv preprint arXiv:2404.11652}\ } (\bibinfo {year} {2024})}\BibitemShut {NoStop}%
\bibitem [{\citenamefont {Mastel}(2023)}]{mastel2023clifford}%
  \BibitemOpen
  \bibfield  {author} {\bibinfo {author} {\bibfnamefont {K.}~\bibnamefont {Mastel}},\ }\bibfield  {title} {\bibinfo {title} {The clifford theory of the $ n $-qubit clifford group},\ }\href {https://arxiv.org/abs/2307.05810} {\bibfield  {journal} {\bibinfo  {journal} {arXiv preprint arXiv:2307.05810}\ } (\bibinfo {year} {2023})}\BibitemShut {NoStop}%
\bibitem [{\citenamefont {Webb}(2016)}]{webb2016clifford}%
  \BibitemOpen
  \bibfield  {author} {\bibinfo {author} {\bibfnamefont {Z.}~\bibnamefont {Webb}},\ }\bibfield  {title} {\bibinfo {title} {The clifford group forms a unitary 3-design},\ }\href {https://doi.org/10.26421/QIC16.15-16-8} {\bibfield  {journal} {\bibinfo  {journal} {Quantum Information and Computation}\ }\textbf {\bibinfo {volume} {16}},\ \bibinfo {pages} {1379} (\bibinfo {year} {2016})}\BibitemShut {NoStop}%
\bibitem [{\citenamefont {Helsen}\ \emph {et~al.}(2018)\citenamefont {Helsen}, \citenamefont {Wallman},\ and\ \citenamefont {Wehner}}]{helsen2018representations}%
  \BibitemOpen
  \bibfield  {author} {\bibinfo {author} {\bibfnamefont {J.}~\bibnamefont {Helsen}}, \bibinfo {author} {\bibfnamefont {J.~J.}\ \bibnamefont {Wallman}},\ and\ \bibinfo {author} {\bibfnamefont {S.}~\bibnamefont {Wehner}},\ }\bibfield  {title} {\bibinfo {title} {Representations of the multi-qubit clifford group},\ }\bibfield  {journal} {\bibinfo  {journal} {Journal of Mathematical Physics}\ }\textbf {\bibinfo {volume} {59}},\ \href {https://doi.org/10.1063/1.4997688} {10.1063/1.4997688} (\bibinfo {year} {2018})\BibitemShut {NoStop}%
\bibitem [{\citenamefont {Heimendahl}\ \emph {et~al.}(2021)\citenamefont {Heimendahl}, \citenamefont {Montealegre-Mora}, \citenamefont {Vallentin},\ and\ \citenamefont {Gross}}]{heimendahl2021stabilizer}%
  \BibitemOpen
  \bibfield  {author} {\bibinfo {author} {\bibfnamefont {A.}~\bibnamefont {Heimendahl}}, \bibinfo {author} {\bibfnamefont {F.}~\bibnamefont {Montealegre-Mora}}, \bibinfo {author} {\bibfnamefont {F.}~\bibnamefont {Vallentin}},\ and\ \bibinfo {author} {\bibfnamefont {D.}~\bibnamefont {Gross}},\ }\bibfield  {title} {\bibinfo {title} {Stabilizer extent is not multiplicative},\ }\href {https://doi.org/10.22331/q-2021-02-24-400} {\bibfield  {journal} {\bibinfo  {journal} {Quantum}\ }\textbf {\bibinfo {volume} {5}},\ \bibinfo {pages} {400} (\bibinfo {year} {2021})}\BibitemShut {NoStop}%
\bibitem [{\citenamefont {Pashayan}\ \emph {et~al.}(2022)\citenamefont {Pashayan}, \citenamefont {Reardon-Smith}, \citenamefont {Korzekwa},\ and\ \citenamefont {Bartlett}}]{fast2022pashayan}%
  \BibitemOpen
  \bibfield  {author} {\bibinfo {author} {\bibfnamefont {H.}~\bibnamefont {Pashayan}}, \bibinfo {author} {\bibfnamefont {O.}~\bibnamefont {Reardon-Smith}}, \bibinfo {author} {\bibfnamefont {K.}~\bibnamefont {Korzekwa}},\ and\ \bibinfo {author} {\bibfnamefont {S.~D.}\ \bibnamefont {Bartlett}},\ }\bibfield  {title} {\bibinfo {title} {Fast estimation of outcome probabilities for quantum circuits},\ }\href {https://doi.org/10.1103/PRXQuantum.3.020361} {\bibfield  {journal} {\bibinfo  {journal} {PRX Quantum}\ }\textbf {\bibinfo {volume} {3}},\ \bibinfo {pages} {020361} (\bibinfo {year} {2022})}\BibitemShut {NoStop}%
\bibitem [{\citenamefont {Reardon-Smith}(2024)}]{reardon2024fermionic}%
  \BibitemOpen
  \bibfield  {author} {\bibinfo {author} {\bibfnamefont {O.}~\bibnamefont {Reardon-Smith}},\ }\bibfield  {title} {\bibinfo {title} {The fermionic linear optical extent is multiplicative for 4 qubit parity eigenstates},\ }\href {https://arxiv.org/abs/2407.20934} {\bibfield  {journal} {\bibinfo  {journal} {arXiv preprint arXiv:2407.20934}\ } (\bibinfo {year} {2024})}\BibitemShut {NoStop}%
\bibitem [{\citenamefont {Reardon-Smith}\ \emph {et~al.}(2023)\citenamefont {Reardon-Smith}, \citenamefont {Oszmaniec},\ and\ \citenamefont {Korzekwa}}]{reardon2023improved}%
  \BibitemOpen
  \bibfield  {author} {\bibinfo {author} {\bibfnamefont {O.}~\bibnamefont {Reardon-Smith}}, \bibinfo {author} {\bibfnamefont {M.}~\bibnamefont {Oszmaniec}},\ and\ \bibinfo {author} {\bibfnamefont {K.}~\bibnamefont {Korzekwa}},\ }\bibfield  {title} {\bibinfo {title} {Improved simulation of quantum circuits dominated by free fermionic operations},\ }\href {https://doi.org/10.48550/arXiv.2307.12702} {\bibfield  {journal} {\bibinfo  {journal} {arXiv preprint arXiv:2307.12702}\ } (\bibinfo {year} {2023})}\BibitemShut {NoStop}%
\bibitem [{\citenamefont {Haug}\ and\ \citenamefont {Piroli}(2023)}]{haug2023stabilizer}%
  \BibitemOpen
  \bibfield  {author} {\bibinfo {author} {\bibfnamefont {T.}~\bibnamefont {Haug}}\ and\ \bibinfo {author} {\bibfnamefont {L.}~\bibnamefont {Piroli}},\ }\bibfield  {title} {\bibinfo {title} {Stabilizer entropies and nonstabilizerness monotones},\ }\href {https://doi.org/10.22331/q-2023-08-28-1092} {\bibfield  {journal} {\bibinfo  {journal} {Quantum}\ }\textbf {\bibinfo {volume} {7}},\ \bibinfo {pages} {1092} (\bibinfo {year} {2023})}\BibitemShut {NoStop}%
\bibitem [{\citenamefont {Gottesman}(1998)}]{gottesman1998theory}%
  \BibitemOpen
  \bibfield  {author} {\bibinfo {author} {\bibfnamefont {D.}~\bibnamefont {Gottesman}},\ }\bibfield  {title} {\bibinfo {title} {Theory of fault-tolerant quantum computation},\ }\href {https://doi.org/10.1103/PhysRevA.57.127} {\bibfield  {journal} {\bibinfo  {journal} {Physical Review A}\ }\textbf {\bibinfo {volume} {57}},\ \bibinfo {pages} {127} (\bibinfo {year} {1998})}\BibitemShut {NoStop}%
\bibitem [{\citenamefont {Wick}(1950)}]{wick1950evaluation}%
  \BibitemOpen
  \bibfield  {author} {\bibinfo {author} {\bibfnamefont {G.-C.}\ \bibnamefont {Wick}},\ }\bibfield  {title} {\bibinfo {title} {The evaluation of the collision matrix},\ }\href {https://doi.org/10.1103/PhysRev.80.268} {\bibfield  {journal} {\bibinfo  {journal} {Physical review}\ }\textbf {\bibinfo {volume} {80}},\ \bibinfo {pages} {268} (\bibinfo {year} {1950})}\BibitemShut {NoStop}%
\bibitem [{\citenamefont {Valiant}(2001)}]{valiant2001quantum}%
  \BibitemOpen
  \bibfield  {author} {\bibinfo {author} {\bibfnamefont {L.~G.}\ \bibnamefont {Valiant}},\ }\bibfield  {title} {\bibinfo {title} {Quantum computers that can be simulated classically in polynomial time},\ }in\ \href {https://doi.org/10.1145/380752.380785} {\emph {\bibinfo {booktitle} {Proceedings of the thirty-third annual ACM symposium on Theory of computing}}}\ (\bibinfo {year} {2001})\ pp.\ \bibinfo {pages} {114--123}\BibitemShut {NoStop}%
\bibitem [{\citenamefont {Terhal}\ and\ \citenamefont {DiVincenzo}(2002)}]{terhal2002classical}%
  \BibitemOpen
  \bibfield  {author} {\bibinfo {author} {\bibfnamefont {B.~M.}\ \bibnamefont {Terhal}}\ and\ \bibinfo {author} {\bibfnamefont {D.~P.}\ \bibnamefont {DiVincenzo}},\ }\bibfield  {title} {\bibinfo {title} {Classical simulation of noninteracting-fermion quantum circuits},\ }\href {https://doi.org/10.1103/PhysRevA.65.032325} {\bibfield  {journal} {\bibinfo  {journal} {Physical Review A}\ }\textbf {\bibinfo {volume} {65}},\ \bibinfo {pages} {032325} (\bibinfo {year} {2002})}\BibitemShut {NoStop}%
\bibitem [{\citenamefont {Biamonte}\ \emph {et~al.}(2017)\citenamefont {Biamonte}, \citenamefont {Wittek}, \citenamefont {Pancotti}, \citenamefont {Rebentrost}, \citenamefont {Wiebe},\ and\ \citenamefont {Lloyd}}]{biamonte2017quantum}%
  \BibitemOpen
  \bibfield  {author} {\bibinfo {author} {\bibfnamefont {J.}~\bibnamefont {Biamonte}}, \bibinfo {author} {\bibfnamefont {P.}~\bibnamefont {Wittek}}, \bibinfo {author} {\bibfnamefont {N.}~\bibnamefont {Pancotti}}, \bibinfo {author} {\bibfnamefont {P.}~\bibnamefont {Rebentrost}}, \bibinfo {author} {\bibfnamefont {N.}~\bibnamefont {Wiebe}},\ and\ \bibinfo {author} {\bibfnamefont {S.}~\bibnamefont {Lloyd}},\ }\bibfield  {title} {\bibinfo {title} {Quantum machine learning},\ }\href {https://doi.org/10.1038/nature23474} {\bibfield  {journal} {\bibinfo  {journal} {Nature}\ }\textbf {\bibinfo {volume} {549}},\ \bibinfo {pages} {195} (\bibinfo {year} {2017})}\BibitemShut {NoStop}%
\bibitem [{\citenamefont {Schuld}\ \emph {et~al.}(2015)\citenamefont {Schuld}, \citenamefont {Sinayskiy},\ and\ \citenamefont {Petruccione}}]{schuld2015introduction}%
  \BibitemOpen
  \bibfield  {author} {\bibinfo {author} {\bibfnamefont {M.}~\bibnamefont {Schuld}}, \bibinfo {author} {\bibfnamefont {I.}~\bibnamefont {Sinayskiy}},\ and\ \bibinfo {author} {\bibfnamefont {F.}~\bibnamefont {Petruccione}},\ }\bibfield  {title} {\bibinfo {title} {An introduction to quantum machine learning},\ }\href {https://doi.org/10.1080/00107514.2014.964942} {\bibfield  {journal} {\bibinfo  {journal} {Contemporary Physics}\ }\textbf {\bibinfo {volume} {56}},\ \bibinfo {pages} {172} (\bibinfo {year} {2015})}\BibitemShut {NoStop}%
\bibitem [{\citenamefont {Cerezo}\ \emph {et~al.}(2021)\citenamefont {Cerezo}, \citenamefont {Arrasmith}, \citenamefont {Babbush}, \citenamefont {Benjamin}, \citenamefont {Endo}, \citenamefont {Fujii}, \citenamefont {McClean}, \citenamefont {Mitarai}, \citenamefont {Yuan}, \citenamefont {Cincio},\ and\ \citenamefont {Coles}}]{cerezo2020variationalreview}%
  \BibitemOpen
  \bibfield  {author} {\bibinfo {author} {\bibfnamefont {M.}~\bibnamefont {Cerezo}}, \bibinfo {author} {\bibfnamefont {A.}~\bibnamefont {Arrasmith}}, \bibinfo {author} {\bibfnamefont {R.}~\bibnamefont {Babbush}}, \bibinfo {author} {\bibfnamefont {S.~C.}\ \bibnamefont {Benjamin}}, \bibinfo {author} {\bibfnamefont {S.}~\bibnamefont {Endo}}, \bibinfo {author} {\bibfnamefont {K.}~\bibnamefont {Fujii}}, \bibinfo {author} {\bibfnamefont {J.~R.}\ \bibnamefont {McClean}}, \bibinfo {author} {\bibfnamefont {K.}~\bibnamefont {Mitarai}}, \bibinfo {author} {\bibfnamefont {X.}~\bibnamefont {Yuan}}, \bibinfo {author} {\bibfnamefont {L.}~\bibnamefont {Cincio}},\ and\ \bibinfo {author} {\bibfnamefont {P.~J.}\ \bibnamefont {Coles}},\ }\bibfield  {title} {\bibinfo {title} {Variational quantum algorithms},\ }\href {https://doi.org/10.1038/s42254-021-00348-9} {\bibfield  {journal} {\bibinfo  {journal} {Nature Reviews Physics}\ }\textbf {\bibinfo {volume} {3}},\ \bibinfo {pages} {625–644} (\bibinfo {year} {2021})}\BibitemShut
  {NoStop}%
\bibitem [{\citenamefont {Cerezo}\ \emph {et~al.}(2022)\citenamefont {Cerezo}, \citenamefont {Verdon}, \citenamefont {Huang}, \citenamefont {Cincio},\ and\ \citenamefont {Coles}}]{cerezo2022challenges}%
  \BibitemOpen
  \bibfield  {author} {\bibinfo {author} {\bibfnamefont {M.}~\bibnamefont {Cerezo}}, \bibinfo {author} {\bibfnamefont {G.}~\bibnamefont {Verdon}}, \bibinfo {author} {\bibfnamefont {H.-Y.}\ \bibnamefont {Huang}}, \bibinfo {author} {\bibfnamefont {L.}~\bibnamefont {Cincio}},\ and\ \bibinfo {author} {\bibfnamefont {P.~J.}\ \bibnamefont {Coles}},\ }\bibfield  {title} {\bibinfo {title} {Challenges and opportunities in quantum machine learning},\ }\bibfield  {journal} {\bibinfo  {journal} {Nature Computational Science}\ }\href {https://doi.org/10.1038/s43588-022-00311-3} {10.1038/s43588-022-00311-3} (\bibinfo {year} {2022})\BibitemShut {NoStop}%
\bibitem [{\citenamefont {Ragone}\ \emph {et~al.}(2024)\citenamefont {Ragone}, \citenamefont {Bakalov}, \citenamefont {Sauvage}, \citenamefont {Kemper}, \citenamefont {Ortiz~Marrero}, \citenamefont {Larocca},\ and\ \citenamefont {Cerezo}}]{ragone2023unified}%
  \BibitemOpen
  \bibfield  {author} {\bibinfo {author} {\bibfnamefont {M.}~\bibnamefont {Ragone}}, \bibinfo {author} {\bibfnamefont {B.~N.}\ \bibnamefont {Bakalov}}, \bibinfo {author} {\bibfnamefont {F.}~\bibnamefont {Sauvage}}, \bibinfo {author} {\bibfnamefont {A.~F.}\ \bibnamefont {Kemper}}, \bibinfo {author} {\bibfnamefont {C.}~\bibnamefont {Ortiz~Marrero}}, \bibinfo {author} {\bibfnamefont {M.}~\bibnamefont {Larocca}},\ and\ \bibinfo {author} {\bibfnamefont {M.}~\bibnamefont {Cerezo}},\ }\bibfield  {title} {\bibinfo {title} {A lie algebraic theory of barren plateaus for deep parameterized quantum circuits},\ }\href {https://doi.org/10.1038/s41467-024-49909-3} {\bibfield  {journal} {\bibinfo  {journal} {Nature Communications}\ }\textbf {\bibinfo {volume} {15}},\ \bibinfo {pages} {7172} (\bibinfo {year} {2024})}\BibitemShut {NoStop}%
\bibitem [{\citenamefont {Fontana}\ \emph {et~al.}(2024)\citenamefont {Fontana}, \citenamefont {Herman}, \citenamefont {Chakrabarti}, \citenamefont {Kumar}, \citenamefont {Yalovetzky}, \citenamefont {Heredge}, \citenamefont {Sureshbabu},\ and\ \citenamefont {Pistoia}}]{fontana2023theadjoint}%
  \BibitemOpen
  \bibfield  {author} {\bibinfo {author} {\bibfnamefont {E.}~\bibnamefont {Fontana}}, \bibinfo {author} {\bibfnamefont {D.}~\bibnamefont {Herman}}, \bibinfo {author} {\bibfnamefont {S.}~\bibnamefont {Chakrabarti}}, \bibinfo {author} {\bibfnamefont {N.}~\bibnamefont {Kumar}}, \bibinfo {author} {\bibfnamefont {R.}~\bibnamefont {Yalovetzky}}, \bibinfo {author} {\bibfnamefont {J.}~\bibnamefont {Heredge}}, \bibinfo {author} {\bibfnamefont {S.~H.}\ \bibnamefont {Sureshbabu}},\ and\ \bibinfo {author} {\bibfnamefont {M.}~\bibnamefont {Pistoia}},\ }\bibfield  {title} {\bibinfo {title} {Characterizing barren plateaus in quantum ansätze with the adjoint representation},\ }\href {https://doi.org/10.1038/s41467-024-49910-w} {\bibfield  {journal} {\bibinfo  {journal} {Nature Communications}\ }\textbf {\bibinfo {volume} {15}},\ \bibinfo {pages} {7171} (\bibinfo {year} {2024})}\BibitemShut {NoStop}%
\bibitem [{\citenamefont {Larocca}\ \emph {et~al.}(2025)\citenamefont {Larocca}, \citenamefont {Thanasilp}, \citenamefont {Wang}, \citenamefont {Sharma}, \citenamefont {Biamonte}, \citenamefont {Coles}, \citenamefont {Cincio}, \citenamefont {McClean}, \citenamefont {Holmes},\ and\ \citenamefont {Cerezo}}]{larocca2024review}%
  \BibitemOpen
  \bibfield  {author} {\bibinfo {author} {\bibfnamefont {M.}~\bibnamefont {Larocca}}, \bibinfo {author} {\bibfnamefont {S.}~\bibnamefont {Thanasilp}}, \bibinfo {author} {\bibfnamefont {S.}~\bibnamefont {Wang}}, \bibinfo {author} {\bibfnamefont {K.}~\bibnamefont {Sharma}}, \bibinfo {author} {\bibfnamefont {J.}~\bibnamefont {Biamonte}}, \bibinfo {author} {\bibfnamefont {P.~J.}\ \bibnamefont {Coles}}, \bibinfo {author} {\bibfnamefont {L.}~\bibnamefont {Cincio}}, \bibinfo {author} {\bibfnamefont {J.~R.}\ \bibnamefont {McClean}}, \bibinfo {author} {\bibfnamefont {Z.}~\bibnamefont {Holmes}},\ and\ \bibinfo {author} {\bibfnamefont {M.}~\bibnamefont {Cerezo}},\ }\bibfield  {title} {\bibinfo {title} {A review of barren plateaus in variational quantum computing},\ }\href {https://doi.org/10.1038/s42254-025-00813-9} {\bibfield  {journal} {\bibinfo  {journal} {Nature Reviews Physics}\ }\textbf {\bibinfo {volume} {3}},\ \bibinfo {pages} {625–644} (\bibinfo {year} {2025})}\BibitemShut {NoStop}%
\bibitem [{\citenamefont {Bermejo}\ \emph {et~al.}(2024)\citenamefont {Bermejo}, \citenamefont {Braccia}, \citenamefont {Rudolph}, \citenamefont {Holmes}, \citenamefont {Cincio},\ and\ \citenamefont {Cerezo}}]{bermejo2024quantum}%
  \BibitemOpen
  \bibfield  {author} {\bibinfo {author} {\bibfnamefont {P.}~\bibnamefont {Bermejo}}, \bibinfo {author} {\bibfnamefont {P.}~\bibnamefont {Braccia}}, \bibinfo {author} {\bibfnamefont {M.~S.}\ \bibnamefont {Rudolph}}, \bibinfo {author} {\bibfnamefont {Z.}~\bibnamefont {Holmes}}, \bibinfo {author} {\bibfnamefont {L.}~\bibnamefont {Cincio}},\ and\ \bibinfo {author} {\bibfnamefont {M.}~\bibnamefont {Cerezo}},\ }\bibfield  {title} {\bibinfo {title} {Quantum convolutional neural networks are (effectively) classically simulable},\ }\href {https://arxiv.org/abs/2408.12739} {\bibfield  {journal} {\bibinfo  {journal} {arXiv preprint arXiv:2408.12739}\ } (\bibinfo {year} {2024})}\BibitemShut {NoStop}%
\bibitem [{\citenamefont {Braccia}\ \emph {et~al.}(2024)\citenamefont {Braccia}, \citenamefont {Bermejo}, \citenamefont {Cincio},\ and\ \citenamefont {Cerezo}}]{braccia2024computing}%
  \BibitemOpen
  \bibfield  {author} {\bibinfo {author} {\bibfnamefont {P.}~\bibnamefont {Braccia}}, \bibinfo {author} {\bibfnamefont {P.}~\bibnamefont {Bermejo}}, \bibinfo {author} {\bibfnamefont {L.}~\bibnamefont {Cincio}},\ and\ \bibinfo {author} {\bibfnamefont {M.}~\bibnamefont {Cerezo}},\ }\bibfield  {title} {\bibinfo {title} {Computing exact moments of local random quantum circuits via tensor networks},\ }\href {https://doi.org/10.1007/s42484-024-00187-8} {\bibfield  {journal} {\bibinfo  {journal} {Quantum Machine Intelligence}\ }\textbf {\bibinfo {volume} {6}},\ \bibinfo {pages} {54} (\bibinfo {year} {2024})}\BibitemShut {NoStop}%
\bibitem [{\citenamefont {Deneris}\ \emph {et~al.}(2025)\citenamefont {Deneris}, \citenamefont {Braccia}, \citenamefont {Bermejo}, \citenamefont {Mele},\ and\ \citenamefont {Cerezo}}]{deneris2025analyzing}%
  \BibitemOpen
  \bibfield  {author} {\bibinfo {author} {\bibfnamefont {A.~E.}\ \bibnamefont {Deneris}}, \bibinfo {author} {\bibfnamefont {P.}~\bibnamefont {Braccia}}, \bibinfo {author} {\bibfnamefont {P.}~\bibnamefont {Bermejo}}, \bibinfo {author} {\bibfnamefont {A.~A.}\ \bibnamefont {Mele}},\ and\ \bibinfo {author} {\bibfnamefont {M.}~\bibnamefont {Cerezo}},\ }\bibfield  {title} {\bibinfo {title} {Analyzing the free states of one quantum resource theory as resource states of another},\ }\href@noop {} {\bibfield  {journal} {\bibinfo  {journal} {Manuscript in preparation}\ } (\bibinfo {year} {2025})}\BibitemShut {NoStop}%
\bibitem [{\citenamefont {Hall}(2013)}]{hall2013lie}%
  \BibitemOpen
  \bibfield  {author} {\bibinfo {author} {\bibfnamefont {B.~C.}\ \bibnamefont {Hall}},\ }\href@noop {} {\emph {\bibinfo {title} {Lie groups, Lie algebras, and representations}}}\ (\bibinfo  {publisher} {Springer},\ \bibinfo {year} {2013})\BibitemShut {NoStop}%
\bibitem [{\citenamefont {Somma}(2005)}]{somma2005quantum}%
  \BibitemOpen
  \bibfield  {author} {\bibinfo {author} {\bibfnamefont {R.~D.}\ \bibnamefont {Somma}},\ }\bibfield  {title} {\bibinfo {title} {Quantum computation, complexity, and many-body physics},\ }\href {https://arxiv.org/abs/quant-ph/0512209} {\bibfield  {journal} {\bibinfo  {journal} {arXiv preprint quant-ph/0512209}\ } (\bibinfo {year} {2005})}\BibitemShut {NoStop}%
\bibitem [{\citenamefont {Harrow}(2013)}]{harrow2013church}%
  \BibitemOpen
  \bibfield  {author} {\bibinfo {author} {\bibfnamefont {A.~W.}\ \bibnamefont {Harrow}},\ }\bibfield  {title} {\bibinfo {title} {The church of the symmetric subspace},\ }\href {https://arxiv.org/abs/1308.6595} {\bibfield  {journal} {\bibinfo  {journal} {arXiv preprint arXiv:1308.6595}\ } (\bibinfo {year} {2013})}\BibitemShut {NoStop}%
\end{thebibliography}%

\clearpage
\newpage

\makeatletter
\close@column@grid
\makeatother
\cleardoublepage
\newpage
\onecolumngrid
\renewcommand\appendixname{Appendix}
\appendix

\setcounter{theorem}{1}

\section*{Appendices}

\section{Proof of Theorem~\ref{prop:compress}}

Here we present a proof for Theorem~\ref{prop:compress}, which we state again for convenience.

\begin{theorem}[MaxEnt and free states]\label{prop:compress-si}
Let  $\rho$ be a pure free state. If we find a pure state  $\sigma$ which solves the MaxEnt problem of Eq.~\eqref{eq:max-ent} for $\Lambda=\dim(\mathfrak{g})$, then $\sigma=\rho$.  
\end{theorem}

\begin{proof}
Consider a QRT where the free operations arise from a representation $R$ of a Lie group $G$ with associated Lie algebra $\mathfrak{g}$. Then, let $\rho=\dya{\psi}$ with $\ket{\psi}$ a free, or generalized coherent state. By definition, there exists a unitary $U\in R(G)$ such that $\ket{\psi}=U\ket{{\rm hw}}$, where $\ket{{\rm hw}}$ is the highest (or lowest) weight of $R(\mathfrak{g})$. 
    
To begin, let us show that if a pure state $\sigma=\dya{\phi}$ solves the MaxEnt problem of Eq.~\eqref{eq:max-ent} for $\Lambda=\dim(\mathfrak{g})$, then $\ket{\phi}$ is a generalized coherent state. We note that  all the expectation values in the irrep $\mathbb{C}R(\mathfrak{g})$ for $\rho$ and $\sigma$ match (as per the statement of Theorem~\ref{prop:compress-si}), then we have
\begin{equation}
\PC_{\mathbb{C}R(\mathfrak{g})}(\sigma)=\PC_{\mathbb{C}R(\mathfrak{g})}(\rho)\,,
\end{equation}
which simply follows from Definition~\ref{def:irrep-purity}. Then simply using   Theorem~\ref{theo:algebra-purity} (see Ref.~\cite{barnum2004subsystem} for its proof), we have that $\ket{\phi}$ is a free, generalized coherent, state.

Next, we need to show that all other expectation values can be obtained from those in $\mathbb{C}R(\mathfrak{g})$. This is known as Wick's theorem, whose proof we briefly sketch here. We also refer the reader to Ref.~\cite{somma2005quantum}. Without loss of generality, Let $R(\g)=\mathfrak{h}\oplus \mathfrak{g}_+\oplus \mathfrak{g}_-$ be the Cartan-Weyl decomposition of $R(\mathfrak{g})$ (where, on the right-hand side, we have left the representation notation implicit.). Here,  $\mathfrak{h}={\rm span}_\mathbb{R}i\{H_k\}$ denotes the Cartan subalgebra, and $\mathfrak{g}_+={\rm span}_\mathbb{R}i\{E_{\alpha }\}$ and $\mathfrak{g}_-={\rm span}_\mathbb{R}i\{E_{-\alpha }\}$ the spaces of raising and lowering operators, respectively. We will assume, without loss of generality that $\ket{\phi}$ is the highest weight for this particular choice of Cartan $\mathfrak{h}$.  As such, we know that the following equation holds
\begin{equation}\label{eq:cases}
E_{\alpha}\ket{\phi}=\bra{\phi}E_{-\alpha }=0,\quad \forall \alpha\quad \text{and} \quad H_k\ket{\phi}=h_k\ket{\phi}\,,
\end{equation}
where $h_k$ are the weights of $\ket{\phi}$. Then, let $O=O_1 O_2\cdots O_M$, with $O_1,O_2,\cdots,O_M\in iR(\g)$ be operators in the algebra. When computing the expectation value $\bra{\phi}O_1 O_2\cdots O_k\ket{\phi}$ (which correspond to all non-trivially-zero expectation values) we decompose each of the operators into elements of the Cartan, as well as lowering and raising operators. That is,
\begin{equation}
    O_j=\sum_k c_k^j H_k  +\sum_\alpha c_\alpha^j E_\alpha+{c_\alpha^{j}}^* E_{-\alpha}\,.
\end{equation}
Consequently, their products can be expressed in terms of summations of products of  elements of the Cartan, lowering and raising operators. Given that these act on the highest-weight vector $\ket{\phi}$, only a few of those terms will be non-zero as per Eq.~\eqref{eq:cases}. For instance, when $M=2$, terms of the form $\bra{\Phi}H_k H_{k'}\ket{\phi}$ and $\bra{\Phi}E_\alpha E_{-\alpha}\ket{\phi}$ are non-zero, whereas terms such as $\bra{\Phi}E_{-\alpha} E_{\alpha}\ket{\phi}$ or $\bra{\Phi}E_{\alpha} E_{\alpha}\ket{\phi}$ vanish. In fact, we can always use the algebra's commutation relations to ``move to the left'' all the lowering operators and ``move to the right'' all the raising operators, as these
will vanish when acting on the highest-weight state. At the end of the day, this will lead to an expectation value which only depends on the components of the root vector and the algebra's structure coefficients, indicating that any expectation value of the form $\bra{\phi}O_1 O_2\cdots O_k\ket{\phi}$ will be determined by the expectation values in $R(\mathfrak{g})$. Since those values match to those in $\ket{\psi}$, we find that all expectation values for $\rho$ and $\sigma$ match, indicating that these two states are equal. 

\end{proof}

\section{GFD purities for a Haar random state in the irreps of $\LC(\HC)$}

In this section we show how to compute expectation values for pure Haar random states over irreps of $\LC(\HC)$ for any representation $R$ of a group $G$. Here we also recall that $\HC=\mathbb{C}^d$. In particular, we will assume that operator spaces is decomposed as
\begin{equation}\label{eq:decomp-irrep-L}
    \LC(\HC)=\bigoplus_\a \LC_\a
\end{equation}
where $\a$ contains the irrep and multiplicity information. Without loss of generality we will assume that $\LC_0={\rm span}_{\mathbb{C}}\{\id\}$. Then, given a Hermitian orthonormal basis $\{B_\a^{\mu}\}_{\mu=1}^{\dim(\LC_\a)}$ of $\LC_\a$, it directly follows that  
\begin{equation}
    \Tr[{\rm SWAP}B_\a^{\mu}\otimes B_{\a'}^{\mu'}]=\Tr[B_\a^{\mu} B_{\a'}^{\mu'}]=\delta_{\a,\a'}\delta_{\mu,\mu'}\,,
\end{equation}

where we have defined the operator ${\rm SWAP}\in\LC(\HC^{\otimes 2})$ as ${\rm SWAP}=\sum_{i,j=1}^d\ket{ij}\bra{ji}$. We then find 
\begin{equation}
    \sum_{\mu=1}^{\dim(\LC_\a)}\Tr[{\rm SWAP}B_\a^{\mu}\otimes B_\a^{\mu}]=\dim(\LC_\a).
\end{equation}

Next, let us recall that the purity of a state $\rho$ on the $\a$-th irrep can be expressed as 
\begin{equation}
    \PC_\a(\rho)=\sum_{\mu=1}^{\dim(\LC_\a)}\Tr[\rho B_\a^{\mu}]^2=\sum_{\mu=1}^{\dim(\LC_\a)}\Tr[\rho^{\otimes 2} {B_\a^{\mu}}^{\otimes 2}]\,.
\end{equation}
For a state sampled according to the Haar measure over $\HC$ we know that~\cite{harrow2013church} 
\begin{equation}\label{eq:Haar}
    \mathbb{E}_\HC[\rho^{\otimes 2}]=\frac{\id\otimes \id +{\rm SWAP}}{d(d+1)}\,.
\end{equation}
Hence, we find that the average GFD purity is  
\begin{equation}\label{ap-eq:haar_random_Lie_gfd}
   \mathbb{E}_\HC[ \PC_\a(\rho)]=\sum_{\mu=1}^{\dim(\LC_\a)}\frac{\Tr[(\id\otimes \id +{\rm SWAP}){B_\a^{\mu}}^{\otimes 2}]}{d(d+1)} =\begin{cases}
       \frac{1}{d} & \text{if $\a=0$}\\
       \frac{\dim(\LC_\a)}{d(d+1)} & \text{else}\\
   \end{cases}\,,
\end{equation}
where we have also used that $\Tr[(\id\otimes \id){B_\a^{\mu}}^{\otimes 2}]=0$ $\forall \mu$ if $\a\neq 0$, which immediately follows from the fact that for any operator $B\in\LC(\HC)$ its trace component $B_{{\rm trace}}=\frac{\Tr[B]}{d}\id$ lives in $\LC_0$. 

\section{GFD purities for a free state in the irreps of $\LC(\HC)$ for a Lie group-based QRT}

In this brief section we present a simple observation that can be used to compute GFD purities for the free states of a Lie group-based QRT. In particular, assume that the free operations are given by a representation $R$ of a Lie group $G$, whose associated Lie algebra is $\mathfrak{g}$. Then, let us denote as $\mathfrak{h}\subset \mathfrak{g}$ the Cartan subalgebra. We know from Eq.~\eqref{eq:cases} that the highest weight of $\mathfrak{h}$ only has non-zero expectation value with the weight-zero elements of the Cartan. Similarly, using Eq.~\eqref{eq:decomp-irrep-L} it will follow that within each $\LC_\a$ only the operators in its weight-zero subspace will have non-zero expectation value. 

As an example, consider the QRT of multipartite entanglement. As discussed in the main text, the free operations arise from the standard representation $R$ of the group $G=\mathbb{SU}(2)\times \mathbb{SU}(2)\times \cdots\times \mathbb{SU}(2)$, whose associated Lie algebra is $\mathfrak{g}=\mathfrak{su}(2)\oplus \mathfrak{su}(2)\oplus \cdots\oplus \mathfrak{su}(2)$. The Cartan subalgebra of each component of the algebra is $\{Z_i\}_{i=1}^n$, and the associated highest weight state is $\ket{0}^{\otimes n}$. The operator space decomposes into $2^n$ multiplicity-free irreps as
\begin{equation}\label{eq:irreps-L-nq-si}
    \LC(\HC)=\bigoplus_{\alpha \in\{0,1\}^{\otimes n}}\LC_{\alpha}\,,
\end{equation}
where $\LC_{\alpha}$ is composed of all operators of the form $P_1^{\alpha_1}\otimes P_2^{\alpha_2}\otimes \cdots \otimes P_n^{\alpha_n}$. Within each irrep the weight-zero subspace contains all operators commuting with $\{Z_i\}_{i=1}^n$, i.e., all tensor products of identities and $Z$ operators. It is not hard to see that there is a single one such operator per $\LC_{\alpha}$, and that it appears in $\dya{0}^{\otimes n}=\left(\frac{\id + Z}{2}\right)^{\otimes n}$. Therefore
\begin{equation}
    \PC_\alpha(\dya{0}^{\otimes n})=\frac{1}{2^n}\,,
\end{equation}
where the factor $1/2^n$ arises from the normalization of the basis elements of $\LC_\alpha$. Hence, the purity in all irreps with Hamming weight $w(\alpha)$ equal to $k$ is equal to the number of such irreps. Since there are $\binom{n}{k}$ irreps with  $w(\alpha)=k$ we find
\begin{equation}
   \PC_{k}(\dya{0}^{\otimes n})= \sum_{\a:\,w(\alpha)=k} \PC_{\alpha}(\rho_P)=\frac{\binom{n}{k}}{2^n}\,.
\end{equation}
As discussed thoroughly in the main text, the GFD purities are invariant under the action of the chosen free group $R(G)$,  therefore all the tensor product free states of the multipartite entanglement QRT have the same GFD purities as  $\dya{0}^{\otimes n}$, matching the results presented in the Methods.

Using this example as a blueprint, one can work out the GFD purities of the free states of any Lie group-based QRT.

\section{GFD purities for bipartite entanglement in two-qubit states}

Given the relative simplicity of this case, all calculations can be achieved by directly computing the expectation values. 

\section{GFD purities for multipartite entanglement in $n$-qubit states}

Here we derive the GFD purities for a generic tensor product state $\rho_P=\bigotimes_{\mu=1}^n \dya{\psi_\mu}$, the GHZ state $\rho_{{\rm GHZ}}=\dya{{\rm GHZ}}$ with $
\ket{{\rm GHZ}}=\frac{1}{\sqrt{2}}(\ket{0}^{\otimes n}+\ket{1}^{\otimes n})$, the $W$ state $\rho_{{\rm W}}=\dya{{\rm W}}$ with $\ket{W}=
\frac{1}{\sqrt{n}}\sum_{\mu=1}^{n}X_i|0 \rangle^{\otimes n}$, and an $n$-qubit Haar random state $\rho_H=\dya{\psi_H}$.

First, we note that the cases of tensor product (free) and Haar random states were addressed in the previous sections. Then, let us consider the GHZ state defined as
\begin{equation}
    \ket{{\rm GHZ}}=\frac{1}{\sqrt{2}}(\ket{0}^{\otimes n}+\ket{1}^{\otimes n})\,.
\end{equation}
Using $\dya{0}=\frac{\id+Z}{2}$, $\ketbra{0}{1}=\frac{X+iY}{2}$, and $\ketbra{1}{0}=\frac{X-iY}{2}$ one can check that $\rho_{{\rm GHZ}}$ can only have non-zero overlap with any operator of the form $\{\id,Z\}^{\otimes n}$ if $Z$ appears an even number of times, and with any operator of the form $\{X,Y\}^{\otimes n}$ with an even number of $Y$'s. Therefore, we find for all $\alpha$ with $w(\alpha)<n-1$ that 
\begin{equation}
    \PC_{\alpha}({\rm GHZ})=\begin{cases}
        \frac{1}{2^n}\, & \text{if $w(\alpha)$ is even}\\
        0\, & \text{if $w(\alpha)$ is odd}
    \end{cases}\,.
\end{equation}
The cases $w(\alpha)=n-1,n$ depend on the parity of $n$. For $n$ even we find
\begin{equation}
    \PC_{\alpha}({\rm GHZ})=\begin{cases}0\, & \text{if $w(\alpha)=n-1$}\\
    \frac{1}{2}+\frac{1}{2^n}\, & \text{if $w(\alpha)=n$}
    \end{cases}\,.
\end{equation}
Then, for odd $n$  
\begin{equation}
    \PC_{\alpha}({\rm GHZ})=\begin{cases}\frac{1}{2^n}\, & \text{if $w(\alpha)=n-1$}\\
    \frac{1}{2}\, & \text{if $w(\alpha)=n$}
    \end{cases}\,.
\end{equation}
Aggregating the purities by the irreps' hamming weight $w(\a)=k$ one recovers the $\PC_k({\rm GHZ})$ presented in the Methods.

Next, we consider the $W$ state given by 
\begin{equation}
|W\rangle = \frac{1}{\sqrt{n}} \sum_{i=1}^{n}X_i\ket{0}^{\otimes n}\,.
\end{equation}
Here, we find it convenient to define as $\BC_\alpha$ an orthonormal basis of Paulis for $\LC_\alpha$, and as $\BC_\alpha^z\subseteq \BC_\alpha$ the set of the Paulis composed only of identities and Z operators therein (i.e., the basis for the weight-zero subspace). The GFD purity in the $\alpha$-th irrep can be expressed as
\begin{align}
\PC_{\alpha}(|W\rangle \langle W|) = \frac{1}{n^2}\sum_{P\in\BC_\alpha }\Tr[\sum_{i,j}\dya{0}^{\otimes n}X_i P X_j] ^2 =
\frac{1}{n^2} \sum_{P\in \BC_\alpha}  \left( \sum_{i=1}^N A_i(P) + \sum_{i\neq j}^N B_{i,j}(P) \right) ^2\,,
\end{align}
where we have defined $A_i(P) = \Tr[\dya{0}^{\otimes n}X_i P X_i]$, $B_{i,j}(P) = \Tr[\dya{0}^{\otimes n}X_i P X_j]$.
Now, let us notice that $\dya{0}^{\otimes n}$ only has support over $\BC_\alpha^z$, thus the only Paulis $P\in\BC_\a$ contributing to the GFD purities will be those such that $X_iPX_j\in\BC_\alpha^z$, modulo signs.
First, let us consider the $A_i$ terms. 
Given that the $i$-th component of the Pauli string $P$ is conjugated by $X$, and that the ensuing Pauli must belong to $\BC_\alpha^z$ to contribute non-vanishingly to the GFD purity, one finds, for $P\in\BC_\a$
\begin{equation}
    A_i(P) = \frac{1}{\sqrt{2^n}}(\delta_{x_i,0}-\delta_{x_i,1})\delta_{P, \BC_\a^z}\,.
\end{equation}
Indeed, $P$ must belong to $\BC_\a^z$, a condition that we denote via the Kronecker delta $\delta_{P, \BC_\a^z}$, leaving only two options for its $i$-th element: $\id$ if the irrep $\alpha\in\{0,1\}^{\otimes n}$ has $\alpha_i=0$, and $Z$ otherwise. The latter yields the sign difference, while the $1/\sqrt{2^n}$ factor comes from the Pauli expansion of the $\dya{0}^{\otimes n}$ state.
Next, consider the $B_{i,j}$ terms.
Since $i\neq j$, now both the $i$-th and $j$-th components of $P$ multiply, rather than conjugate, by $X$, and hence they must belong to $\{X,Y\}$ to produce an element of $\BC_\a^z$. Furthermore, notice that the ``mixed terms'' $X_iY_j$ and $Y_iX_j$ will always appear in the sum over $\BC_\a$ and cancel each other out, as it is easy to see that they yield the same contribution but with opposite signs. Thus, we can just consider the ``square terms'' $X_iX_j$ and $Y_iY_j$, leading to, for $P\in\BC_\a$
\begin{equation}
    B_{i,j}(P) = \delta_{\alpha_i,1}\delta_{\alpha_j,1}\frac{1}{\sqrt{2^n}}(\delta_{P,P_{(\bar{ij})}^zX_iX_j} + \delta_{P,P_{(\bar{ij})}^zY_iY_j})\,,
\end{equation}
where we have defined as $\BC_{(\bar{ij})}^zX_iX_j\subseteq\BC_\a$ the set of operators composed of identity and $Z$ on all qubit sites except the $i$-th and $j$-th ones, with the latter being $X$ operators. Then, the Kronecker delta $\delta_{P,\BC_{(\bar{ij})}^zX_iX_j}$ is defined as being non-zero only if $P$ belongs to such subset, and analogously for the $Y$ terms.

Combining the previous results leads to
\begin{align}
&\PC_{\alpha}(\dya{W}) = \frac{1}{n^2 2^n}\sum_{P\in\BC_\a} \left(\sum_{i=1}^{n}(\delta_{\alpha_i,0} - \delta_{\alpha_i,1})\delta_{P, \BC_\a^z} \right.+2\sum_{i<j}\delta_{\alpha_i,1}\delta_{\alpha_j,1}\left.(\delta_{P,\BC_{(\bar{ij})}^zX_iX_j} + \delta_{P,\BC_{(\bar{ij})}^zY_iY_j}) \right) ^2\,.
\end{align}

Since all the delta terms are independent, we can carry out the sum, leading to
\begin{equation}
\PC_{\alpha}(\dya{W}) = \frac{(n-2w(\alpha))^2 + 8\binom{w(\alpha)}{w(\alpha)-2}}{n^22^n} \,.
\end{equation}
Notice that, consistently with $\ket{W}$ being translationally invariant, the GFD purity in each irrep $\a$ with fixed Hamming weight $w(\alpha)$ is the same. Hence, we can readily obtain the aggregated purities in all the irreps with the same weight $w(\alpha)=k$ as
\begin{equation}
\PC_{k}(\dya{W}) = \frac{(n-2k)^2 + 8\binom{k}{k-2}}{n^22^n}\binom{n}{k} \,,
\end{equation}
leading to the result presented in the Methods section.

\section{GFD purities for fermionic Gaussianity}

First, let us consider the pure Gaussian state $\ket{0}^{\otimes n}$. As noted above, this state only has component in the irreps containing weight-zero subspaces composed of identity and $Z$ operators. Given that any such Pauli $P$ can be expressed as products of $Z_i$ with $i=1,\ldots,n$, and that each $Z_i$ is composed of two Majoranas ($Z_i=c_{2i-1}c_{2i}$), then we know that it will be expressed as a product of $2w(P)$ distinct Majorana operators and hence belong to the irrep $\LC_{2w(P)}$, with $w(P)$ the weight of the Pauli $P$. From the previous, we find that a given irrep $\LC_\alpha$ contains a weight-zero subspace of size $\binom{n}{\alpha/2}$ and hence we recover
\begin{equation}\label{eq:pur-gauss-si}
\PC_\alpha(\dya{0}^{\otimes n})=\frac{\binom{n}{\alpha/2}}{2^n}\,.
\end{equation}

Next, let us consider the GHZ state, which we can expand into Paulis  as
\begin{align}
    \dya{{\rm GHZ}}=\frac{1}{2}(\dya{0}^{\otimes n}+\dya{1}^{\otimes n}+\ket{0}\bra{1}^{\otimes n}+\ket{1}\bra{0}^{\otimes n})=\frac{1}{2^{n}}\left(\sum_{\substack{\vec{x}\in\{0,1\}^{\otimes n}\\ w(\vec{x})={\rm even}}}Z^{\vec{x}}+(-i)^{w(x)}X^{\vec{\bar{x}}}Y^{\vec{x}}\right)\,,
\end{align}
where we defined $\vec{\bar{x}}$ as the negation of the bitstring $\vec{x}$. Notice that $n$ here must be even, or else the GHZ state would not have definite parity, resulting in it not being fermionic. The number of Majorana modes $n$ being even also fixes the GHZ state parity to +1, meaning $\ket{\rm GHZ}\in\HC_+$, implying that all the GFD purities $\PC_\a$ for odd $\a$ are zero.
Now, from the expansion above we can see that all the terms $X^{\vec{\bar{x}}}Y^{\vec{x}}$ belong to $\LC_n$, as they correspond to the $2^n$ possible choices of different Majorana operators at each site, while the diagonal terms $Z^{\vec{x}}$ will be in $\LC_\alpha$ with $\alpha$ a multiple of 4. Hence, we find 
\begin{equation}\label{eq:ghz_fent_1}
    \PC_\alpha(\dya{{\rm GHZ}})=\frac{\binom{n}{\alpha/2}}{2^n}\,,
\end{equation}
if $\alpha$ is a multiple of four, otherwise 
\begin{equation}\label{eq:ghz_fent_2}
        \PC_n(\dya{{\rm GHZ}})=\frac{2^{n-1}+\binom{n}{n/2}}{2^n}\,.
\end{equation}
Putting everything together recovers the result presented in the Methods.

Now let us restrict to the case when $n$ is a multiple of four, and consider the parametrized extent state
\begin{equation}
    \ket{\chi(\gamma)}=(\cos(\gamma/4)\ket{0}^{\otimes 4}+\sin(\gamma/4)\ket{1}^{\otimes 4})^{\otimes n/4}\,.
\end{equation}
To compute the GFD irrep purities, we find it convenient to first note that the density matrix $\rho_\gamma=\dya{\chi(\gamma)}$ can be expressed as
\begin{align}
    \rho_\gamma=\frac{1}{2^{n}}\Big(& \id^{\otimes 4}+\cos(\gamma/2)\sum_i Z_i+\sin(\gamma/2)(X^{\otimes 4}+Y^{\otimes 4})-{\sin(\gamma/2)}\sum_{i<j<k<l}X_iX_jY_kY_l+\sum_{i<j} Z_iZ_j\nonumber\\&+\cos(\gamma/2)\sum_{i<j<k} Z_iZ_jZ_k+Z^{\otimes 4}\Big)^{\otimes n/4}\,.\nonumber
\end{align}
Then, we note that the operators in the parentheses are such that $\id^{\otimes 4}$ is the product of zero Majoranas; $Z_i$ is the product of two; $X^{\otimes 4},Y^{\otimes 4},X_iX_jY_kY_l, Z_iZ_j$ are the product of four; $Z_iZ_jZ_k$ is the product of six, and $Z^{\otimes 4}$ is the product of eight. When expanding the tensor product, we can thus know to which module each term will belong by expressing it as a product of the operators above.
Moreover, we note that there are  $\binom{4}{1}$ terms that can be expressed as a product of two Majoranas, and that each has a weight  $\cos(\gamma/2)$. Same for the case of six Majoranas. The eight Majoranas operator $Z^{\otimes 4}$ is unique and has weight one, while there are $\binom{4}{2}$ Paulis with weight one that are a product of four Majoranas, and $2+\binom{4}{2}$ of them with weight of module $\sin(\gamma/2)$.
Noticing that the weights will be squared when computing the purities, we find it convenient to define the quantities
\begin{align}
    N_2(\gamma)=\binom{4}{1}\cos^2(\gamma/2)\,,\quad 
    N_4(\gamma)=\left(\binom{4}{2}+\left(2+{\binom{4}{2}}\right)\sin^2(\gamma/2)\right)\,,\quad 
    N_6(\gamma)=\binom{4}{1}\cos^2(\gamma/2)\,,\quad N_8=1\,.
\end{align}
A simple counting argument then reveals that the module purity of this state is given by
\small
\begin{align}
\PC_\alpha(\rho_\gamma)=\frac{1}{2^n}\sum_{i_8=0}^{\lfloor \frac{\alpha}{8} \rfloor} \sum_{i_6=0}^{\lfloor \frac{\alpha}{6} \rfloor} \sum_{i_4=0}^{\lfloor \frac{\alpha}{4} \rfloor} \sum_{i_2=0}^{\lfloor \frac{\alpha}{2} \rfloor}
    \binom{\frac{n}{4}}{i_8} \binom{\frac{n}{4}-i_8}{i_6}  \binom{\frac{n}{4}-i_8-i_6}{i_4} \binom{\frac{n}{4}-i_8-i_6-i_4}{i_2}  N_2(\gamma)^{i_2}N_4(\gamma)^{i_4}N_6(\gamma)^{i_6}N_8^{i_8}\delta_{\a,8i_8+6i_6+4i_4+2i_2}\,.\nonumber
\end{align}
\normalsize

Finally, we consider a Haar random state with defined fermionic parity (i.e., a state sampled according to the Haar measure over $\HC_+$). In particular, we can start with any Haar random state $\ket{\psi}$ over the full Hilbert space $\HC$ and project it into the even parity sub-sector as 
\begin{equation*}
\ket{\phi} = \frac{\ket{\psi} + Z^{\otimes n}\ket{\psi}}{\norm{\ket{\psi} + Z^{\otimes n}\ket{\psi}}} =  \frac{\ket{\psi} + Z^{\otimes n}\ket{\psi}}{\sqrt{2+2\bra{\psi}Z^{\otimes n}\ket{\psi}}} \,.
\end{equation*}
Now, the average GFD purities over the fermionic Gaussianity QRT irreps $\LC_\a$ would then read
\begin{equation}
\mathbb{E}_{\ket{\phi}\sim\HC_+}[ \PC_\a(\dya{\phi})] = 
\mathbb{E}_{\ket{\phi}\sim\HC}[\frac{\PC_\a((\ket{\psi} + Z^{\otimes n}\ket{\psi})(\bra{\psi} + \bra{\psi}Z^{\otimes n}))}{2+2\bra{\psi}Z^{\otimes n}\ket{\psi}}]\,.
\end{equation}
At this point we notice that, given two functions $f$ and $g$ of a random variable $x$, the expectation value over the latter of the ratio $f/g$ can be expanded as
\begin{equation}
    \mathbb{E} [\frac{f}{g}] \approx \frac{\mathbb{E} [f]}{\mathbb{E} [g]} + \frac{\mathbb{E}[f]\,\mathrm{Var}(g)}{\mathbb{E}[g]^3} - \frac{\mathrm{Cov}(f,g)}{\mathbb{E}[g]^2} + \cdots \, ,
\end{equation}
where higher order terms have been dropped.
Since in the case at hand, the expectation values are carried over Haar random states in a $2^n$-dimensional Hilbert space, one can show that the corrections to the zeroth-order approximation $\frac{\mathbb{E} [f]}{\mathbb{E} [g]}$ are going to be exponentially small.
Hence, let us put ourselves in a sufficiently large $n$ regime and rewrite
\begin{equation*}
\ket{\phi} = \frac{1}{c} (\ket{\psi} + Z^{\otimes n}\ket{\psi})\,,
\end{equation*}
where $c$ now is a constant normalization coefficient which we will fix by imposing, on average, the constraint $|\langle \phi|\phi\rangle| =1$. That is, we need to solve for $c$ in the equation
\begin{equation}
\mathbb{E}_{\ket{\psi}\sim\HC}\left[\frac{1}{c^2} (2+2\bra{\psi}Z^{\otimes n}\ket{\psi})\right]=1\,.  
\end{equation}
Using Weingarten calculus~\cite{mele2023introduction} we can find that $\mathbb{E}_{\ket{\psi}\sim\HC}[\bra{\psi}Z^{\otimes n}\ket{\psi}]=0$, leading to $c = \sqrt{2}$.

Now, let us denote as $\BC_\a=\{P_\mu\}_{\mu=1}^{\binom{2n}{\a}}$ an orthonormal basis of Pauli operators of the irreps $\LC_\a$, namely those (normalized) Pauli strings consisting of the products of $\a$ distinct Majorana operators. 
The Haar averaged GFD purities for the fermionic Gaussianity QRT thus read
\begin{equation}
\mathbb{E}_{\ket{\phi}\sim\HC_+}[ \PC_\a(\dya{\phi})]=\sum_{P\in\BC_a}\Tr[ \mathbb{E}_{\ket{\phi}\sim\HC_+}[\dya{\phi}^{\otimes 2}]P^{\otimes 2}]\,.
\end{equation}
To perform this computation, we will explicitly expand $\dya{\phi}^{\otimes 2}$, in terms of $\dya{\psi}^{\otimes 2}$, and shift the expectation value from $\mathbb{E}_{\ket{\phi}\sim\HC_+}$ to $\mathbb{E}_{\ket{\psi}\sim\HC}$, allowing us to leverage the tools of Weingarten calculus. 
First, we find 
$\dya{\phi} = \frac{1}{c^2} (\dya{\psi} + Z^{\otimes n} \dya{\psi} + \dya{\psi} Z^{\otimes n} + Z^{\otimes n} \dya{\psi} Z^{\otimes n})$, which implies 
\small
\begin{align}
\dya{\phi} ^{\otimes 2} =& \frac{1}{c^4}(\dya{\psi}^{\otimes 2} + (\id\otimes Z^{\otimes n}) \dya{\psi}^{\otimes 2}  + \dya{\psi}^{\otimes 2} (\id\otimes Z^{\otimes n})+  (\id\otimes Z^{\otimes n}) \dya{\psi}^{\otimes 2} (\id\otimes Z^{\otimes n}) + (Z^{\otimes n}\otimes\id)\dya{\psi}^{\otimes 2}  \nonumber\\
&    + (\Z^{\otimes n}\otimes Z^{\otimes n})\dya{\psi}^{\otimes 2}   + (\Z^{\otimes n}\otimes \id)\dya{\psi}^{\otimes 2}(\id\otimes\Z^{\otimes n})
+ (\Z^{\otimes n}\otimes \Z^{\otimes n})\dya{\psi}^{\otimes 2}(\id\otimes\Z^{\otimes n})  \nonumber\\
& + \dya{\psi}^{\otimes 2} (Z^{\otimes n} \otimes \id) + (\id\otimes Z^{\otimes n})\dya{\psi}^{\otimes 2} (Z^{\otimes n}\otimes\id) +\dya{\psi}^{\otimes 2}(Z^{\otimes n}\otimes Z^{\otimes n})  + (\id\otimes Z^{\otimes n})\dya{\psi}^{\otimes 2}(Z^{\otimes n}\otimes Z^{\otimes n})  \nonumber \\
&     + (Z^{\otimes n}\otimes\id)\dya{\psi}^{\otimes 2}(Z^{\otimes n}\otimes\id) +(Z^{\otimes n}\otimes Z^{\otimes n})\dya{\psi}^{\otimes 2}(Z^{\otimes n}\otimes\id) + (Z^{\otimes n}\otimes\id)\dya{\psi}^{\otimes 2}(Z^{\otimes n}\otimes Z^{\otimes n}) \nonumber\\
&   + (Z^{\otimes n}\otimes Z^{\otimes n})\dya{\psi}^{\otimes 2} (Z^{\otimes n}\otimes Z^{\otimes n}) \nonumber\,.
\end{align}
\normalsize
Then, we can replace $\mathbb{E}_{\ket{\psi}\sim\HC}[\dya{\psi}^{\otimes 2}]$ by the result in Eq.~\eqref{eq:Haar}, and evaluate the contribution arising from  $\id \otimes \id$ and from ${\rm SWAP}$ separately. 
First let us notice that in the case of the trivial irrep $\LC_0$, the associated GFD purity $\mathbb{E}_{\ket{\phi}\sim\HC_+}[ \PC_0(\dya{\phi})]=\frac{1}{2^n}$ since we are considering pure states. Hence, in the following, we will only consider $\a>0$.
Then, with the help of some bookkeeping, one finds that the terms arising from the $\id \otimes \id$ component of $\mathbb{E}_{\ket{\psi}\sim\HC}[\dya{\psi}^{\otimes 2}]$ lead to the following contribution to the GFD purities
\begin{align}
\frac{4}{c^4d(d+1)}(\Tr[P]^2 + \Tr[P\otimes P Z^{\otimes n}] +  \Tr[P Z^{\otimes n}\otimes P] +  \Tr[P Z^{\otimes n} \otimes P Z^{\otimes n}])\,,
\end{align}
where the only possibly contributing term is $\Tr[P Z^{\otimes n}\otimes P Z^{\otimes n}]$ when $P$ is equal to Z$^{\otimes n}$. This is the only element in the one-dimensional irrep $\LC_{2n}$, so this expectation value gives $\frac{4}{c^4d(d+1)}\delta_{\a, 2n}$. 
Next, the terms arising from ${\rm SWAP}$ lead to
\begin{align}
  & \frac{4}{c^4d(d+1)}(\Tr[P^2] + \Tr[\text{SWAP} \cdot P \otimes P Z^{\otimes n}]  + 
 \Tr[\text{SWAP} \cdot P Z^{\otimes n} \otimes P ] + \Tr[\text{SWAP} \cdot P Z^{\otimes n}\otimes P Z^{\otimes n}])  \nonumber\\
&= \frac{4}{c^4d(d+1)}( \Tr[P ^2] + \Tr[P \cdot P Z^{\otimes n}] 
 + \Tr[\rm P Z^{\otimes n} \cdot P ] +  \Tr[P Z^{\otimes n} \cdot P Z^{\otimes n}])\,. \nonumber
\end{align}
Here the only contributing terms are $\Tr[P^2]$, which is equal to one by the normalization of $\BC_\a$, and $\Tr[P Z^{\otimes n} P Z^{\otimes n}]$. The latter is equal to one if $P$ commutes with $Z^{\otimes n}$ and to minus one if it anticommutes with it. Since $Z^{\otimes n}$ is the product of all $2n$ Majorana operators, and for any given $\LC_\a$, $P$ is the product of $\a$ of them, we readily have that $P$ commutes with $Z^{\otimes n}$ when $\a$ is even and anticommutes otherwise. In the latter case the two contributions from the ${\rm SWAP}$ operator cancel out, coherently with the fact that we are sampling $\ket{\phi}$ from $\HC_+$, whereas for even $\a$ we get a contribution $\frac{8}{c^4d(d+1)}$.
Combining the previous results, summing over the $\binom{2n}{\a}$ Paulis in each $\BC_a$, and substituting the value of $c$, we obtain the result presented in the Methods
\begin{equation}
    \mathbb{E}_{\ket{\phi}\sim\HC_+}[ \PC_\a(\dya{\phi})]\approx\begin{cases}
        0 & \text{if $\a$ is odd}\\
        \frac{1}{2^n} & \text{if $\a=0$}\\
        \frac{2\binom{2n}{\a}}{2^n(2^n+1)}  & \text{if $\a$ is even and $0<\a<2n$} \\
        \frac{3}{2^n(2^n+1)}  & \text{if $\a=2n$}
    \end{cases}\,.
\end{equation}

\section{GFD purities for spin coherence}

Let $G=\mathbb{SU}(2)$ and consider a spin-$s$ irreducible module $\HC=\spn\{ \ket{s,m}\}_{m= -s}^{s}$. Then, the operator space decomposes as  
\begin{equation}
\LC(\HC)= \bigoplus_{s'=0}^{2s} \LC_{s'}\,,
\end{equation}
where each irrep  $\LC_{s'}={\rm span}_{\mathbb C}\{ L_{s',m'}\}_{m'=- s'}^{s'}$ is spanned by the operators
\begin{equation}\label{eq:su2_irrep_basis}
L_{s',m'}= \sum_{m =-s}^s (-1)^{s-(m-m')} c^{s,s,s'}_{m,m'-m,m'} \ketbra{m}{m-m'} \,,
\end{equation}
where $c_{m_1,m_2,m}^{s_1,s_2,s}$ is a Clebsch-Gordan coefficient.

For a state of the form $\ket{s,m}$ with associated density matrix $\rho_{m}$, where we drop the $s$ index for convenience, we can readily find
\begin{align}
    \PC_{\alpha=s'}(\rho_m)=&\sum_{m'=-s'}^{s'} \langle s,m|L_{s',m'}|s,m\rangle^2=\left( c_{m,-m,0}^{s,s,\alpha}\right)^2\,.
\end{align}
In turn, using the identity $c_{s,-s,0}^{s,s,\alpha} = (2 s)! \sqrt{\frac{2 \a+1}{(2 s-\a)! \Gamma (2 s+\a+2)}}$ one recovers the GFD purities of the free states, i.e. the orbit of $\ket{s,s}$.

Analogously, when considering the GHZ state $\ket{\rm GHZ}=\frac{1}{\sqrt{2}}(\ket{s,s}+\ket{s,-s})$, using the previous results for $\rho_m$ allows us to show that if $\alpha\neq 2s$  we get
\begin{align}\label{eq:pur-ghz-su2-si}
    \PC_{\alpha}(\rho_{{\rm GHZ}})&=\frac{1}{4}\left( c_{m,-m,0}^{s,s,\alpha}+(-1)^{2 s}  c_{m,-m,0}^{s,s,\alpha}\right)^2=\begin{cases}\left( c_{m,-m,0}^{s,s,\alpha}\right)^2\,,  & \text{ if $\alpha$ is even}\\
    0\,, & \text{if $\alpha$ is odd}
    \end{cases}\,,
\end{align}
while for $\alpha=2s$ 
\begin{align}
    \PC_{2s}(\rho_{{\rm GHZ}})=&\frac{1}{4}\left( c_{s,-s,0}^{s,s,2s}+(-1)^{2 s}  c_{-s,s,0}^{s,s,2s}\right)^2+\frac{1}{4} (c_{-s,-s,-2s}^{s,s,2s})^2 + (c_{s,s,2s}^{s,s,2s})^2\,.
\end{align}

Lastly, the GFD purities presented in the Methods section for a Haar random quantum state $\rho_H=\dya{\psi}$, with $\ket{\psi}$ sampled uniformly according to the Haar measure on $\HC$, can be recovered directly from Eq.~\eqref{ap-eq:haar_random_Lie_gfd}.

\section{GFD purities for Clifford stabilizerness}

We refer the reader to Ref.~\cite{leone2022stabilizer} for a derivation of the stabilizer entropies for a stabilizer state, a Haar random state and for the magic state $\ket{M}^{\otimes n}$. Then, given that the stabilizer entropy is linearly proportional to the irrep purity in $\LC_1$, we obtain the desired results.

\end{document}